\documentclass[american,letterpaper]{article}
\usepackage{fullpage}
\usepackage[utf8]{inputenc}
\usepackage{amsmath,amssymb,amsthm,amsfonts,verbatim,paralist}
\usepackage{caption}
\usepackage{xr}
\makeatletter
\usepackage{hyperref}
\usepackage{subcaption}
\usepackage{multirow}
\usepackage{fancyhdr,graphicx,color}
\usepackage{graphicx}
\usepackage{footmisc}
\usepackage{graphics}
\RequirePackage[numbers]{natbib}
\usepackage[ruled,vlined]{algorithm2e}
\usepackage{enumitem}
\usepackage{colortbl}
\newtheorem{rmk}{Remark}

\newtheorem*{asu*}{Assumption}

\theoremstyle{plain}
\newtheorem{lem}{Lemma}
\newtheorem{thm}{Theorem}
\newtheorem{pro}{Proposition}
\newtheorem{cro}{Corollary}
\newtheorem{exa}{Example}
\newcommand{\argmin}{\mathop{\mathrm{argmin}}}

\newcommand{\tr}{\mathsf{Tr}}
\newcommand{\R}{\mathbb{R}}

\newcommand{\p}{{\rm I}\kern-0.18em{\rm P}}
\newcommand{\E}{{\rm I}\kern-0.18em{\rm E}}
\newcommand{\CS}{\mathcal{S}}
\newcommand{\cN}{\mathcal{N}}
\newcommand{\CR}{\mathcal{R}}
\newcommand{\CG}{\mathcal{G}}
\newcommand{\schur}[3]{{#1}/ [{#2},{#3}]}
\usepackage[affil-it]{authblk}
\newcommand{\1}{{\rm 1}\kern-0.24em{\rm I}}
\title{Sparse PCA: A New Scalable Estimator Based On Integer Programming}
\author[1]{Kayhan Behdin\thanks{behdink@mit.edu}}
\author[1,2]{Rahul Mazumder\thanks{rahulmaz@mit.edu}}

\affil[1]{Operations Research Center, Massachusetts Institute of Technology}
\affil[2]{Sloan School of Management, Massachusetts Institute of Technology}

\date{}

\usepackage{setspace}
\begin{document}
\maketitle
\begin{abstract}
We consider the Sparse Principal Component Analysis (SPCA) problem under the well-known spiked covariance model. Recent work has shown that the SPCA problem can be reformulated as a Mixed Integer Program (MIP) and can be solved to global optimality, leading to estimators that are known to enjoy optimal statistical properties. However, prior MIP algorithms for SPCA appear to be limited in terms of scalability to up to a thousand features or so. In this paper, we propose a new estimator for SPCA which can be formulated as a MIP. Different from earlier work, we make use of the underlying spiked covariance model and properties of the multivariate Gaussian distribution to arrive at our estimator. We establish statistical guarantees for our proposed estimator in terms of estimation error and support recovery. We derive guarantees under departures from the spiked covariance model, and for approximate solutions to the optimization problem. We propose a custom algorithm to solve the MIP, which scales better than off-the-shelf solvers, and demonstrate that our approach can be much more computationally attractive compared to earlier exact MIP-based approaches for the SPCA problem. Our numerical experiments on synthetic and real datasets show that our algorithms can address problems with up to $20,000$ features in minutes; and generally result in favorable statistical properties compared to existing popular approaches for SPCA.
\end{abstract}

\section{Introduction}
Principal Component Analysis (PCA)~\citep{hotelling1933analysis} is a well-known dimensionality reduction method where one seeks to find a direction, a Principal Component (PC), that describes the variance of a collection of data points as well as possible. In this paper, we consider the spiked covariance model~\citep{johnstone2001}. Suppose we are given $n$ samples $X_i\in\R^p$ for $i\in[n]$, drawn independently from a multivariate normal distribution $\mathcal{N}(0,G^*)$ (where $G^*\in\R^{p\times p}$ is the covariance matrix). 
We let $X\in\R^{n\times p}$ denote the data matrix with $i$-th row $X_{i}$.
In the spiked covariance model, we assume 
\begin{equation}\label{Gcov1}
    X_{1}, \ldots, X_{n} \stackrel{\text{iid}}{\sim} {\mathcal N}(0,G^*) ~~~\text{and}~~~G^*=I_p+\theta u^*({u^*})^T
\end{equation}
where $u^*\in\R^p$ has unit $\ell_{2}$-norm $\|u^*\|_2=1$ and $I_p$ is the identity matrix of size $p$. In model~\eqref{Gcov1}, the vector $u^*$ is referred to as the spike.
We let $\theta$ be the Signal to Noise Ratio (SNR) of the model~\eqref{Gcov1}.
Informally speaking, all other factors remaining same, a higher value of the SNR makes the task of estimating $u^*$ relatively easier. Following~\citep{bresler-10.5555/3327546.3327752}, here we consider the setting where $\theta \in (0, 1]$.
The spiked covariance model is a standard statistical model in the PCA context and has been studied extensively in the literature, for example see~\cite{johnstone2001,bresler-10.5555/3327546.3327752,deshpande2016sparse,amini2009,krauthgamer2015semidefinite}.

Given $X$, we seek to estimate the vector $u^*$ (up to a sign flip) which corresponds to the true PC. 
If $\hat{u}$ denotes an estimator of $u^*$, then we can measure the quality of this estimator via the sine of the angle between $\hat{u}$ and $u^*$ (a smaller number denoting higher accuracy). Mathematically, for $\hat{u}$ such that $\|\hat{u}\|_2=1$, we define:
\begin{equation}\label{sinedef}
     |\sin\angle(\hat{u},u^*)|=\sqrt{1-(\hat{u}^T{u^*})^2}.
\end{equation}

In the high dimensional regime when the number of features ($p$) is much larger than the number of samples i.e., $p \gg n$, \cite{johnstone2009consistency,10.2307/24307692} show that the simple PCA estimator may be inconsistent for $u^*$ unless one imposes additional assumptions on $u^*$. A popular choice is to assume that $u^*$ is sparse with $s$ nonzeros (i.e, $s$-sparse)---in this case, we find a PC with $s$ nonzeros (for example)---this is commonly known as the sparse PCA problem or SPCA for short~\citep{jolliffe2003modified,hastie2019statistical}. \cite{johnstone2009consistency} propose the diagonal thresholding algorithm for SPCA that results in consistent estimation of the support of $u^*$ if $n\gtrsim s^2\log p$ \citep{amini2009}\footnote{The notations $\lesssim,\gtrsim$ are used to show an inequality holds up to a universal constant that does not depend upon problem data.}, a notable improvement over vanilla PCA that requires $n$ to be larger than $p$. 
In this paper, we study the SPCA problem under the spiked covariance model where $u^*$ is $s$-sparse. We also study a generalization where the covariance model is a perturbation of a spiked covariance matrix.

In the remainder of this section, we first present an overview of current algorithms for SPCA and then summarize our key contributions in this paper.

\subsection{Background and Literature Review}\label{sec:lit-review}
A well-known procedure to estimate $u^*$ is to consider the following optimization problem:
\begin{equation}\label{spca-QPform1}
    \max_{u \in \mathbb{R}^p}~~ u^T X^TX u ~~~~ \text{s.t.} ~~~~\| u \|_0 \leq s, \| u \|_{2} = 1
\end{equation}
which finds a direction $u$ with at most $s$ nonzero entries that maximizes the variance of the data along that direction. \cite{amini2009} have shown that the optimal solution to this problem attains the minimax optimal rate in terms of variable selection. In fact, \cite{amini2009} show that no algorithm can recover the PC consistently unless $n\gtrsim s\log p$ (information theory lower bound). Problem~\eqref{spca-QPform1} involves the maximization of a convex quadratic over a nonconvex set, and poses computational challenges. Several algorithms~\citep[see for example]{hastie2019statistical} have been proposed to obtain good solutions to~\eqref{spca-QPform1}---this includes both exact and approximate methods. In what follows, we review some existing algorithms and divide them into three broad categories.

\smallskip

\noindent \textbf{Polynomial-time Algorithms:} 
A well-known procedure to obtain convex relaxations for Problem~\eqref{spca-QPform1} is based on Semidefinite Programming~(SDP)~\citep{d2005direct,d2008optimal}. 
SDP-based relaxations of SPCA can recover the support of PC if the SDP returns a rank one solution and when $n\gtrsim s\log p$ \citep{amini2009}. Interestingly, \cite{krauthgamer2015semidefinite} show that if $n/\log p\gtrsim s\gtrsim \sqrt{n}$, the SDP formulation may not result in a rank one solution. 

In an interesting line of work,~\citep{berthet2013,10.1214/15-AOS1369} provide theoretical evidence that under the planted clique hypothesis, it may not be possible to have a polynomial-time algorithm for SPCA if $s\gtrsim \sqrt{n}$.
\cite{deshpande2016sparse} show that a polynomial-time method called covariance thresholding can estimate $u^*$ if $n\gtrsim s^2\log p$. Moreover, when $p\leq cn$ for some constant $c>0$, \cite{deshpande2016sparse} show $n\gtrsim s^2$ samples suffice.
The recent work of \cite{bresler-10.5555/3327546.3327752} shows that SPCA solutions can be obtained by solving a sequence of sparse linear regression problems. 
One of their approaches requires solving $p$ separate Lasso-regression problems.

A key difference between the SDP-relaxation approach and the covariance thresholding method lies in their computational efficiencies: the covariance thresholding method can be scaled to instances with $p \approx 10^4$ or so, and is computationally (much more) attractive compared to SDPs which are mostly limited to instances with $p$ in the hundreds (using off-the-shelf convex optimization solvers).

\smallskip

\noindent\textbf{Exact and MIP-based Algorithms:} 
In another line of work, exact approaches i.e., methods that lead to a globally optimal solution to~\eqref{spca-QPform1} have been proposed. For example, authors in~\citep{moghaddam2006spectral} present a tailored branch-and-bound method to solve SPCA when $p\leq 40$. Fairly recently, starting with the work of~\citep{bertsimas2016best} in sparse regression, there has been significant interest in exploring Mixed Integer Programming (MIP) approaches 
in solving sparse learning problems that admit a combinatorial description. In particular, MIP-based approaches have been proposed to address Problem~\eqref{spca-QPform1} for moderate/small-scale instances, as we discuss below. 
MIP-based approaches seek to obtain the global optimum of 
Problem~\eqref{spca-QPform1}---they deliver a good feasible solution and corresponding upper bounds (aka dual bounds)---taken together they certify the quality of the solution. 
\cite{berk2019certifiably} present a tailored branch-and-bound method that can obtain near-optimal solutions to Problem~\eqref{spca-QPform1} for $p\leq 250$. In a different line of work,~\cite{dey2018convex} present MIP-formulations and algorithms for~\eqref{spca-QPform1} and a relaxation that replaces the $\ell_0$-constraint on $u$, by an $\ell_{1}$-constraint resulting in a non-convex optimization problem. This approach delivers
upper bounds i.e., dual bounds to~\eqref{spca-QPform1} for up to $p\approx 2000$.
Other recent MIP-based approaches include~\citep{bertsimas2020solving} who consider mixed integer semidefinite programming (MISDP) approaches and associated relaxations of computationally expensive semidefinite constraints.  
In another line of work, \cite{li2020exact} discuss 
MISDP formulations and approximate mixed integer linear programming approaches 
for SPCA. Some of the approaches presented in~\cite{li2020exact,bertsimas2020solving} succeed in obtaining dual bounds for SPCA for $p\approx 1000$ or so.

In this paper, we also pursue a MIP-based approach to the SPCA problem. However, instead of Problem~\eqref{spca-QPform1} which in our experience appears to be challenging to scale to large instances, we consider a different criterion based on the statistical model~\eqref{Gcov1}. 
Our estimator is given by a 
structured sparse least squares problem---our custom MIP-based approach can solve this problem to optimality for instances up to $p\approx 20,000$, which is 10X larger than earlier 
MIP approaches presented in~\cite{dey2018convex,li2020exact}.

\smallskip

\noindent \textbf{Heuristic Algorithms:} Several appealing heuristic algorithms have been proposed to obtain good feasible solutions for Problem~\eqref{spca-QPform1} and close variants (e.g., involving an $\ell_{1}$-sparsity constraint instead of the $\ell_0$-constraint)---see for example,~\citep{zou2006sparse,richtarik2020alternating,teboulleconditional,witten2009penalized}. 
A popular approach is the Truncated Power Method \citep{yuan2013truncated}, which enjoys good statistical guarantees and numerical performance if a \emph{good} initialization is available. However, in the absence of a good initialization, this method can converge to a suboptimal local solution which may be even orthogonal to the underlying PC---this is illustrated by our experiments in Section~\ref{numerical} and also discussed in \citep{https://doi.org/10.1111/rssb.12360}. Unlike MIP-based approaches discussed above, heuristic methods do not deliver dual bounds---it may not be possible to certify the quality of the solution. In practice, some of these heuristic algorithms can lead to good feasible solutions for~\eqref{spca-QPform1} for $p \approx 10^4$ or so.

\subsection{Outline of Approach and Contributions}
Different algorithms discussed above 
have their respective strengths and limitations. While there is an impressive body of work on heuristics and polynomial-time algorithms for the SPCA problem, there has been limited investigation on MIP-based approaches, despite promising recent results. In this paper, our focus is to consider a MIP-based approach for the SPCA problem. However, in departure from current approaches based on MIP, which focus on~\eqref{spca-QPform1}, we consider a different estimation criterion. We make use of the spiked covariance model~\eqref{Gcov1} and properties of a multivariate Gaussian distribution to arrive at a mixed integer second order cone program (MISOCP)
that can be interpreted as a least squares problem with structured sparsity constraints, where the constraints arise from the spiked covariance model structure. To our knowledge, the estimator we propose and study herein has not been considered earlier in the context of the SPCA problem; and we study both statistical and algorithmic aspects of the proposed estimator. 

By extending recent work on specialized algorithms for sparse regression-type problems~\cite[see for example]{hazimeh2020fast,hazimeh2020sparse,bertsimas2020sparse}, we propose a custom algorithm that can solve our proposed SPCA MIP formulation for instances with $p\approx 20,000$ (and $n\approx 15,000,s=10$) in minutes---the corresponding estimation accuracy of these solutions appear to be numerically better than several polynomial-time and heuristic algorithms that we compare with.

In terms of statistical properties, we establish an estimation error bound as $ |\sin\angle(\hat{u},u^*)|\lesssim \sqrt{s^2\log p/n}$ where $\hat{u}$ is the estimated PC by our method, when $u^*$ is $s$-sparse. 
Importantly, we also establish a novel oracle-type inequality where the underlying covariance model can be different from a (sparse) spiked covariance model. This shows that we continue to obtain an estimation error rate of 
$s^2\log p/n$ even if the underlying covariance model deviates from a spiked covariance model (as long as the deviation is not too much --- this will be made precise later).   
Our analysis is different from the work of \cite{johnstone2009consistency} who assume a certain decay for the coordinates of $u^*$. In terms of variable selection, under certain regularity conditions, our method recovers the full support of $u^*$ correctly with high probability if $n\gtrsim s\log p$. Moreover, for our estimator, the minimum signal strength required for full support recovery is weaker compared to the existing work such as~\citep{deshpande2016sparse,bresler-10.5555/3327546.3327752,amini2009}. In particular, we show an example where our estimator can provably recover the correct support while existing guarantees for common SPCA estimators fall short. 
We also derive error bounds and support recovery guarantees for an approximate solution available from our optimization framework\footnote{Our MIP-based algorithms deliver optimality certificates (available from dual bounds) along the course of the algorithm. These optimality certificates appear in the statistical error bounds of approximate solutions available from our optimization criterion.} which can be useful in practice, especially if a practitioner decides to use an approximate solution available from our MIP framework under computational constraints.

 To our knowledge, our work is novel in that it makes use of statistical modeling assumptions to obtain a MIP-based optimization formulation for the SPCA problem, which is computationally \emph{more attractive} compared to the original formulation~\eqref{spca-QPform1}. While our proposed estimator is not equivalent to~\eqref{spca-QPform1}; we establish that the statistical properties of our estimator are similar to that of an optimal solution to~\eqref{spca-QPform1} under suitable conditions. 
We summarize our contributions below:
\begin{enumerate}
    \item We present a novel MISOCP formulation for the SPCA problem under the assumption of the spiked covariance model~\eqref{Gcov1}. To the best of our knowledge, our work is the first to reformulate SPCA into a computation-friendly MIP via statistical modeling reductions.
    \item We show that our proposed estimator has an estimation error rate of ${(s^2/n)\log (p/s)}$---this rate is achievable as long as the underlying covariance model can be approximated with a sparse spiked covariance model, even if the underlying covariance does not follow such a model. We also show that if $n\gtrsim s\log p$, under certain regularity conditions, our method recovers the support of $u^*$ correctly.
    \item We propose a custom cutting plane algorithm to solve the MISOCP describing our SPCA estimator---our approach can address instances of the SPCA problem with $p\approx 20,000$, and leads to
    significant computational improvements compared to off-the-shelf MIP solvers for the same problem.  
    \item We demonstrate empirically that our framework can lead to SPCA estimators with favorable statistical properties compared to several popular SPCA algorithms.
\end{enumerate}
\textbf{Notation.} For a vector $x\in\R^p$, let $x_i$ denote the $i$-th coordinate of $x$. For a positive integer $a$, we let $[a]:=\{1, \ldots, a\}$.
For a matrix $X\in\mathbb{R}^{n\times p}$, let 
$X_{i,j}$, $X_{i,:}$ and $X_{:,i}$ denote the $(i,j)$-th coordinate, the $i$-th row and $i$-th column of $X$, respectively. For a matrix $G\in\R^{p_1\times p_2}$ and sets $A_1\subseteq[p_1],A_2\subseteq[p_2]$, we let $G_{A_1,A_2}$ denote the sub-matrix of $G$ with rows in $A_1$ and columns in $A_2$. 
In particular, if $A_2=\{j\}$, we let $G_{A_1,j}$ be the vector $\{G_{i,j}\}_{i \in A_{1}}$.
For the data matrix $X\in \R^{n\times p}$ and a set $A\subseteq[p]$, we let $X_A$ denote $X_{[n],A}$. If $x\in\R^p$ and $A\subseteq[p]$, we let $x_A$ be a sub-vector of $x$ with coordinates in $A$. For a vector $x\in\mathbb{R}^p$, we let $S(x)$ denote its support,
$S(x)=\{i\in[p]:|x_i|>0\}.$
We let $\mathcal{B}(p)$ denote the unit Euclidean unit ball in $\mathbb{R}^p$ i.e., 
$\mathcal{B}(p)=\{x\in\R^p: \|x\|_2\leq 1\}.$
For a vector $u\in\R^p$, we let $\|u\|_0$ be the $\ell_0$-pseudonorm of $u$, denoting the number of nonzero coordinates of $u$. We denote the smallest and largest eigenvalues of a positive semidefinite matrix $G\in\R^{p\times p}$, by $\lambda_{\min}(G),\lambda_{\max}(G)$, respectively. Similarly, $\sigma_{\min}(X)$, $\sigma_{\max}(X)$ denote the smallest and largest singular values of matrix $X$, respectively. For a convex function $f:\R^p\to \R$, $\partial f(x)$ denotes the set of subgradients of $f$ at $x$. The notations $\lesssim,\gtrsim$ denote that inequality holds up to a universal constant (i.e., one with no dependence on problem data). Throughout the paper, we use the convention ${0}/{0}=0$. The proofs of main results can be found in the appendix.

\section{Proposed Estimator}\label{problemform}

We present our estimator for the SPCA problem.
We motivate our estimator using the spiked covariance model~\eqref{Gcov1}---an useful and insightful model that is often used to study the SPCA problem. We draw inspiration from earlier papers on SPCA~\citep{johnstone2001,bresler-10.5555/3327546.3327752,deshpande2016sparse,amini2009,krauthgamer2015semidefinite} who study the spiked covariance model, and its variants. Although the derivation of our estimator uses the spiked covariance model assumption, in Section~\ref{statreg} we see that our estimator appears to perform well in terms of statistical properties as long as the underlying model can be well-approximated by a sparse spiked covariance model. 
For the rest of this section, to present the general idea and intuition behind our method, we assume model~\eqref{Gcov1} where  $u^*$ is $s$-sparse. 
We first recall a well-known result pertaining to multivariate Gaussian distributions.
\begin{lem}\label{lem_reg}
Let $x\in \mathbb{R}^p$ be a random vector from $\mathcal{N}(0,G)$ where
$G$ is a positive definite matrix. For any $j \in [p]$, the conditional distribution of $x_{j}$, given $x_{i}, i \neq j$ is given by
\begin{equation}\label{lem1-conditional}
    x_j|\{x_i\}_{i\neq j}\sim \cN(\sum_{i\neq j} \beta_{i,j}^* x_i, ({\sigma_j^*})^2)
\end{equation}
where, $\beta^*_{i,j}=-\frac{(G^{-1})_{j,i}}{(G^{-1})_{j,j}}$ and $({\sigma_j^*})^2=\frac{1}{(G^{-1})_{j,j}}$.
\end{lem}

The following lemma writes $\beta^*$ in terms of the parameters of the spiked covariance model~\eqref{Gcov1}.
\begin{lem}\label{Ginv}
The inverse of the covariance matrix $G^*$ in \eqref{Gcov1}, can be written as:
\begin{equation}
    {(G^*)}^{-1}=I_p-\frac{\theta u^*({u^*})^T}{1+\theta}.
\end{equation}
The parameters $\{\beta_{i,j}^*\}_{i,j}$ and $\{\sigma_i^*\}_{i}$ defined in Lemma~\ref{lem_reg} can be expressed as:
\begin{equation}\label{beta_lem}
\begin{aligned}
 \beta_{i,j}^* = \frac{\theta u^*_i u^*_j}{1+\theta -\theta ({u^*_j})^2},~i \neq j~~~~\text{and}~~~~~
    ({\sigma_j^*})^2 = 1 + \frac{\theta ({u^*_j})^2}{1 + \theta - \theta({u^*_j})^2}, ~j \in [p].
\end{aligned}
\end{equation}
\end{lem}
For notational convenience, we let the diagonal entries of the matrix $\beta$ to be zero, that is, $\beta_{i,i}^*=0$ for $i\in[p]$.
Lemma \ref{Ginv} shows that the sparsity pattern of the off-diagonal entries of
$\beta^*$ is the same as that of the off-diagonal entries of 
$G^*-I_p$. In particular, $\beta^*_{i,j}$ can be nonzero only if both the $i$-th and $j$-th coordinates of $u^*$ are nonzero i.e., $u_i^* \neq 0$,$u_j^*\neq 0$. 
As we assume $u^*$ has at most $s$ nonzeros, at most $s^2$ values of $\{\beta^*_{i,j}\}_{i,j}$ can be nonzero. 
By Lemma~\ref{lem_reg}, for every $j \in [p]$, we have
$$X_{:,j}|\left\{X_{:,i}\right\}_{i\neq j}\sim\mathcal{N}\left(\sum_{i\neq j}\beta^*_{i,j}X_{:,i},(\sigma_j^*)^2 I_n\right).$$
For a fixed $j$, we can consider estimating the parameters $\{\beta^*_{i,j}\}_{j \neq i}$ by minimizing a least squares objective 
$\| X_{:,j} - \sum_{i\neq j}\beta_{i,j}X_{:,i}  \|_2^2$
under suitable (sparsity) constraints on $\{\beta_{i,j}\}_{i \neq j}$ imposed by the sparsity structure on $u^*$. As $j$ runs from $\{1, \ldots, p\}$, we have $p$-many least squares problems subject to (structured) sparsity constraints on $\beta$.
This motivates us to consider the following optimization problem which minimizes the sum of square of the residual errors, across all $p$ regression problems: 
     \begin{subequations}\label{reg_multi}
     \begin{align}
     \min_{\beta,{z}} \quad &  \sum_{j=1}^p \left\Vert X_{:,j} -\sum_{i\neq j}\beta_{i,j}X_{:,i} \right\Vert_2^2  \\
     \text{s.t.} \quad & {z}\in \{0,1\}^p,\beta\in\R^{p\times p},\beta_{i,i}=0~\forall i\in[p]  \\
     & \beta_{i,j}(1-z_i)=\beta_{i,j}(1-z_j) = 0\quad \forall i,j\in[p]\label{secondconst}\\
     & \sum_i {z}_i \leq s \label{thirdconst}
     \end{align}
     \end{subequations}
where above, the decision variables are 
$\beta\in\R^{p\times p}, z\in\{0,1\}^p$.
The binary variable $z \in \{0, 1\}^p$ controls the sparsity pattern of $\{\beta_{i,j}\}$. 
The complementarity constraint \eqref{secondconst} states that $\beta_{i,j} \neq 0$  only if $z_{i}=z_{j}=1$; and $\beta_{i,j}=0$ otherwise. The constraint \eqref{thirdconst} ensures that $z$ has at most $s$-many nonzeros (under the assumption that $u^*$ has at most $s$ nonzeros). Problem~\eqref{reg_multi} is a mixed integer quadratic optimization problem with complementarity constraints. 
We present different integer programming formulations for~\eqref{reg_multi} in 
Section~\ref{optalg}.
\begin{rmk}
It follows from~\eqref{beta_lem} that the off-diagonal entries of $\beta^*$ are the entries of an asymmetric matrix with rank one.  
Due to computational reasons, we drop the rank constraint 
and hence, formulation~\eqref{reg_multi} can be interpreted as an approximation (by removing the rank constraint) to formulation~\eqref{spca-QPform1} under the covariance model~\eqref{Gcov1}. As we will see in our statistical error bounds, estimator~\eqref{reg_multi} leads to a sub-optimal error bound compared to estimator~\eqref{spca-QPform1}. However, Problem~\eqref{reg_multi} appears to be friendlier from a computational viewpoint, and via our tailored algorithms, we are able to compute optimal solutions to~\eqref{reg_multi} for instances with $p\approx 20,000$. Recall that in contrast, current MIP-based approaches for SPCA 
can deliver dual bounds for~\eqref{spca-QPform1} for $p \approx 1,000$.  
\end{rmk}
Our estimator~\eqref{reg_multi}, can be interpreted as estimating the inverse of the covariance matrix $G^*$ under a structured sparsity pattern (imposed by $u^*$) as determined by the spiked covariance model~\eqref{Gcov1} (see Lemma~\ref{Ginv}). 
To our knowledge, estimator~\eqref{reg_multi} with the structured sparsity constraint has not been studied earlier in the context of the SPCA problem.
A common procedure to estimate a sparse precision matrix is the nodewise sparse regression framework of \cite{meinhausenbuhlmann}, where the regression problems are solved independently across the $p$ nodes under a plain sparsity assumption on the entries of the precision matrix. (Note that we consider a structured sparsity pattern that is shared across the matrix, which may not be generally 
achievable via a nodewise regression approach).
The framework in~\citep{bresler-10.5555/3327546.3327752} for the SPCA problem is similar to~\citep{meinhausenbuhlmann} in that they also compute  $p$ different regression problems independently of one another. 
Compared to~\citep{bresler-10.5555/3327546.3327752}, our joint estimation framework generally leads to improved statistical error bounds (for estimating $\beta^*$)---see Section~\ref{statreg}, and also leads to better numerical performance (see Section~\ref{numerical}). Note that the $j$-th term of the sum appearing in the objective function of Problem~\eqref{reg_multi}
serves to approximate the negative log-likelihood of the conditional distribution~\eqref{lem1-conditional}. In particular, to estimate the coefficients $\{\beta_{ij}\}$ from~\eqref{reg_multi} we take all the variances in the conditional distributions to the same.
Once we obtain the regression coefficients, we estimate the variances via
update~\eqref{sigmahatdef}.
In what follows, we discuss how to estimate the PC in model~\eqref{Gcov1} from an estimate of $\beta^*$ available from~\eqref{reg_multi}.
Statistical properties of our estimator (both estimation and support selection properties) are discussed in Section \ref{statreg}.

\section{Statistical Properties}\label{statreg}
In this section, we establish some statistical properties of our estimator. 
We assume throughout that data is generated as per the model
\begin{equation}\label{Gcov-general}
    X_{1}, \ldots, X_{n} \stackrel{\text{iid}}{\sim} {\mathcal N}(0,G^*)
\end{equation}
for some positive definite $G^*\in\R^{p\times p}$. Note  
that we allow $G^*$ to depart from a spiked covariance model---our results allow for model misspecification, and the error stemming from approximating $G^*$ with a spiked covariance model appears in the error bound.

\smallskip 

\noindent {\bf Notation:} We first introduce some notation. Given the SNR level $\theta\in(0,1]$ and sparsity level $s>0$, we define a 
set of sparse spiked covariance matrices as: 
\begin{equation}
    \CG_{s,\theta} = \left\{I_p+\theta u u^T: \|u\|_2=1, \|u\|_0\leq s\right\}.
\end{equation}
For a general covariance matrix $G^*$, we can obtain an element in $\CG_{s,\theta}$ that is closest to it as
\begin{equation}
    \tilde{G} \in\argmin_{G\in\CG_{s,\theta}} \|G-G^*\|_F^2.
\end{equation}
Note that as $\tilde{G}\in \CG_{s,\theta}$, there exists $\tilde{u}\in\R^p$ such that $\|\tilde{u}\|_2=1$, $\|\tilde{u}\|_0\leq s$ and $\tilde{G}=I_p+\theta\tilde{u}\tilde{u}^T$. Recall that we denote $\{\beta_{i,j}^*\}$ to be the regression coefficients in Lemma~\ref{lem_reg} that correspond to the covariance matrix $G^*$. Similarly, we define the ``oracle regression coefficients'' $\{\tilde{\beta}_{i,j}\}$ to be the regression coefficients corresponding to 
the covariance matrix $\tilde{G}$. Specifically, in light of Lemma~\ref{Ginv} we define:
\begin{equation}\label{tildebeta}
    \tilde{\beta}_{i,j}=\begin{cases} \frac{\theta \tilde{u}_i\tilde{u}_j}{1+\theta-\theta\tilde{u}_j^2} & \text{~if~} i \neq j\in[p] \\
    0 & \text{otherwise.}
\end{cases}
\end{equation}
Therefore, the term $\sum_{j=1}^p\sum_{i\neq j}|\tilde{\beta}_{i,j}-\beta^*_{i,j}|$ indicates the error in approximating $\beta^*$ by $\tilde{\beta}$---this error stems from approximating $G^*$ by an $s$-sparse spiked covariance matrix
$\tilde{G} \in  \CG_{s,\theta}$. 
In particular, if $G^*\in\CG_{s,\theta}$, then $\sum_{j=1}^p\sum_{i\neq j}|\tilde{\beta}_{i,j}-\beta^*_{i,j}|=0$. Similarly, we denote variances from Lemma~\ref{lem_reg} for $G^*$ and $\tilde{G}$ as $(\sigma_j^*)^2$ and $\tilde{\sigma}_j^2$, respectively. The term $\sum_{j=1}^p (\tilde{\sigma}_j^2-(\sigma_j^*)^2)^2$ captures the error in approximating $\{\sigma^*_j\}$ by $\{\tilde{\sigma}_j\}$.

In this section, our theory will make use of some or all of the assumptions stated below.
\begin{asu*} We have the following:
\begin{enumerate}[leftmargin=*,label=\textbf{(A\arabic*})]
    \item \label{gdtsrassumption} The matrix $G^*$ satisfies $\max_{i,j\in[p]}|G^*_{i,j}|\leq 2$ and $\lambda_{\max}(G^*)\leq 2$, and 
    \begin{equation}
        \min_{\substack{S\subseteq[p] \\ |S|\leq 2s}} \lambda_{\min}(G^*_{S,S}) \geq 1.
    \end{equation}
    Moreover, for $j\in[p]$ we have $\sigma_j^*\leq 2$ where $\sigma^*_j$ is introduced in Lemma~\ref{lem_reg}.
    \item \label{orsclebetaassumtion} The oracle regression coefficients satisfy $\sum_{j=1}^p \sum_{i\neq j}|\tilde{\beta}_{i,j}-\beta^*_{i,j}|\leq 1$. 
    \item \label{sum_assum} The error in approximating $\beta^*$ by $\tilde{\beta}$ is  bounded above as:
\begin{equation}\label{sumassum}
    \sum_{j=1}^p \sum_{i\neq j}|\tilde{\beta}_{i,j}-\beta^*_{i,j}|\lesssim \frac{s\log(p/s)}{n}
\end{equation}
and the true variances $\{(\sigma_{i}^*)^2\}_{i}$ and oracle variances $\{\tilde{\sigma}_{i}^2\}_{i}$ (see Lemma~\ref{Ginv}) satisfy
\begin{equation}\label{sigmaassum}
    \sum_{j=1}^p\left(  \tilde{\sigma}_j^2-({\sigma_{j}^*})^2  \right)^2\lesssim \frac{s^2\log(p/s)}{n}.
\end{equation}
\item \label{beta_min_assum} There is a numerical constant $c_{\min}>0$ such that for $i,j\in[p]$ with $\tilde{u}_i,\tilde{u}_j\neq 0$,
we have $|\tilde{\beta}_{i,j}|>c_{\min} \sqrt{(s/n)\log p}$. We denote this lower bound by $\widetilde{\beta_{\min}}$ that is, 
$\widetilde{\beta_{\min}}=c_{\min} \sqrt{(s/n)\log p}.$
\end{enumerate}
\end{asu*}
Assumption~\ref{gdtsrassumption} ensures the matrix $G^*$ is properly scaled and does not have very large or small eigenvalues. Assumptions~\ref{orsclebetaassumtion} and~\ref{sum_assum} state that the oracle $\tilde{G}$ is a good approximation of $G^*$. In particular, Assumption~\ref{sum_assum} is a special case of Assumption~\ref{orsclebetaassumtion}. Assumption~\ref{beta_min_assum} ensures that the nonzero coordinates of $\tilde{\beta}$ are sufficiently large, to facilitate their identification. 
\begin{rmk}
As we noted earlier, the spiked covariance model 
has been frequently studied in earlier work. 
Observe that if $G^*=I_p+\theta u^*(u^*)^T$ follows a spiked covariance model with $\|u\|_2=1$ and $0<\theta\leq 1$, it readily satisfies Assumption~\ref{gdtsrassumption}. Moreover, if $u^*$ is $s$-sparse, that is $\|u^*\|_0=s$, the matrix $G^*$ also satisfies Assumptions~\ref{orsclebetaassumtion} and~\ref{sum_assum}. Hence, our results that use Assumptions~\ref{gdtsrassumption}-\ref{sum_assum} apply to a sparse spiked covariance model, and a spiked covariance model that can be well-approximated by a sparse one as in Assumptions~\ref{orsclebetaassumtion} and~\ref{sum_assum}.
\end{rmk}

\begin{rmk}\label{new-multispike-remark}
    We note that our model for the covariance matrix is different from the multiple-spike covariance models considered by~\citep{10.2307/24307692,johnstone2009consistency}.
\end{rmk}

Informally speaking, our first result Theorem~\ref{reg_thm} shows that a solution $\hat{\beta}$ to Problem~\eqref{reg_multi} is close to the corresponding population coefficients $\beta^*$ in the squared $\ell_{2}$ norm, as long as the oracle coefficients $\tilde{\beta}$ and $\beta^*$ are sufficiently close.  
Theorem~\ref{reg_thm} (for proof see Appendix~\ref{appendix:c}) presents an error bound of the estimator arising from Problem~(\ref{reg_multi}). 
\begin{thm}\label{reg_thm}
Let $\beta^*$ be the regression coefficients in Lemma~\ref{lem_reg}; and
$\tilde{\beta}$ be as defined in~\eqref{tildebeta}. Suppose Assumptions~\ref{gdtsrassumption} and~\ref{orsclebetaassumtion} hold and $s\leq p/5$. If $(\hat{\beta},\hat{z})$ is a solution to Problem~(\ref{reg_multi}) 
and $n\gtrsim s\log p$, then with high probability\footnote{An explicit expression for the probability can be found in~\eqref{thm1prob}.}, we have:
\begin{equation}\label{bound-rhs-thm1}
   \sum_{\substack{ i,j\in[p]\\ j\neq i }}(\hat{\beta}_{i,j}-\beta_{i,j}^*)^2\lesssim \frac{s^2\log(p/s)}{n} + s\sum_{j=1}^p \sum_{i\neq j}|\tilde{\beta}_{i,j}-\beta^*_{i,j}|.
\end{equation}
\end{thm}

In the special case when $G^*$ follows a sparse spiked covariance model, that is $G^*\in\CG_{s,\theta},$ the second term in the rhs of~\eqref{bound-rhs-thm1} is zero, and Theorem~\ref{reg_thm} shows that Problem~\eqref{reg_multi} delivers an estimate of $\beta^*$ with an error bound of $s^2\log(p/s)/n$.
In what follows, we show how to estimate $u^*$ i.e., the PC of interest, from the estimated regression coefficients $\{\hat{\beta}_{i,j}\}$.

Let us define a matrix $B^*\in \R^{p\times p}$ such that 
\begin{equation}\label{new-bstar-def}
  B^*_{i,j} = \begin{cases} \beta^*_{i,j} &\text{if}~ i \neq j \\
 (\sigma^*_i)^2-1 & \text{if}~ i = j \end{cases}   
\end{equation}
for all $i,j \in [p]$. To gather some insights on the matrix $B^*$, we study some of its properties when $G^*=I_p+\theta u^*(u^*)^T$ follows the spiked covariance model~\eqref{Gcov1}. 
From Lemma \ref{Ginv}, it can be seen that under the spiked model assumption, we have
\begin{equation}\label{bstardef}
    B^*= \frac{\theta}{1+\theta} u^* ({u^*})^T D = \left(\|Du^*\|_2\frac{\theta}{1+\theta}\right)u^*\left(\frac{Du^*}{\|Du^*\|_2}\right)^T
\end{equation}
where $D\in\R^{p\times p}$ is a diagonal matrix such that its $j$-th diagonal entry is given by
$$D_{j,j}=\left(1- \frac{\theta}{1+\theta}{(u^*_j)}^2\right)^{-1}.$$
Observe that $B^*$ is a rank one matrix with nonzero singular value $\|Du^*\|_2\theta/(1+\theta)$ and left singular vector 
$u^*$.
Since
\begin{align*}
    \|Du^*\|_2^2 & = \sum_{i=1}^p \frac{({u_i^*})^2}{(1- \frac{\theta}{1+\theta}{(u^*_j)}^2)^2} \geq \sum_{i=1}^p {(u_i^*)}^2 = 1,
\end{align*}
the nonzero singular value of $B^*$ satisfies
\begin{equation}\label{Bsigma}
    \frac{\theta}{1+\theta}\|Du^*\|_2 \gtrsim \theta.
\end{equation}
To estimate $u^*$, we first obtain an estimate of $B^*$ and consider its leading left singular vector. We show below that if we obtain a good estimate of $B^*$, then its left singular vector will be close to $u^*$. 

In the general case when $G^*$ may not follow a spiked covariance model, the matrix $B^*$ might not have rank one. Our goal in such cases is to estimate the leading left singular vector of $B^*$ corresponding to its largest singular value. For the rest of this section, we use $u^*$ to denote the leading left singular vector of $B^*$.
Our estimation framework described below provides a consistent estimator of $u^*$, showing that we can estimate the PC when $G^*$ follows a spiked covariance model (i.e. when $B^*$ has rank one as discussed above). In Appendix~~\ref{bstar-invest}, we show that when $G^*$ is a suitable perturbation of a spiked covariance model, that is:
$$G^* = I_p+\theta \bar{u}\bar{u}^T + \Delta$$
where $\|\bar{u}\|_2=1$ and $\|\Delta\|_F$ is sufficiently small, the leading left singular vector of $B^*$ is close to the PC from the spiked model, i.e., $\bar{u}$ (for details see Appendix~\ref{bstar-invest}).

Let $\hat{B}\in\R^{p\times p}$ be our estimate of $B^*$. 
A solution $\hat{\beta}$ obtained from Problem~\eqref{reg_multi} can be used to estimate the off-diagonal coordinates of $\hat{B}$, that is, $\hat{B}_{i,j}=\hat{\beta}_{i,j}$ for $i\neq j$. 
To estimate the diagonal entries of $B^*$, since 
${B}^*_{i,i}={(\sigma_i^*)}^2-1$, we need an estimate of $\sigma_i^*$ for each $i$---let us denote this estimate by $\hat{\sigma}_{i}$ for all $i$.
Note that in the special case of the spiked covariance model assumption (see Lemma~\ref{Ginv}), $\sigma_{i}^*=1$ when $u^*_i=0$ --- therefore, we use $\hat{\sigma}_{i}=1$ as an estimate for the indices $i$, where $\hat{z}_{i}=0$.
For the other coordinates, we compute $\hat{\sigma}_i^2$
based on the residual variance of the $i$-th regression problem:
\begin{equation}\label{sigmahatdef}
    \hat{\sigma}_i^2=\begin{cases}
    \frac1n \|X_{:,i}-\sum_{j\neq i}\hat{\beta}_{j,i}X_{:,j}\|_2^2 &\mbox{ if } \hat{z}_i =1\\
    1 & \mbox{ if } \hat{z}_i =0.
    \end{cases}
\end{equation}
To summarize, our estimator $\hat{B}$ is given by:
\begin{equation}\label{defn-hatB}
\hat{B}_{i, j} = \begin{cases}  \hat{\beta}_{i,j}~& ~\text{if}~ i \neq j \in [p] \\
\hat{\sigma}^2_{i} - 1 & ~\text{if}~j=i \in [p]. 
\end{cases}     
\end{equation}

Theorem \ref{Bestconst} (for proof see Appendix~\ref{appendix:c}) presents an error bound in estimating $B^*$ by $\hat{B}$.  
\begin{thm}\label{Bestconst}
Let $\tilde{\beta}$ be as defined in Theorem~\ref{reg_thm} and $B^*$ be defined in \eqref{bstardef}. Suppose Assumptions~\ref{gdtsrassumption} to~\ref{beta_min_assum} hold true. Moreover, assume  $n\gtrsim s^2\log p$. 
Let $\hat{z}$ be as defined in Theorem~\ref{reg_thm} and $\hat{B}$ as defined in~\eqref{defn-hatB}. Then with high probability\footnote{See~\eqref{probexp} for the exact probability expression.\label{probfootnote}} we have 
(i) The supports of $\hat{z}$ and $\tilde{u}$ are the same, that is $S(\hat{z})=S(\tilde{u})$; and (ii) $\|\hat{B}-B^*\|_F^2\lesssim 
(s^2/n)\log(p/s).$
\end{thm} 
In light of Theorem~\ref{Bestconst}, $\hat{B}$ is close to $B^*$. Furthermore, as  $u^*$ is the leading left singular vector of $B^*$, the leading singular vector of $\hat{B}$ should be close to $u^*$ (assuming of course, a separation between the first and second leading singular values of $B^*$)---see for example,~\cite[Chapter~V]{stewart1990matrix} and \cite{wedin1972perturbation} for basic results on the perturbation theory of eigenvectors. This leads to the following result.
\begin{cro}\label{bcro}
Suppose $\hat{u}$ is the leading left singular vector of $\hat{B}$ from Theorem \ref{Bestconst}. Moreover, assume 
$$\sigma_{\max}(B^*)-\sigma_{\max,2}(B^*)\gtrsim \theta$$
where $\sigma_{\max,2}(B^*)$ denotes the second largest singular value of $B^*$. Then, under the assumptions of Theorem \ref{Bestconst}, with high probability\footref{probfootnote} we have that
\begin{equation}
    |\sin\angle(\hat{u},u^*)| \lesssim \frac{1}{\theta}\sqrt{\frac{s^2\log(p/s)}{n}}.
\end{equation}
\end{cro}

\begin{rmk}
Following the discussion preceding~\eqref{Bsigma},
the eigen-gap assumption in Corollary~\ref{bcro} is satisfied if $G^*$ is a spiked covariance model~\eqref{Gcov1}. Hence, Corollary~\ref{bcro} readily applies to a sparse spiked covariance model, or a spiked covariance model that can be well-approximated by a sparse one as in Assumptions~\ref{orsclebetaassumtion} and~\ref{sum_assum}.
\end{rmk}

\subsection{Support recovery under the sparse spiked covariance model}
In this section, we show that under suitable assumptions, Problem~\eqref{reg_multi} recovers the support of $u^*$ correctly  under model~\eqref{Gcov1} when $\|u^*\|_0=s$. In other words, when $\hat{z}$ is an optimal solution to Problem~\eqref{reg_multi}, 
the nonzero coordinates of $\hat{z}$ and $u^*$ are the same.
\begin{thm}\label{supportthm}
Suppose model~\eqref{Gcov1} holds with $\|u^*\|_2=1$ and $\|u^*\|_0=s$; and for some 
universal constant $\eta>0$, either one of~\ref{betamincond} or~\ref{betamincond-extend} holds: 
\begin{enumerate}
[leftmargin=*,label=\textbf{(S\arabic*)}]
    \item \label{betamincond} There exists a value $\beta_{\min}>\sqrt{(2\eta/n)\log p}$ such that for any $i,j\in[p]$ with $u_i^*,u_j^*\neq 0$, we have $|{\beta}^*_{i,j}|>\beta_{\min}$.
    \item\label{betamincond-extend} 
    For all $j\in[p]$ such that $u_j^*\neq 0$, we have that
    \begin{equation*}
    \sum_{i:i\neq j} (\beta_{i,j}^*)^2 >\frac{\eta s\log p}{n}.
\end{equation*}
\end{enumerate}
Then, if $n\geq c s\log p$ for a sufficiently large universal constant $c>0$, we have $S(\hat{z})=S(u^*)$ with high probability\footnote{An explicit expression for the probability can be found in~\eqref{supthmprob}\label{sup_thm_prob-foot}.}, where $\hat{z}$ is as defined in Theorem~\ref{reg_thm}.
\end{thm}
A complete proof of Theorem~\ref{supportthm} can be found in Appendix~\ref{appendix:d}. We present an outline of our key proof strategy.
For $j\in[p]$ and $z\in\{0,1\}^p$, let us define:
\begin{equation*}
    g_j(z) = \min_{\beta_{:,j}} ~\left\Vert X_{:,j}-\sum_{i: i\neq j}\beta_{i,j}X_{:,i}\right\Vert_2^2~~\text{s.t.} ~~\beta_{i,j}(1-z_i)=\beta_{i,j}(1-z_j)=0~\forall i\in[p],\beta_{j,j}=0.
\end{equation*}
With this definition in place, Problem~\eqref{reg_multi} can be equivalently written as
\begin{equation}\label{reg_multi_thm3-main}
    \min_{z} ~\sum_{j=1}^p g_j(z)~~~\text{s.t.}~~~z\in\{0,1\}^p, \sum_{i=1}^p z_i \leq s.
\end{equation}
To analyze Problem \eqref{reg_multi_thm3-main}, for every $j$, we compare the value of $g_j(z)$ for a feasible $z$ to an \emph{oracle} candidate $g_j(z^*)$ where $z^*$ has the same support as $u^*$. Then, we show that 
$$\sum_{j=1}^p g_j(z^*)<\sum_{j=1}^p g_j(z) $$
unless $z=z^*$, which completes the proof.
\begin{rmk}
    We note that Condition~\ref{betamincond}  
    implies Condition~\ref{betamincond-extend}. However, to facilitate comparison of our work to earlier work, we keep these two cases separate. Conceptually,~\ref{betamincond} requires each nonzero $\beta_{i,j}^*$ to be sufficiently far from zero, while~\ref{betamincond-extend} only requires that on average, the nonzero $\beta_{i,j}^*$'s are sufficiently far from zero.
\end{rmk}
\begin{rmk}\label{sufficientremark}
 In the context of support recovery in SPCA under the sparse spiked covariance model, a standard assumption appearing in the literature is: for every nonzero coordinate $u^*_{i}$, we assume
 a lower bound $|u^*_i|>u_{\min}\gtrsim 1/\sqrt{s}$---in particular, support recovery results of~\citep{bresler-10.5555/3327546.3327752,deshpande2016sparse} are derived under this assumption. Moreover, the theoretical analysis of~\citep{amini2009} for both SDP relaxation and diagonal thresholding is valid under the stronger assumption $u_i\in\{0,1/\sqrt{s},-1/\sqrt{s}\}$. The same strong assumption is also used by~\citep{krauthgamer2015semidefinite} for the SDP relaxation analysis. 
In light of identity~\eqref{beta_lem}, for $i,j \in S(u^*)$, this lower bound condition on $|u_i^*|$ leads to:
$$\theta |u_i^*u_j^*|\gtrsim \frac{\theta}{s}\gtrsim \beta_{\min} \gtrsim \sqrt{\frac{\log p}{n}}$$
which implies, $n\gtrsim (s^2\log p)/{\theta^2}.$
Therefore, a sufficient (but as we discuss in Section~\ref{subsubdeep}, not necessary) condition
for~\ref{betamincond} to hold true, is to have $u_{\min}\gtrsim 1/\sqrt{s}$ and $n\gtrsim s^2\log (p)/\theta^2$. This shows our estimator recovers the correct support in the regimes that are discussed commonly in the literature. However, we show in Section~\ref{subsubdeep} that there are examples where our theory guarantees support recovery for our estimator, but earlier results do not guarantee support recovery. 
\end{rmk}

\subsection{Discussion of related prior work}\label{regcomp}

Following our earlier discussion, the covariance thresholding algorithm of \cite{deshpande2016sparse} leads to an estimate $\bar{u}$ of $u^*$ with a bound on the error $|\sin\angle(\bar{u}, u^*)|$ scaling as $\sqrt{s^2\log(p/s^2)/n}$, which is the same as our proposal (up to logarithmic factors). On the other hand, \citet{deshpande2016sparse}'s support recovery results require $n\gtrsim s^2\log(p/s^2)$, while as we showed in Theorem~\ref{supportthm}, under suitable $\beta_{\min}$ condition, our estimator only requires $n\gtrsim s\log p$.
Our approach is different from the covariance thresholding algorithm. We operate on the precision matrix by jointly optimizing over $p$-many node-wise linear regression problems under a structured sparsity assumption, while covariance thresholding operates directly on the covariance matrix. In particular, our algorithm leads to an estimate $\hat{B}$ with nonzero entries appearing in a $s\times s$ sub-matrix, which may not be available via the covariance thresholding estimator.
Moreover, as a by-product of our estimation procedure criterion~\eqref{reg_multi}, we are able to recover (with high probability) the correct support of the PC i.e., 
$u^*$---the covariance thresholding algorithm on the other hand, needs a data-splitting method to recover the support of $u^*$. In our numerical experiments, our method appears to have a notable advantage over the covariance thresholding algorithm in terms of estimation and support recovery performance.

Our approach is related to the proposal of~\citet{bresler-10.5555/3327546.3327752} who use connections between sparse linear regression and SPCA to propose a polynomial time algorithm for SPCA. Their algorithm requires solving $p$ separate sparse linear regression problems, each problem performs a sparse regression with the $i$-th column of $X$ as response and the remaining columns as predictors.
Each sparse regression problem has an estimation error of $(s/n)\log p$ resulting in an overall error rate of $(sp/n)\log p$.
Based on these $p$ separate regression problems, Bresler et al. propose testing methods to identify the support of $u^*$ when $n\gtrsim s^2 \log p$. In contrast, Theorem~\ref{supportthm} shows that our estimator can recover the correct support under suitable signal conditions~\ref{betamincond} or~\ref{betamincond-extend} as long as $n\gtrsim s\log p$.
Additionally, the procedure of \cite{bresler-10.5555/3327546.3327752} requires choosing tuning parameters for each of $p$-many Lasso problem, which can be challenging in practice. 
Our approach differs in that we consider these $p$ different regression problems {\emph{jointly}} under a structured sparsity assumption and we require knowing
the support size $s$.  
We are able to estimate the matrix of regression coefficients with an estimation error scaling as $(s^2/n)\log(p/s)$ (see Theorem~\ref{reg_thm}), which is an improvement over~\citep{bresler-10.5555/3327546.3327752}
by a factor of $p/s$. 
Another point of difference is that the approach of \citep{bresler-10.5555/3327546.3327752} requires solving $p$ separate sparse linear regression problems which can be computationally expensive when $p$ is large, while our approach requires solving \emph{one} problem which we are able to scale to quite large instances by exploiting problem-structure. Our numerical experiments show that our estimator leads to superior statistical performance compared to~\citep{bresler-10.5555/3327546.3327752}.

\citep{https://doi.org/10.1111/rssb.12360} present an interesting polynomial-time algorithm for SPCA based on random projections. They show that their algorithm can achieve the minimax optimal rate in terms of estimation of the PC under certain incoherence assumptions when $n$ is sufficiently large. 
Their algorithm may require $O(p^2\log p)$ random projections.

To our knowledge, the oracle bounds presented under model misspecification (see Theorems~\ref{reg_thm} and~\ref{Bestconst}) are new. 
On a related note, \cite{johnstone2009consistency, cai2012} consider the SPCA problem when $u^*$ belongs to a weak $\ell_q$ ball. In another line of work, authors in \cite{10.1214/14-AOS1273} consider a misspecified model for SPCA where the covariance matrix might not necessarily follow a spiked model---their results are true under the so-called \emph{limited correlation} condition, which differs from what we consider here.
\subsubsection{Discussion of Theorem~\ref{supportthm} and comparisons to prior work}\label{subsubdeep}
We study the results of Theorem~\ref{supportthm} and examine how our theory improves upon earlier results. 

\smallskip 

\noindent\textbf{Minimum signal required for support recovery:} 
We note that our support recovery result in Theorem~\ref{supportthm} is generally tighter than the results appearing in the literature in terms of the minimum signal required (e.g., Condition~\ref{betamincond}). As mentioned in Remark~\ref{sufficientremark}, the current methods~\citep{amini2009,bresler-10.5555/3327546.3327752,deshpande2016sparse} require $|u_i^*|\gtrsim 1/\sqrt{s}$ for any $i\in[p]$ such that $u_i^*\neq 0$. This implies that for a fixed value of the sparsity budget $s$, the nonzero coordinates of $u^*$ cannot be made arbitrarily small.  
The implications of our results are different.
For example, let $p=\lceil n^{\alpha} \rceil$ for some $\alpha>0$ and fix $s$. Then, $\beta_{\min}$ in~\ref{betamincond} can be made arbitrarily small, because 
${\log p}/{n} \rightarrow 0$ as $n \rightarrow \infty$.
This shows our proposed estimator can correctly identify the support even when the nonzero entries in $u^*$ are arbitrarily close to zero. 
Specifically, we consider the following example:
\begin{exa}\label{example1}
Consider the sparse spiked covariance model~\eqref{Gcov1} with $s=17$ and $\theta=1$. Specifically, assume
    $$u^*_1 =  \sqrt{\frac{1000\eta\log p}{n}},~~u^*_2=\cdots=u^*_{17} = \frac{1}{4}\sqrt{1-\frac{1000\eta\log p}{n}},~~u_i^*=0, i=18,\cdots,p$$
     where $\eta$ is defined in Theorem~\ref{supportthm}.
Suppose
\begin{equation}\label{example-nmin}
    n\geq \frac{1000\eta\log p}{1-0.96^2}\lor \frac{2\eta\log p}{(0.24^2\times 0.5)^2}.
\end{equation}
    This implies
    $u_2^*,\cdots,u^*_{17}\geq 0.24$. Then, from~\eqref{beta_lem},
    $$\beta_{1,i}^*,\beta_{i,1}^* \geq\frac{1}{2} u_1^*u_i^* \geq 0.12\sqrt{\frac{1000\eta\log p}{n}}> \sqrt{\frac{2\eta \log p}{n}}, ~~~\text{for}~i=2,\cdots,17.$$
    Moreover, 
    $$\beta_{i,j}^*\geq \frac{1}{2}u_i^*u_j^*\geq \frac{0.24^2}{2}\geq \sqrt{\frac{2\eta\log p}{n}},~~i\neq j,~ i,j\in\{2,\cdots,17\}$$
    where the last inequality is by~\eqref{example-nmin}.
    Therefore, the underlying model described here satisfies  Condition~\ref{betamincond} of Theorem~\ref{supportthm}, implying that our estimator is guaranteed to recover the correct support for this case under condition~\eqref{example-nmin}, with high probability. 

    On the other hand, for any $c_b>0$, we have $u_1^*<c_b/\sqrt{s}=c_b/\sqrt{17}$ for sufficiently large $n\gtrsim \log p$, showing that theoretical guarantees of~\citep{amini2009,bresler-10.5555/3327546.3327752,deshpande2016sparse,krauthgamer2015semidefinite}, derived under the minimum signal assumption of $|u_1^*|\gtrsim 1/\sqrt{s}$ do not apply to this setup. 

    To gain further insight into why such methods might fail, we study the diagonal thresholding method for this case. In the diagonal thresholding method, one sorts the diagonal coordinates of $X^TX/n-I_p$ and chooses the coordinates corresponding to the $s$-largest diagonal entries of $X^TX/n-I_p$ to be in the support. Note that in the spiked covariance model, $X^TX/n-I_p$ concentrates around $u^*(u^*)^T$. Particularly, under the model we described in Example~\ref{example1}, 
    $$\E\left[\frac{(X^TX)_{1,1}}{n}-1\right]=(u_1^*)^2 = \frac{1000\eta\log p}{n},~~~~~\E\left[\frac{(X^TX)_{j,j}}{n}-1\right]=(u_j^*)^2 =0$$
for $j=18,\cdots,p$.
    Therefore, for the diagonal thresholding method to choose the first coordinate to be in the support over some coordinate $j$ for $j\geq 18$, one needs to be able to distinguish between two random variables: one converging to $1000\eta\log p/n$ and the other to zero. On the other hand, as our estimator operates on the $\{\beta_{i,j}\}$ space, it should be able to identify $\beta_{1,i}^*$, $i=2,\cdots,17$ to be different from zero. However, as we showed above, $\beta_{1,i}^*$ concentrates around some value $\gtrsim\sqrt{\log p/n}$. 
    If $n\gtrsim \log p$ is sufficiently large, we see that the diagonal thresholding method has to detect a significantly smaller signal compared to our method, which can explain why the theoretical guarantees of the diagonal thresholding method do not apply to the setup we discussed in Example~\ref{example1}.

   We now discuss the shortcomings of the Lasso-based method of~\cite{bresler-10.5555/3327546.3327752} for Example~\ref{example1}. Recall that the method of~\cite{bresler-10.5555/3327546.3327752} is based on solving $p$-many Lasso problems, where in the $j$-th Lasso problem, we regress $X_{:,j}$ onto $\{X_{:,i}\}_{i\neq j}$. \cite{bresler-10.5555/3327546.3327752} design a statistical test on the outcome of $j$-th Lasso to determine if coordinate $j$ should be included in the support or not, independent of any other coordinate. This is because these Lasso problems are solved independently; and they do not use the common sparsity structure of the underlying model to select the support.
In particular, looking at the first Lasso (regressing $X_{:,1}$ onto the remaining columns), the method of~\cite{bresler-10.5555/3327546.3327752} requires (we refer to Appendix C of~\cite{bresler-10.5555/3327546.3327752} for more details) the following condition
    $$\frac{(u_1^*)^2(1-(u_1^*)^2)}{2-(u_1^*)^2}\gtrsim \frac{1}{s}$$
     to hold true in order to declare that $u^*_1\neq 0$. This condition however does not hold in Example~\ref{example1}.
     On the other hand, our method can handle this case. 
     We believe the main reason is that in our discrete optimization based estimator, we estimate all regression coefficients \emph{jointly, in one-shot}. This allows for the use of the sparsity structure of the underlying model, resulting in better variable selection properties.
\end{exa}

\smallskip 

\noindent\textbf{Minimum number of samples required for support recovery:} 
We study an example where the number of samples required for full support recovery scales as $n\gtrsim s\log p$, which is an improvement over $n\gtrsim s^2\log p$ required by most existing estimators. 

\begin{exa}\label{example2}
Let us consider model~\eqref{Gcov1} with $\theta=1$. Let
    $$u^*_1=\cdots=u^*_{s-2} =  \sqrt{\frac{2000\eta\log p}{n}},~~u^*_{s-1}=u_s^*=  \frac{1}{\sqrt{2}}\sqrt{1-\frac{2000(s-2)\eta\log p}{n}},$$
and $u_i^*=0$ for $i>s$.
    Let $c$ be the universal constant from Theorem~\ref{supportthm}. Assume 
    \begin{equation}\label{exa2-n}
        n = (c \lor 4000\eta  )s\log p
    \end{equation}
    and $s\leq 125$. Then, we have 
    $$\frac{2000(s-2)\eta\log p}{n}\leq \frac{2000s\eta\log p}{n}\leq\frac{1}{2}$$
    which implies $u_{s-1}^*,u_s^*\geq 0.5$. Moreover, noting that $|\beta_{i,j}^*|\geq |u_i^*u_j^*|/2$, for $1\leq j\leq s-2$, we have:
    $$\sum_{i:i\neq j}(\beta_{i,j}^*)^2\geq (\beta_{s-1,j}^*)^2\geq 0.25^2 \frac{2000\eta \log p}{n}\geq \frac{s\eta\log p}{n}$$
    where the last inequality is true by using $s\leq 125$. Moreover, for $j=s-1$, 
    $$\sum_{i:i\neq s-1}(\beta_{i,s-1}^*)^2\geq (\beta_{s,s-1}^*)^2\geq \frac{1}{64}\geq \frac{s\eta\log p}{n}$$
    by~\eqref{exa2-n}. A similar statement holds for $j=s$, showing that Condition~\ref{betamincond-extend} in Theorem~\ref{supportthm} holds true. Hence, in Example~\ref{example2}, $n \geq (c \lor 4000\eta  )s\log p$ is sufficient for correct support recovery to occur. Proof of Theorem~\ref{supportthm} for the case with Condition~\ref{betamincond-extend} makes extensive use of the joint sparsity structure of Problem~\eqref{reg_multi}, showing another benefit of our integer programming based framework.  
\end{exa}
\subsection{Statistical properties of approximate solutions}
As discussed earlier, an appealing aspect of a MIP-based global optimization framework (see Section~\ref{optalg} for our custom MIP solver) is its ability to deliver both upper and lower bounds (or dual bounds) as the algorithm progresses. These upper and dual bounds, taken together, provide a certificate of how close the current objective value is to the optimal solution, and can be useful when one wishes to terminate the algorithm early due to computational budget constraints.   
Below we show that the estimation error of an approximate solution 
is comparable to an optimal solution to~\eqref{reg_multi} as long as the objective value is sufficiently close to the optimal objective.
For the following result, for simplicity, we consider the well-specified case where $u^*$ has $s$ nonzeros.
\begin{pro}\label{suboptpro}
Under model~\eqref{Gcov1} suppose $\|u^*\|_0=s$ and let $\beta^*$ be the matrix of true regression coefficients as in 
Lemma~\ref{Ginv}.
Moreover, assume $s\leq p/5$. Suppose  $\text{LB}\geq 0$ is a value of the lower bound that is, the optimal objective of Problem~\eqref{reg_multi} is at least $\text{LB}$. If $(\hat{\beta},\hat{z})$ is a feasible (though not necessarily an optimal) solution to Problem~(\ref{reg_multi}),
and $n\gtrsim s^2\log(p/s)$, then
\begin{equation}
   \sum_{\substack{ i,j\in[p]\\ j\neq i }}(\hat{\beta}_{i,j}-\beta_{i,j}^*)^2\lesssim \frac{s^2\log(p/s)}{n} + p\frac{\text{UB}-\text{LB}}{\text{LB}}
\end{equation}
with high probability\footnote{An explicit expression for the probability can be found in~\eqref{prop1prob}.} where
\begin{equation}
    \text{UB}=  \sum_{j=1}^p \left\Vert X_{:,j} -\sum_{i\neq j}\hat{\beta}_{i,j}X_{:,i} \right\Vert_2^2
\end{equation}
is the objective value of the feasible solution $\hat{\beta}$.
\end{pro}
Proposition \ref{suboptpro} states that for an approximate solution to Problem~(\ref{reg_multi}) with optimization error satisfying
$(\text{UB}-\text{LB})/\text{LB} \lesssim s^2\log(p/s)/(pn)$, the error bounds have the same scaling as that of an optimal solution to~\eqref{reg_multi} where $\text{UB}=\text{LB}$.

Next, we discuss support recovery properties of approximate solutions.
\begin{pro}\label{approx-supp}
Suppose the Assumptions of Theorem~\ref{reg_thm} hold with Condition~\ref{betamincond}. Let $(\hat{\beta},\hat{z})$ be a feasible (though not necessarily an optimal) solution to Problem~(\ref{reg_multi}) such that $\sum_{j=1}^p \hat{z}_j=s$, and let $\text{UB},\text{LB}$ be defined as in Proposition~\ref{suboptpro}. Then, with high probability\footref{sup_thm_prob-foot} the number of false negatives satisfies
$$|\{j\in[p]:z_j^*=1,\hat{z}_j=0 \}|\lesssim \frac{\text{UB}-\text{LB}}{s\log p}.$$
\end{pro}
Proposition~\ref{approx-supp} shows that as long as ${(\text{UB}-\text{LB})}/{(s\log p)}$ is sufficiently small, the support of $\hat{z}$ is correct. Even in the absence of full support recovery, Proposition~\ref{approx-supp} provides an upper bound on the total number of errors in the support, which can be useful to a practitioner using an approximate solution.

\section{Optimization Algorithms}\label{optalg}
Problem~\eqref{reg_multi} is a mixed integer optimization problem with $O(p)$ binary variables and $O(p^2)$ continuous variables. When $p$ is large (e.g., of the order of a few hundred), commercial solvers such as Gurobi and Mosek face computational challenges. Here, we propose a specialized algorithm for~\eqref{reg_multi} that can scale to much larger instances than off-the-shelf commercial solvers. The key idea of our approach is to reformulate~\eqref{reg_multi} into minimizing a convex function involving $p$ binary variables with additional cardinality constraints:
\begin{align}\label{generalmi}
     \min_{z} ~~~   F(z)  ~~~\text{s.t.} ~~  z\in\{0,1\}^p;~~   \sum_{i=1}^p z_i \leq s,
    \end{align}
where, as we show subsequently, $F:[0,1]^p\to\R$ is convex. 
Note that Problem~\eqref{generalmi} is an optimization problem with $p$ binary variables, unlike~\eqref{reg_multi} involving $O(p^2)$ continuous and $O(p)$ binary variables. To optimize~\eqref{generalmi}, we employ an outer approximation 
algorithm~\citep{duran1986outer} as we discuss below. The following section shows how to reformulate Problem~\eqref{reg_multi} into form~\eqref{generalmi}.

 \subsection{Reformulations of Problem~\eqref{reg_multi}}\label{misocpsec}
We discuss different reformulations of Problem~\eqref{reg_multi} that use the underlying statistical model to obtain important computational savings.
 
 \smallskip
 
\noindent{\bf Valid inequalities under model~\eqref{Gcov1}:} 
We discuss some inequalities or constraints that are {\emph{implied}} from the spiked covariance model in~\eqref{Gcov1}---when these constraints are included in 
formulation~\eqref{reg_multi}, the resulting solution has the same statistical properties as that of an estimator from the original formulation~\eqref{reg_multi}. We call such constraints/inequalities \emph{valid}\footnote{Our usage of the term is inspired by the notion of valid inequalities in integer programming~\citep{valid}, where a constraint is said to be valid if it does not change the optimal solution of the original integer program. We note, however, that our usage differs as the constraints we consider originate from the underlying statistical model, and the constraints are valid in the context of the statistical model~\eqref{Gcov1}.}.  
We use these inequalities or constraints as the resulting reformulation of Problem~\eqref{reg_multi} has better computational properties. The computational improvements are due in part to tighter convex relaxations in the nodes of the branch-and-bound procedure. 
 
Note that the complementarity constraint~\eqref{secondconst} in formulation~\eqref{reg_multi}, can be linearized by using BigM constraints 
\begin{equation}\label{BigM-compl1}
|\beta_{ij}| \leq M z_{i}, ~~ |\beta_{ij}| \leq M z_{j}, ~\forall i, j.
\end{equation}
Above, $M$ is a BigM parameter that controls the magnitude of $\hat{\beta}$, an optimal solution to~\eqref{reg_multi}---we need $M$ to be sufficiently large such that 
$M \geq \max_{i,j} |\hat{\beta}_{ij}|$. A large (conservative) choice of the BigM parameter $M$, would not affect the logical constraint~\eqref{secondconst}, but will affect the runtime of our mixed integer optimization algorithm~\cite[see for example]{bertsimas2016best}.
In other words, a tight estimate of $M$ is desirable from an algorithmic standpoint. While there are several ways to obtain a good estimate of $M$, here we make use of the statistical model~\eqref{Gcov1} to obtain an estimate of $M$. In particular, we show that it suffices to consider $M=1/2$. To see this, note that~\eqref{beta_lem} implies 
\begin{equation}\label{eqn-abs-con1}
|\beta^*_{i,j}|\leq \theta |u_i^*u_j^*|=\theta |u_i^*|\sqrt{1-\sum_{k:k\neq j}(u_k^*)^2}\leq \theta |u_i^*|\sqrt{1-(u_i^*)^2}\leq \theta/2
\end{equation}
so the absolute value of $\beta^*_{i,j}$ is less than $\theta/2\leq 1/2$. As a result, a choice of $M=1/2$ in~\eqref{BigM-compl1} implies that $\beta^*$ is feasible for 
Problem~\eqref{reg_multi} where, constraint~\eqref{secondconst} is replaced by~\eqref{BigM-compl1}.

Additionally, we show that under model~\eqref{Gcov1}, the squared $\ell_{2}$-norm of every column of $\beta^*$ can be bounded from above.  
From~\eqref{beta_lem}, for every column $j\in[p]$, we have:
 \begin{equation}\label{eqn-absnorm-con1}
 \sum_{i: i \neq j}{(\beta_{i,j}^*)}^2=\frac{\theta^2}{(1+\theta-\theta ({u_j^*})^2)^2}\sum_{i:i\neq j}({u_i^*})^2({u_j^*})^2\leq z_j\sum_{i:i\neq j}({u_i^*})^2\leq z_j, 
 \end{equation}
where $z\in\{0,1\}^p$ has the same support as $u^*$. Therefore, the inequality constraint $\sum_{i: i\neq j}\beta_{i,j}^2\leq M' z_j$ for some $M' \in (0, 1]$ (for example, $M'=1$) is a valid inequality for Problem~\eqref{reg_multi} under model~\eqref{Gcov1}---i.e., adding this inequality does not change the statistical properties of an optimal solution to the problem. In particular, owing to the sparsity structure in $\beta$, this implies that 
$\sum_{i,j: i \neq j} \beta_{i,j}^2 \leq sM'$. Instead of the constraint, we can also penalize the squared $\ell_2$-norm of the off-diagonal entries of $\{\beta_{i,j}\}$.
Based on the above discussion, 
we reformulate Problem~\eqref{reg_multi} as  
\begin{align}\label{intermediateform2}
     \min_{\beta,{z}} \quad &  \frac{1}{2}\sum_{j=1}^p \left\Vert X_{:,j} -\sum_{i\neq j}\beta_{i,j}X_{:,i} \right\Vert_2^2 + \lambda\sum_{j=1}^p\sum_{i\neq j}\beta_{i,j}^2 \\
     \text{s.t.} \quad & {z}\in \{0,1\}^p;~~|\beta_{i,j}|\leq Mz_i,~~ |\beta_{i,j}|\leq Mz_j, ~~\beta_{i,i}=0~~ \forall i,j\in[p];~~\sum_{i=1}^p {z}_i \leq s,\nonumber
\end{align}
where $\lambda\geq 0$ is a parameter that corresponds to $M'$---a nonzero value of $\lambda$ in~\eqref{intermediateform2} leads to computational benefits, as we discuss below. Note that by setting $\lambda=0$ we recover the original Problem~\eqref{reg_multi}.

\begin{rmk}
Under the sparse spiked covariance model assumption (i.e., the setup of Theorem~\ref{supportthm}) and Assumption~\ref{beta_min_assum}, the theoretical guarantees of Theorems~\ref{reg_thm} and~\ref{Bestconst} extend to Problem~\eqref{intermediateform2} with $M\geq 1/2$ (as a consequence of~\eqref{eqn-abs-con1}) and $\lambda\lesssim s^2 \log p$.
\end{rmk}
\noindent{\bf Perspective Formulation:} Formulation~\eqref{intermediateform2} involves bounded continuous variables and binary variables---the binary variables activate the sparsity pattern of the continuous variables. In MIP problems with a similar structure (i.e., where bounded continuous variables are set to zero by binary variables), the perspective reformulation is often used to obtain stronger MIP formulations~\citep{Frangioni10.5555/1113382.1113384,akturk2009strong,gunluk2010perspective}. In other words, these MIP formulations lead to tight continuous relaxations, which in turn can lead to improved runtimes in the overall branch-and-bound method. Here we apply a perspective relaxation on the squared $\ell_2$-term $\sum_{i\neq j}\beta_{i,j}^2$, appearing in the objective in~\eqref{intermediateform2}. This leads to the following MIP formulation:
\begin{align}\label{perspectivereform}
     \min_{\beta,{z}} ~~ &  \frac{1}{2}\sum_{j=1}^p \left\Vert X_{:,j} -\sum_{i\neq j}\beta_{i,j}X_{:,i} \right\Vert_2^2 +\lambda \sum_{j\in[p]}\sum_{i\neq j}q_{i,j}\\
     \text{s.t.} ~~ & {z}\in \{0,1\}^p;~~q_{i,j}\geq0,~~\beta_{i,j}^2\leq q_{i,j}z_j,~~|\beta_{i,j}|\leq Mz_i,~|\beta_{i,j}|\leq Mz_j~~\forall i,j\in[p]\nonumber \\
     & \beta_{i,i}=0~\forall i\in[p],~~\sum_{i=1}^p {z}_i \leq s\nonumber
\end{align}
where above we have introduced auxiliary nonnegative continuous variables $\{q_{ij}\}$ and added rotated second order cone constraints 
$\beta_{i,j}^2\leq q_{i,j}z_j$ for all $i,j$. 
When $z_{j}=0$ then $q_{i,j}=\beta_{i,j}=0$; 
when $z_{j}=1$ then, at an optimal solution, $q_{i,j}=\beta_{i,j}$. 
Note that~\eqref{perspectivereform} is an equivalent reformulation\footnote{That is, the optimal objectives of both these problems are the same, and an optimal solution to~\eqref{perspectivereform} leads to an optimal solution to~\eqref{intermediateform2} and vice-versa.} of Problem~\eqref{intermediateform2}. 
However, when $\lambda>0$, the interval relaxation of formulation~\eqref{perspectivereform} (obtained by relaxing all binary variables to the unit interval), results in a tighter relaxation compared to the interval relaxation of~\eqref{intermediateform2}~\citep[see for example]{hazimeh2020sparse}. 
Our proposed custom algorithm solves Problem~\eqref{perspectivereform} to optimality, and can handle both cases 
$\lambda=0$ and $\lambda>0$.

\smallskip

\noindent{\bf A convex integer programming formulation for~\eqref{perspectivereform}:} When $\lambda \geq 0$, Problem~\eqref{perspectivereform} is a MISOCP\footnote{In the special case of $\lambda=0$, this can be expressed as a mixed integer quadratic program}. Commercial solvers like Gurobi, Mosek can reliably solve MISOCPs for small-moderate scale problems. However, formulation~\eqref{perspectivereform} involves $O(p^2)$-many continuous variables and $O(p)$-many binary variables --- posing computational challenges for modern MIP solvers. Therefore, we propose a tailored algorithm to solve Problem~\eqref{perspectivereform}. To this end, we first reformulate Problem~\eqref{perspectivereform} into a convex integer program i.e., in the form~\eqref{generalmi}. To lighten notation, let $X_{-j}\in\mathbb{R}^{n\times p}$ be the matrix $X$ with the column $j$ replaced with zero. For $z\in[0,1]^p$, let us define the following function:
\begin{align}\label{Freg2}
     F_1(z)=\min_{\beta,\xi,W,q} \quad &  \frac{1}{2}\sum_{j=1}^p \left\Vert \xi_{:,j} \right\Vert_2^2+\lambda\sum_{j=1}^p\sum_{i\neq j}q_{i,j} \\
     \text{s.t.} \quad & q_{i,j}\geq0,~~\beta_{i,j}^2\leq q_{i,j}z_j~~\forall i,j\in[p]\nonumber \\
     & |\beta_{i,j}|\leq MW_{i,j},~~W_{i,j}\leq z_i,~~W_{i,j}\leq z_j,~~\beta_{i,i}=0~~ \forall i,j\in[p]\nonumber\\
     \quad & \xi_{:,j} = X_{:,j}-X_{-j}\beta_{:,j}~~~\forall j\in[p]\nonumber
\end{align}
obtained by minimizing~\eqref{perspectivereform} partially with respect to all continuous variables, with a fixed $z$. 
In light of formulation~\eqref{Freg2}, we see that 
Problem~\eqref{perspectivereform} is equivalent to~\eqref{generalmi} where $F(z)$ is replaced with $F_1(z)$. 

Proposition~\ref{subgradpers} shows that $F_1$ is convex and characterizes its subgradient. Before presenting the proposition, we introduce some notation.
    For $j\in[p]$, define
\begin{align}\label{linregperstosolve}
     \bar{\beta}_{:,j}\in\argmin_{\beta} ~ \frac{1}{2} \left\Vert X_{:,j} - X_{-j}\beta_{:,j} \right\Vert_2^2+\lambda\sum_{i\neq j}\beta_{i,j}^2~~ \text{s.t.} ~~|\beta_{i,j}|\leq Mz_i ~\forall i\in[p], \beta_{j,j}=0
\end{align}
when $z_{j}=1$; and $\bar{\beta}_{:,j}=0$ when $z_{j}=0$. Let
$\alpha\in \R^{n\times p}$ be a matrix, with its $j$-th column given by:
\begin{equation}\label{regpersalpha}
   \alpha_{:,j}=\begin{cases}
    X_{:,j}-X_{-j}\bar{\beta}_{:,j} &\mbox{ if } z_j=1\\
    X_{:,j} & \mbox{ if } z_j = 0.
    \end{cases}
\end{equation}
Define matrices $\Gamma^{(1)},\Gamma^{(2)}\in\mathbb{R}^{p\times p}$ as follows
\begin{equation}\label{regpersm1m2}
\begin{aligned}
     \Gamma^{(1)}_{i,j}&=\begin{cases}
    M|(X_{-j}^T\alpha_{:,j})_i-2\lambda\bar{\beta}_{i,j}|/2 &\mbox{ if } z_i=z_j=1\\
     M|(X_{-j}^T\alpha_{:,j})_i|/2 &\mbox{ if } z_i=z_j=0\\
    0 & \mbox{ if } z_i=1,z_j=0 \\
     M|(X_{-j}^T\alpha_{:,j})_i| & \mbox{ if } z_i=0,z_j = 1,
    \end{cases}
    \\
        \Gamma^{(2)}_{i,j}&=\begin{cases}
    M|(X_{-j}^T\alpha_{:,j})_i-2\lambda\bar{\beta}_{i,j}|/2 &\mbox{ if } z_i=z_j=1\\
    M|(X_{-j}^T\alpha_{:,j})_i|/2 &\mbox{ if } z_i=z_j=0\\
    M|(X_{-j}^T\alpha_{:,j})_i| & \mbox{ if } z_i=1,z_j=0 \\
     0 & \mbox{ if } z_i=0,z_j = 1.
    \end{cases}
\end{aligned}
\end{equation}
\begin{pro}\label{subgradpers}
Let $F_{1}$ be as defined in~\eqref{Freg2} with $\lambda \geq 0$. The following hold true
\begin{enumerate}
    \item (\textit{Convexity}) The function $z \mapsto F_1(z)$ on 
    $z \in [0,1]^p$ is convex. 
    \item (\textit{Subgradient}) 
Let $\bar{\beta}$, $\Gamma^{(1)}$ and $\Gamma^{(2)}$ be as defined~\eqref{linregperstosolve},\eqref{regpersalpha} and~\eqref{regpersm1m2}.
The vector $g\in\mathbb{R}^p$ with $i$-th coordinate given by
$$g_i= -\sum_{j=1}^p  \left\{ \Gamma^{(1)}_{i,j}+\Gamma^{(2)}_{j,i}+\lambda\bar{\beta}_{j,i}^2\right\}$$
for $i\in[p]$, is a subgradient of $z \mapsto F_1(z)$ for $z\in\{0,1\}^p$.
\end{enumerate}
\end{pro}

Note that $F_{1}(z)$ in~\eqref{Freg2} is implicitly defined via the solution of a quadratic program (QP). 
For a vector $z$ which is dense, calculating $F_{1}$ in \eqref{Freg2} requires solving $p$-many QPs each with $p$ variables. However, for a feasible $z$ which is $s$-sparse, calculating $F_{1}$ requires solving $s$-many QPs each with $s$ variables, which is substantially faster---furthermore, these QPs are independent of each other and can hence be solved in parallel. Similar to the case of evaluating the objective $F_{1}(z)$, we do not have a closed form expression for a subgradient of $F_{1}(z)$ (cf Proposition~\ref{subgradpers}), but this can be computed for $z \in \{0,1\}^p$ as a by-product of solving the 
QP in~\eqref{Freg2}. 

\subsection{Outer Approximation Algorithm}\label{outersec}
We present an outer approximation (or cutting plane) algorithm~\citep{duran1986outer} to solve Problem~\eqref{generalmi}. 
Our algorithm requires access to an oracle that can compute the tuple 
$(F(z), g^{(z)})$
where, $g^{(z)} \in \partial F(z)$ is a subgradient of $F(z)$ at an integral $z$. We refer the reader to Proposition~\ref{subgradpers} for computation of $(F(z), g^{(z)})$ when $F = F_{1}$.
As $F$ is convex, we have 
$F(x)\geq F(z) +g^T(x-z)$ for all $x\in [0,1]^p$. Therefore, if $z^0,\cdots,z^{t-1}$ are feasible for \eqref{generalmi}, we have: 
    \begin{equation}\label{lowerbound}
    \begin{aligned}
     F_{LB}(z):=&\max\left\{F(z^0)+(z-z^0)^Tg^{(z^0)},\ldots,F(z^{t-1})+(z-z^{t-1})^Tg^{(z^{t-1})}\right\} \\
     \leq& F(z)
    \end{aligned}
    \end{equation}
    where, $z \mapsto F_{LB}(z)$ is a lower bound to the map $z \mapsto F(z)$ on $z \in [0,1]^p$.
        At iteration $t\geq 1$, the outer approximation algorithm finds
        $z^{t}$, a minimizer of $F_{LB}(z)$ under the constraints of Problem~\eqref{generalmi}. 
        This is equivalent to the mixed integer linear program:
    \begin{align}\label{outerMILP}
    (z^t,\eta^t) \in\argmin_{z,\eta} \quad  & \eta \\
     \text{s.t.} \quad & z\in\{0,1\}^p,~\sum_{i=1}^p z_i \leq s\nonumber\\
     & \eta \geq F(z^i)+(z-z^i)^Tg^{(z^{i})}, i = 0,\cdots, t-1\nonumber .
    \end{align}
    As the feasible set of Problem \eqref{generalmi} contains finitely many elements, an optimal solution is found after finitely many iterations\footnote{This is true as the lower bound obtained by the outer approximation in each iteration removes the current solution from the feasible set, unless it is optimal---see~\citet{duran1986outer} for details.}, say, $t$.  In addition, $\eta^t$ is a lower bound of the optimal objective value in~\eqref{generalmi}; and $z^{t}$ leads to an upper bound for Problem~\eqref{generalmi}.
    Consequently, the optimality gap of the outer approximation algorithm can be calculated as $\text{OG}=(\text{UB}-\text{LB})/\text{UB}$ where $\text{LB}$ is the current (and the best) lower bound achieved by the piecewise approximation 
    in~\eqref{outerMILP}, and $\text{UB}$ is the best upper bound thus far. The procedure is summarized in Algorithm \ref{outer}.
    \begin{algorithm}[h]
\SetAlgoLined
$t = 1$\\
\While {$\text{OG} > \text{tol}$}{
$(z^t,\eta^t)$ are solutions of \eqref{outerMILP}.\\
$F_{\text{best}}=\min_{i=0,\cdots,t} F(z^i)$\\
$\text{OG} = (F_{\text{best}}-\eta_t)/F_{\text{best}}$\\
$t = t + 1$
}
 \caption{Cutting plane algorithm}
 \label{outer}
\end{algorithm}

\begin{rmk}
Outer approximation based methods are commonly used to solve nonlinear integer programs, and have been used recently in the context of $\ell_0$-sparse regression problems~\citep[see for example]{bertsimas2020sparse,behdin2021}. Our approach differs in that we consider structured sparsity, and use a perspective formulation which is motivated by the underlying spiked covariance model resulting in a MISOCP~\eqref{perspectivereform}. 
In particular, \cite{bertsimas2020sparse} consider an outer approximation method for a ridge regularized sparse regression problem where the tuple 
$(F(z), \nabla F(z))$ can be computed in closed form, which is in contrast to our setting where $(F(z), \nabla F(z))$ is available implicitly (cf Proposition~\ref{subgradpers}). 
\cite{bertsimas2020solving}~consider an outer approximation approach for sparse PCA, but they consider a mixed integer semidefinite optimization problem, which is computationally more demanding than what we propose here. 
 \end{rmk}

\subsubsection{Initializing Algorithm~\ref{outer}}\label{initsec} 
While Algorithm~\ref{outer} leads to an optimal solution to~\eqref{generalmi} (irrespective of the initialization used), we have empirically observed that a good initial solution can decrease the number of the iterations and the overall runtime of the algorithm. 
We describe an initialization scheme that we found to be useful in our work.
We first obtain a screened set ${\mathcal A} \subset [p]$, with 
size a constant multiple of $s$, 
with the hope that this contains many of the true nonzero coefficients.  
In our experiments, we use the diagonal thresholding algorithm \citep{johnstone2009consistency} with a larger number of nonzeros $\tilde{s}=3s$ to obtain the candidate set $\mathcal{A}\subseteq[p]$ with $|\mathcal{A}|\leq 3s$. 
We then consider a restricted version of 
Problem~\eqref{intermediateform2} with $z_{i}=0$ for all $i \notin A$ to obtain an estimate for $z^0$. This reduced problem contains $|{\mathcal A}|$-binary variables, which is considerably easier to solve compared to the original formulation with $p$-binary variables. In our experiments, this usually leads to a good initialization for Problem~\eqref{generalmi}.

\section{Numerical Experiments}\label{numerical}
We demonstrate the performance of our approach via numerical experiments on synthetic and real datasets and present comparisons with several 
state-of-the-art SPCA algorithms.  We have implemented all algorithms in Julia and use Gurobi v9.1.1 to solve the mixed integer linear sub-problems. We refer to our framework as {\texttt{SPCA-SLS}} (where SLS stands for Structured Least Squares). An implementation of our framework {\texttt{SPCA-SLS}} in Julia is available at:

~~~~~~~~~~~~~~~~\url{https://github.com/kayhanbehdin/SPCA-SLS}.
\subsection{Synthetic Data}\label{synthetic}
In this section, we use synthetic data and consider different scenarios and values for parameters in the model. In all experiments, we set the SNR value to $\theta=1$ which is a lower SNR level compared to experiments in \citep{deshpande2016sparse,bresler-10.5555/3327546.3327752}---this makes the problems statistically harder.

\smallskip

\noindent \textbf{Experimental Setup.} For a fixed set of values of $p,n,s$, each coordinate of the PC, $u^*$, is drawn independently from $\text{Unif}[0,1]$. Next, $(p-s)$ randomly chosen coordinates of $u^*$ are set to zero and $u^*$ is normalized as $u^*/\|u^*\|_2$ to have unit $\ell_{2}$-norm. We draw $n$ samples from the multivariate normal distribution with mean zero and covariance $I_p+u^*({u^*})^T$. 

\smallskip

\noindent \textbf{Competing Methods.} We present comparisons with the following  algorithms for SPCA, as discussed in Section~\ref{sec:lit-review}:
\begin{itemize}
    \item The Truncated Power Method (shown as {\texttt{TrunPow}})\citep{yuan2013truncated}---this is a heuristic method
    \item  The Covariance Thresholding (shown as {\texttt{CovThresh}})\citep{deshpande2016sparse}---this is a polynomial-time method with good theoretical guarantees
    \item The method by \cite{bresler-10.5555/3327546.3327752}  (shown as {\texttt{SPCAvSLR}})---we use a Lasso for every node, leading to a polynomial-time algorithm.
\end{itemize}
All the above methods can handle problems with $p \approx 10^4$. We do not consider other MIP-based methods and the SDP-based relaxations (See  Section~\ref{sec:lit-review})
in our experiments, as they are unable to handle the 
problem sizes we consider here.

\smallskip 

\noindent \textbf{Parameter Selection.} 
For our proposed algorithm (i.e., {\texttt{SPCA-SLS}}, which considers~\eqref{perspectivereform}), and the truncated power method, we assume that the sparsity level $s$ is known. The parameter $\lambda$ in~\eqref{perspectivereform} is set to zero unless stated otherwise. Following the discussion of \cite{yuan2013truncated} on initialization, we initialize {\texttt{TrunPow}} by $e_j$ with $j$ randomly chosen in $[p]$ where $e_j$ is the vector with all coordinates equal to zero except coordinate $j$ equal to one. For {\texttt{CovThresh}}, we choose the threshold level based on the theoretical results of~\cite{deshpande2016sparse} as $\alpha\sqrt{\log(p/s^2)/n}$ for some\footnote{In our experiments, we choose a value of $\alpha \approx 2$ that leads to the best estimation performance.} $\alpha>0$.  
For {\texttt{SPCAvSLR}}, we use Lasso for each nodewise regression 
problem and use a tuning parameter as per the recommendation of~\cite{bresler-10.5555/3327546.3327752}.

\smallskip 

\noindent \textbf{Results.} 
We consider different scenarios for our numerical experiments. First, we compare our newly proposed custom algorithm to a commercial solver in Section~\ref{vsGurobi}. Next, we fix $p$ and $s$ and investigate the effect of changing $n$ on the statistical and computational performance of different methods in Section \ref{varyingn}. Then, we fix $n,p$ and vary $s$ in Section \ref{varyings}. Finally, we run our algorithm on a large dataset with $p=20,000$ in Section~\ref{miqpvsmisocp}. To compare the quality of estimation for each algorithm, we report $|\sin\angle (\hat{u},u^*)|$ in our figures, where $\hat{u}$ is the recovered PC and $u^*$ is the true PC. To compare the quality of the support recovery, we report false positives and false negatives as $(|S^*\setminus \hat{S}|+|\hat{S}\setminus {S}^*|)/2$ where $S^*$ is the correct support and $\hat{S}$ is the estimated support. We perform the experiments on 10 replications, and summarize the results.

\subsubsection{Optimization performance: Comparison with MIP solvers}\label{vsGurobi}
Before investigating the statistical performance of our proposed approach, we show how our proposed algorithm can solve large-scale instances of Problem~\eqref{perspectivereform}.

We use our custom algorithm to solve 
Problem~\eqref{perspectivereform} with $\lambda=0$ --- we call this  \texttt{SPCA-SLS}. We also consider~\eqref{perspectivereform} with a small value of $\lambda>0$, taken to be $\lambda= 0.1{\sum_{j=1}^p\|X_{:,j}-\sum_{i\neq j}\beta_{i,j}^0X_{:,i}\|_2^2}/{\sum_{j=1}^p\sum_{i\neq j}(\beta^0_{i,j})^2}$ where $\beta^0$ is the initial solution---we call this \texttt{SPCA-SLSR}. We compare our approaches with Gurobi's MIP solver which is set to solve~\eqref{perspectivereform} with $\lambda=0$.
We run experiments for different values of $p,n,s$. In Table~\ref{table1}, we report the average optimality gap achieved after at most 5 minutes, for {\texttt{SPCA-SLS}}, {\texttt{SPCA-SLSR}} and Gurobi's MIP-solver. 
We observe that Gurobi quickly struggles to solve~\eqref{perspectivereform} when $p>100$ --- our custom approach on the other hand, can obtain near-optimal solutions to~\eqref{perspectivereform} for up to $p \approx 20,000$. For reference, recall that problem~\eqref{perspectivereform} involves $O(p^2)$-many continuous variables --- therefore, our approaches {\texttt{SPCA-SLS}}, {\texttt{SPCA-SLSR}} are solving quite large MIP problems to near-optimality within a modest computation time limit of $5$ minutes.
We also see that {\texttt{SPCA-SLSR}} provides better optimality gaps---we believe this due to the use of the perspective formulation, which leads to tighter relaxations.

\begin{table}[t!]
\centering
\begin{tabular}{ |c||ccc|ccc| }
 \hline
 & \multicolumn{3}{c|}{$s=5,n=500$} & \multicolumn{3}{c|}{$s=10,n=1000$}\\
 & {\texttt{SPCA-SLS}}&{\texttt{SPCA-SLSR}} & Gurobi &{\texttt{SPCA-SLS}} &{\texttt{SPCA-SLSR}} & Gurobi \\
 \hline
$p=100$&$3.1\%$& $2.3\%$ & $16.3\%$ & $3.0\%$ & $2.1\%$&$9.21\%$ \\
$p=1000$& $4.9\%$ &$3.8\%$& - & $3.9\%$ & $3.0\%$ & -\\
$p=10000$& $5.2\%$ &$4.2\%$& -&$4.1\%$& $3.1\%$ & -\\
$p=20000$& $6.8\%$ & $5.1\%$ & -&$5.4\%$& $3.7\%$ & -\\
 \hline
\end{tabular}
\caption{\small Comparison of Gurobi's off-the-shelf solver (Gurobi) and our specialized algorithm in Section~\ref{vsGurobi}. The numbers show the average optimality gap after 5 minutes. A dash shows no non-trivial dual bound was returned.}
\label{table1}
\end{table}

\subsubsection{Estimation with varying sample sizes ($n$)}\label{varyingn}
In this scenario, we fix $p=10,000$ and let $n\in\{5000,6000,7000,8000,9000,10000\}$ and $s\in\{5,10\}$. We compare our algorithm with the competing methods outlined above. Our experiments in this scenario are done on a machine equipped with two Intel Xeon Gold 5120 CPU $@$ 2.20GHz, running CentOS version~7 and using 20GB of RAM. The runtime of our method is limited to 2 minutes\footnote{For the examples in Figure~\ref{fig:synt1}, our custom algorithm provides optimality gaps less than $4\%$ and $7\%$ for $s=5$ and $s=10$, respectively (after 2 minutes). As a point of comparison, off-the-shelf MIP-solvers are unable to solve these problem instances (cf Table~\ref{table1}).}.

The results for these examples are shown in Figure~\ref{fig:synt1}: we compare the 
estimation (top panels) and support recovery (bottom panels) performance of different algorithms.  In terms of estimation, increasing $n$ results in a lower estimation error as anticipated by our theoretical analysis in Section~\ref{statreg}. In addition, our method {\texttt{SPCA-SLS}} appears to work better than competing algorithms and leads to lower estimation error. {\texttt{SPCA-SLS}} also appears to provide the best support recovery performance in these cases. The method {\texttt{CovThresh}} does not lead to a sparse solution and therefore has a high false positive rate, but may still lead to good estimation performance. On the other hand, {\texttt{SPCAvSLR}} has a higher false negative rate compared to {\texttt{CovThresh}} which leads to worse estimation performance. We note that in our experiments, the non-zero entries of the PC were drawn from the uniform $\text{Unif}[0,1]$ distribution, which can lead to some small nonzero coordinates. As discussed in Example~\ref{example1}, for our method, the signal strength required for support recovery is weaker than other methods --- this can partially explain the better support recovery performance of our method over competing methods in the figure.

\begin{figure*}[t!]
     \centering
     \scalebox{0.95}{\begin{tabular}{lll}
&  $~~~~~~~~~~~~~~~~~~p=10^4,s=5$ & $~~~~~~~~~~~~~~~~~~p=10^4,s=10$ \\
     \rotatebox{90}{~~~~~~~~~~~~~~~~~~~~~~~~~~~~~~~~~~~$|\sin\angle(\hat{u},u^*)|$}& %\includegraphics[width=5.94cm, trim=0cm 0cm 0cm 0cm, clip]{10000_5.pdf}& 
\includegraphics[width=0.46\textwidth,trim=.5cm 0cm 0cm 0cm, clip]{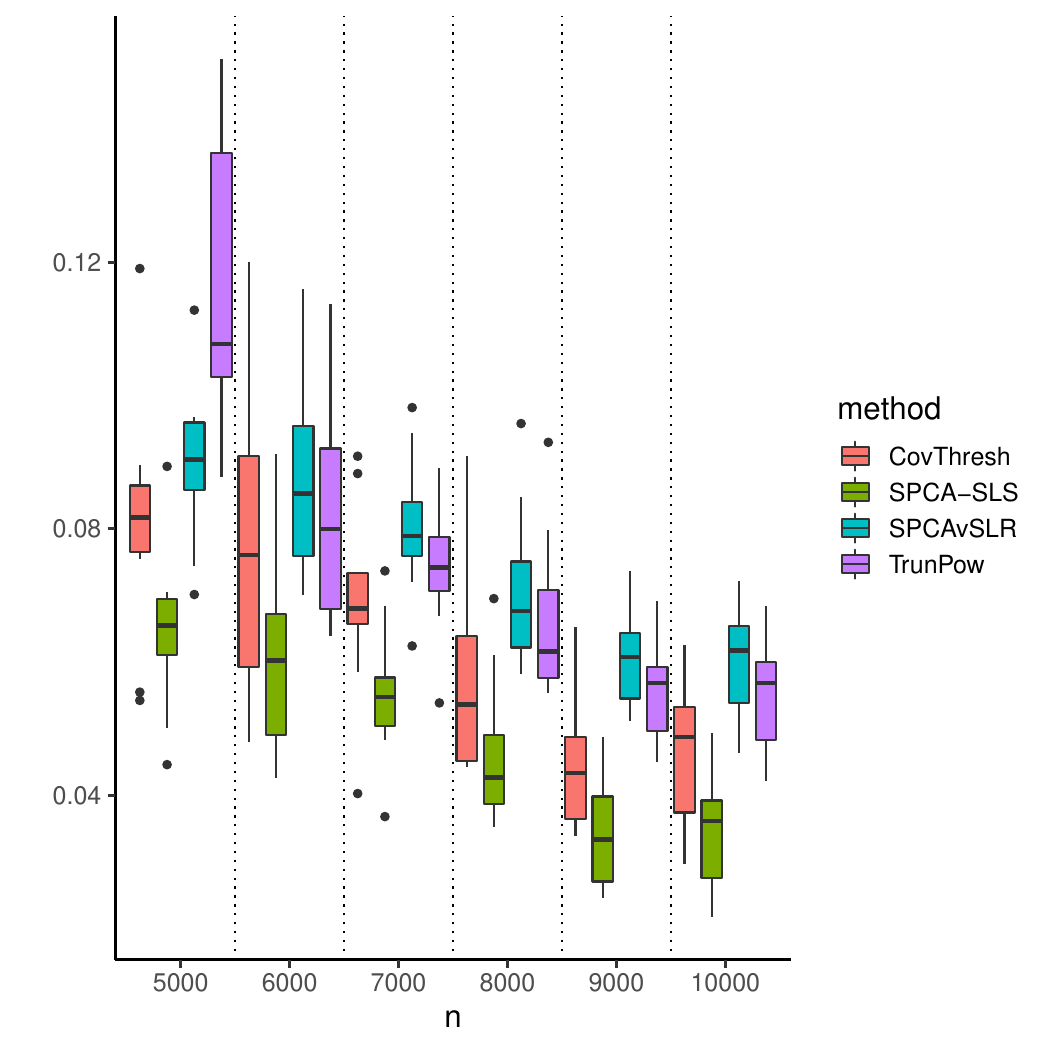}& 
 \includegraphics[width=0.46\textwidth,trim=.5cm 0cm 0cm 0cm, clip]{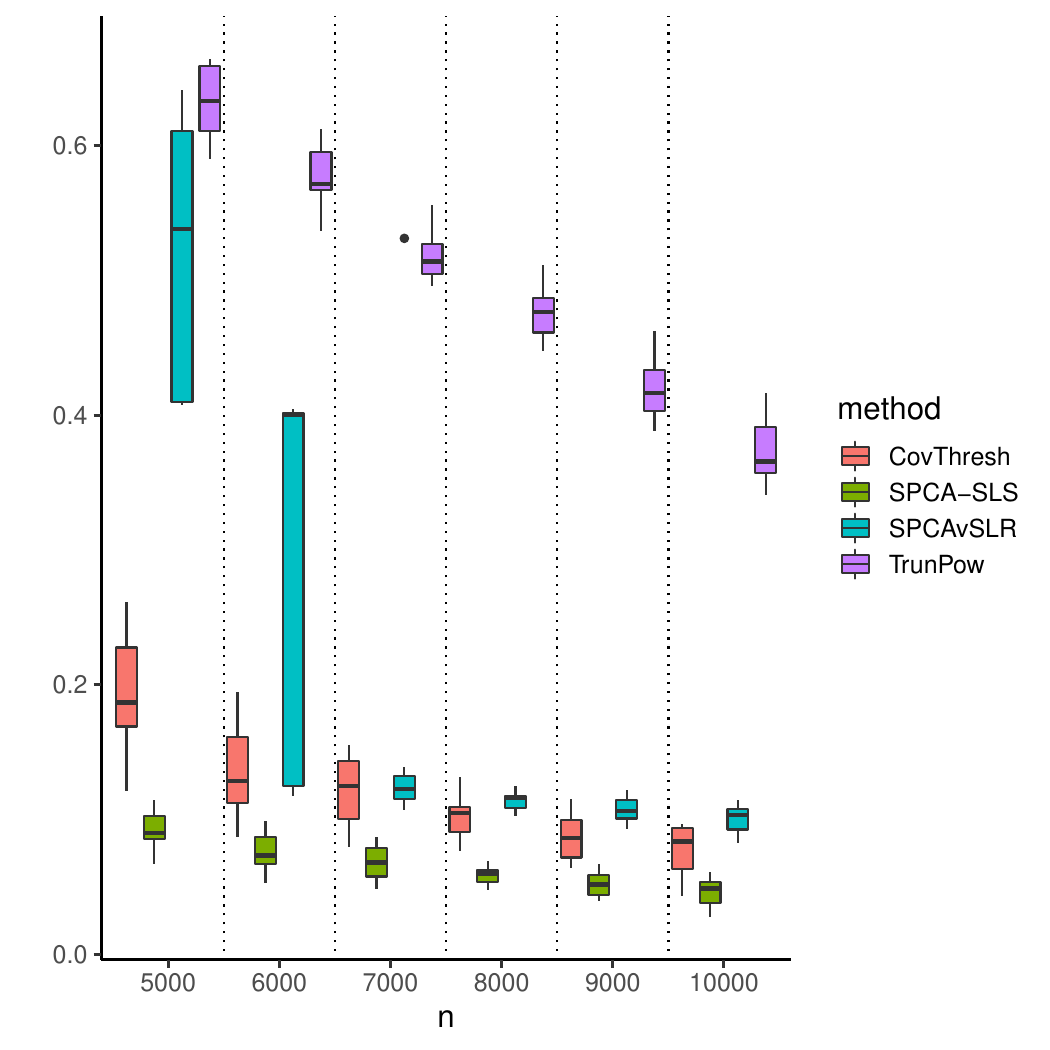}\\ 
       \rotatebox{90}{~~~~~~~~~~~~~~~~~~~~~~$(FP+FN)/2$}& \includegraphics[trim=1cm 0cm 0cm 0cm, clip,width=0.45\textwidth]{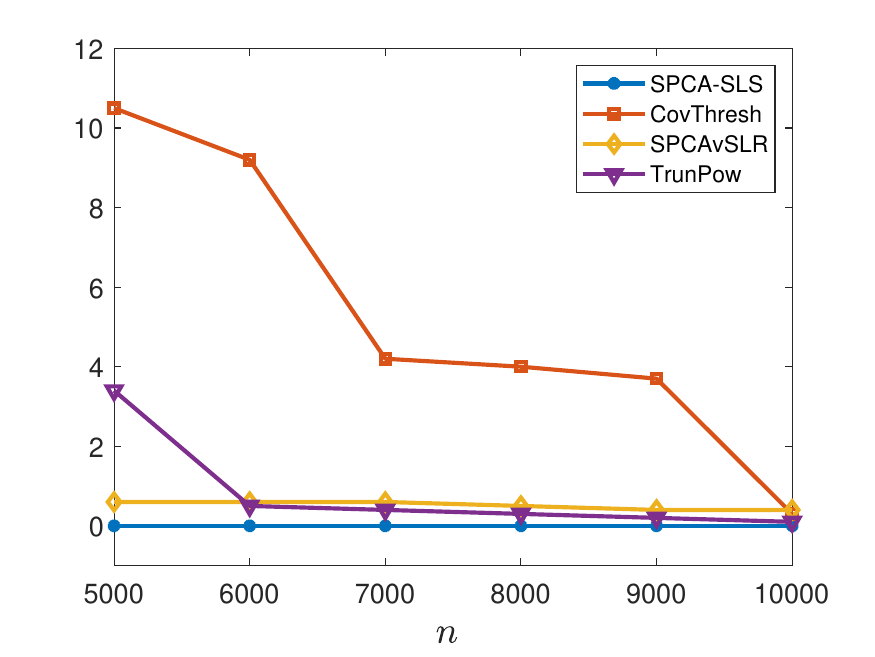}& 
      \includegraphics[trim=1cm 0cm 0cm 0cm, clip,width=0.45\textwidth]{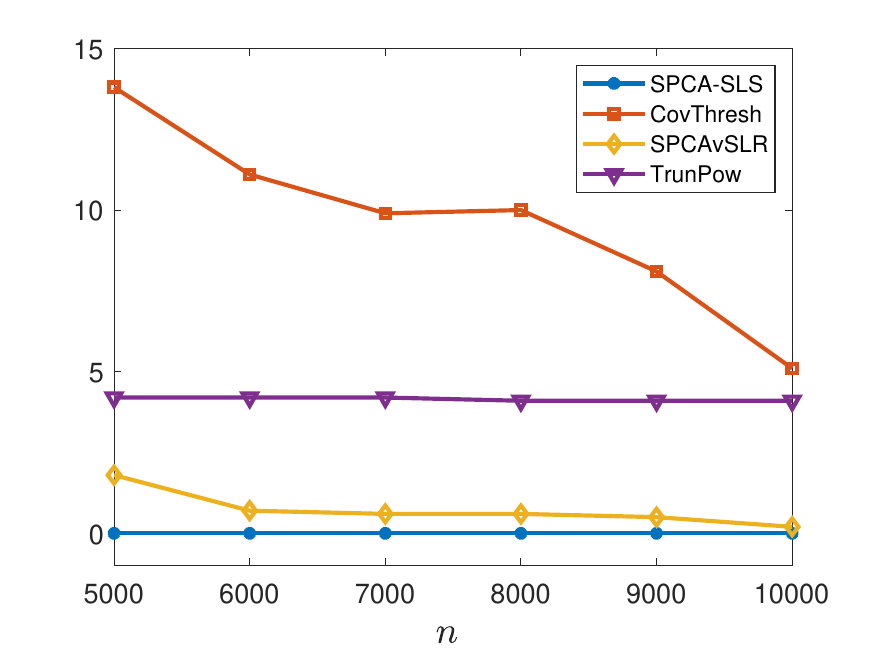}
\end{tabular}}
        \caption{\small Estimation and support recovery error as available from different methods on synthetic data with $p=10000$ and different values of $n$ (along the x-axis), as discussed in the text (Section \ref{varyingn}). 
        The left panels show results with $s=5$ and right panels $s=10$. The top panels compare estimation performance and bottom ones compare support recovery. Our proposed approach results in high-quality estimation performance, and perfect support recovery.}
        \label{fig:synt1}
\end{figure*}

\subsubsection{Estimation with varying sparsity levels ($s$)}\label{varyings}
In this scenario, we fix $p=10,000$ and $n=7500$ and consider different sparsity levels $s\in\{5,7,10,12,15\}$. We compare our {\texttt{SPCA-SLS}} algorithm with competing methods outlined above. We use the same computational setup as in Section~\ref{varyingn}.

The results for this scenario (estimation and support recovery properties) are shown in Figure \ref{fig:synt4}. 
As it can be seen, increasing $s$ results in worse statistical performance. However, our proposed approach continues to work better than other estimators in terms of estimation and support recovery\footnote{In all these cases, our algorithm delivers a near-optimal solution. The MIP-optimality gaps are $2\%,3.5\%,5\%,6.5\%$ and $8\%$ for values of $s=3,5,7,10,12,15$, respectively. The runtime of our method is limited to 2 minutes.}.

\begin{figure*}[t!]
     \centering
     \begin{tabular}{lclc}
&  Estimation Error, $p=10^4,n=7500$ & & Support Recovery, $p=10^4,n=7500$  \\
     \rotatebox{90}{~~~~~~~~~~~~~~~~~~~~~~~~~~~~~~~~~$|\sin\angle(\hat{u},u^*)|$}& \includegraphics[width=0.42\linewidth,trim=.5cm 0cm 0cm 0cm, clip]{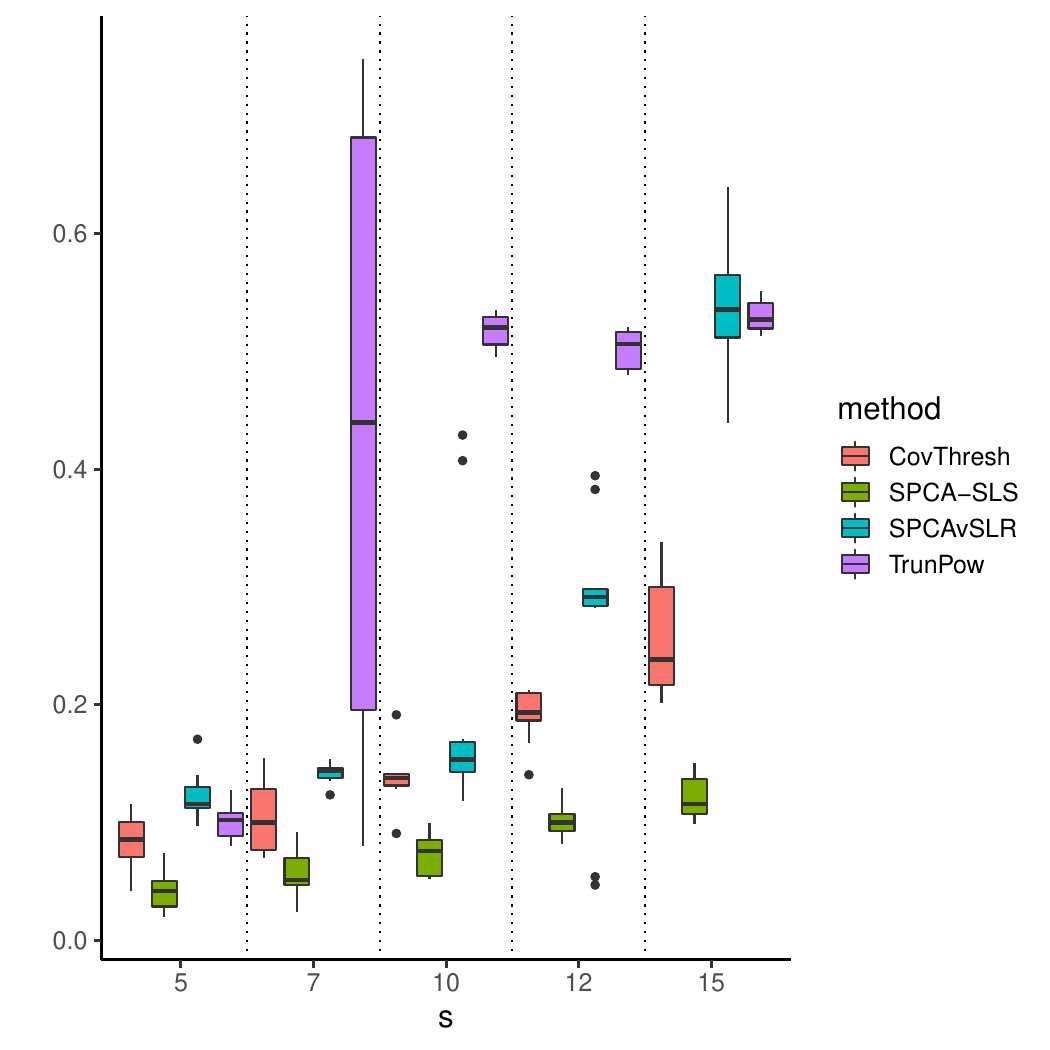} & \rotatebox{90}{~~~~~~~~~~~~~~~~$(FP+FN)/2$}& \includegraphics[width=0.4\linewidth,trim=1cm 0cm 0cm 0cm, clip]{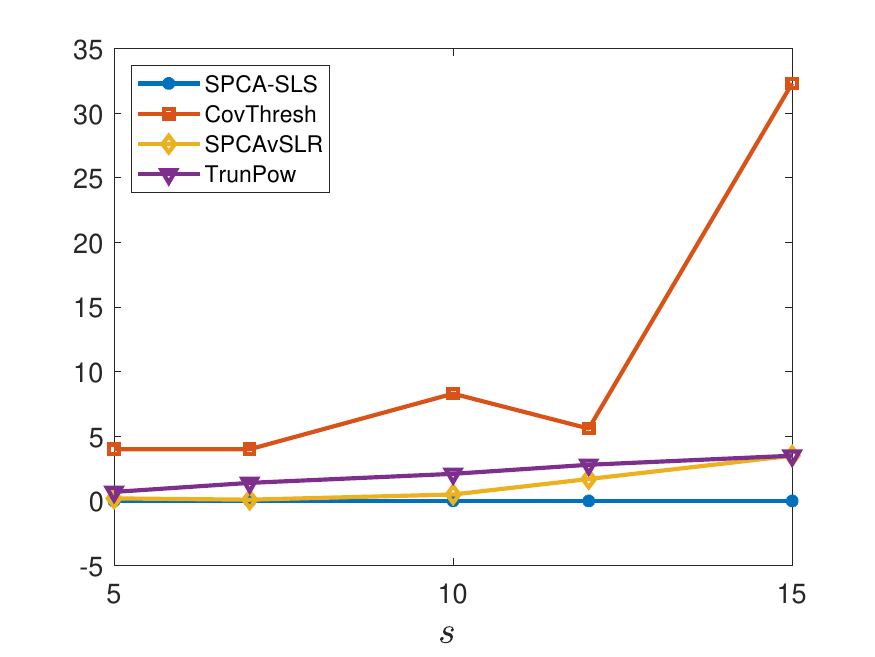}
\end{tabular}
        \caption{\small Numerical results for the synthetic dataset with $p=10000,n=7500$ in Section \ref{varyings}. The left panel shows the estimation performance and the right panel shows the support recovery performance.}
        \label{fig:synt4}
\end{figure*}

\subsubsection{A large-scale example with $p=20,000$}\label{miqpvsmisocp}
To show the scalability of our algorithm, we consider instances with $p=20,000$. We perform a set of experiments on a personal computer. We use a machine equipped with AMD Ryzen 9 5900X CPU $@$ 3.70GHz, using 32GB of RAM. However, in our experiments, Julia did not use more than 12GB of RAM. We set $p=20,000, s\in\{5,10\}$ and  $n\in\{5000,10000,15000,20000\}$. The runtime of our methods is limited to 5 minutes. 

In this case, we explore Algorithm~\ref{outer} applied to two cases of Problem~\eqref{perspectivereform}: (a) we consider the unregularized case $\lambda=0$ denoted by {\texttt{SPCA-SLS}} and (b) $\lambda>0$ chosen as in Section~\ref{vsGurobi}, denoted by \texttt{SPCA-SLSR}.
We use (b) to demonstrate the effect of the perspective regularization in~\eqref{perspectivereform}.
The results for this case are shown in Figure~\ref{fig:synt6}, which compares the estimation performance of different algorithms. As it can be seen, our method leads to better estimation performance compared to other SPCA methods. Moreover, having $\lambda>0$ (in case of {\texttt{SPCA-SLSR}}) does not negatively affect the estimation performance of our method, while it helps to achieve stronger optimality certificates\footnote{For $s=5$, \texttt{SPCA-SLS} and \texttt{SPCA-SLSR} provide optimality gaps under $9\%$ and $3\%$, respectively. For $s=10$, \texttt{SPCA-SLS} and \texttt{SPCA-SLSR} provide optimality gaps under $18\%$ and $6\%$, respectively.}.

\begin{figure*}[t!]
     \centering
     \begin{tabular}{lcc}
&  $s=5$ & $s=10$ \\
     \rotatebox{90}{~~~~~~~~~~~~~~~~~~~~~~~~~~~~~~~~~$|\sin\angle(\hat{u},u^*)|$}& \includegraphics[width=0.44\textwidth,trim=.5cm 0cm 0cm 0cm, clip]{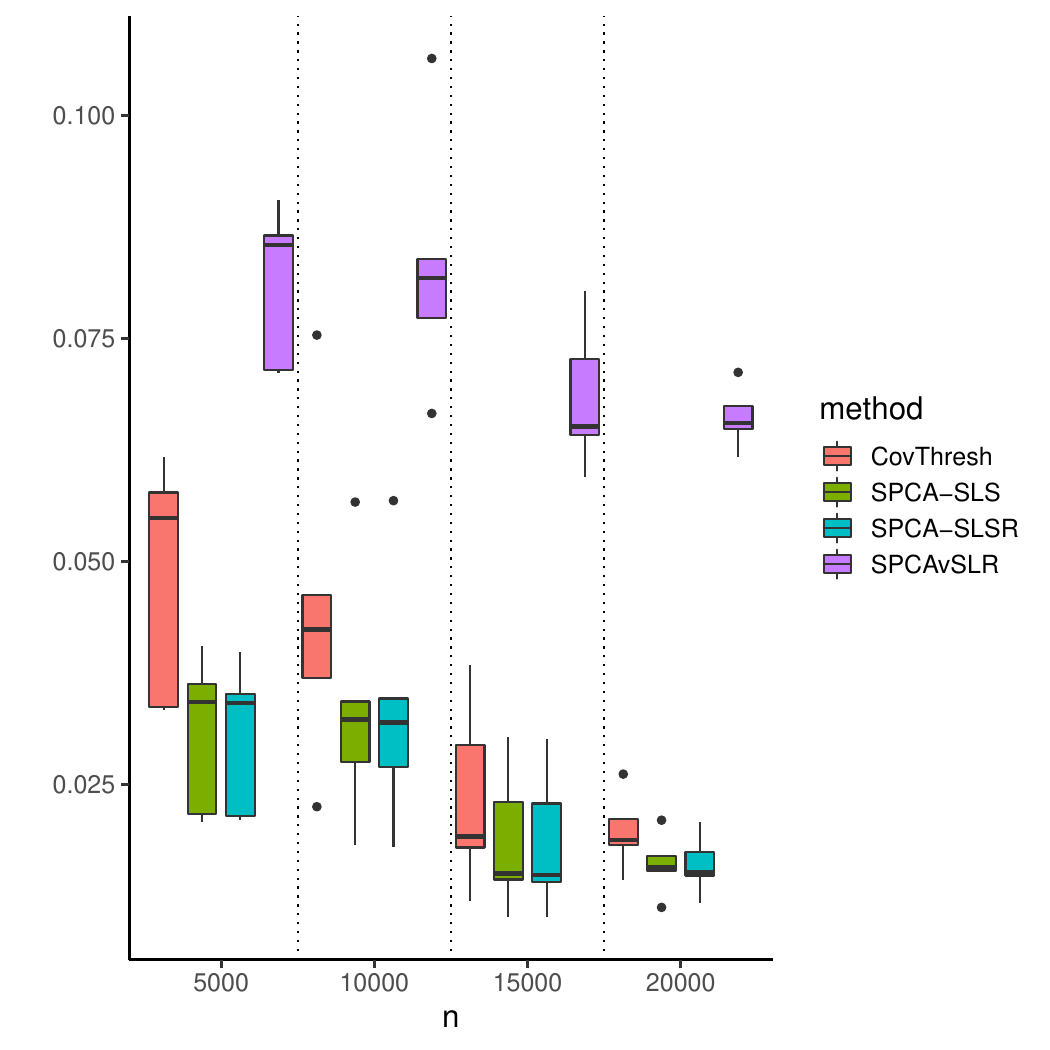}& 
 \includegraphics[width=0.44\textwidth,trim=.5cm 0cm 0cm 0cm, clip]{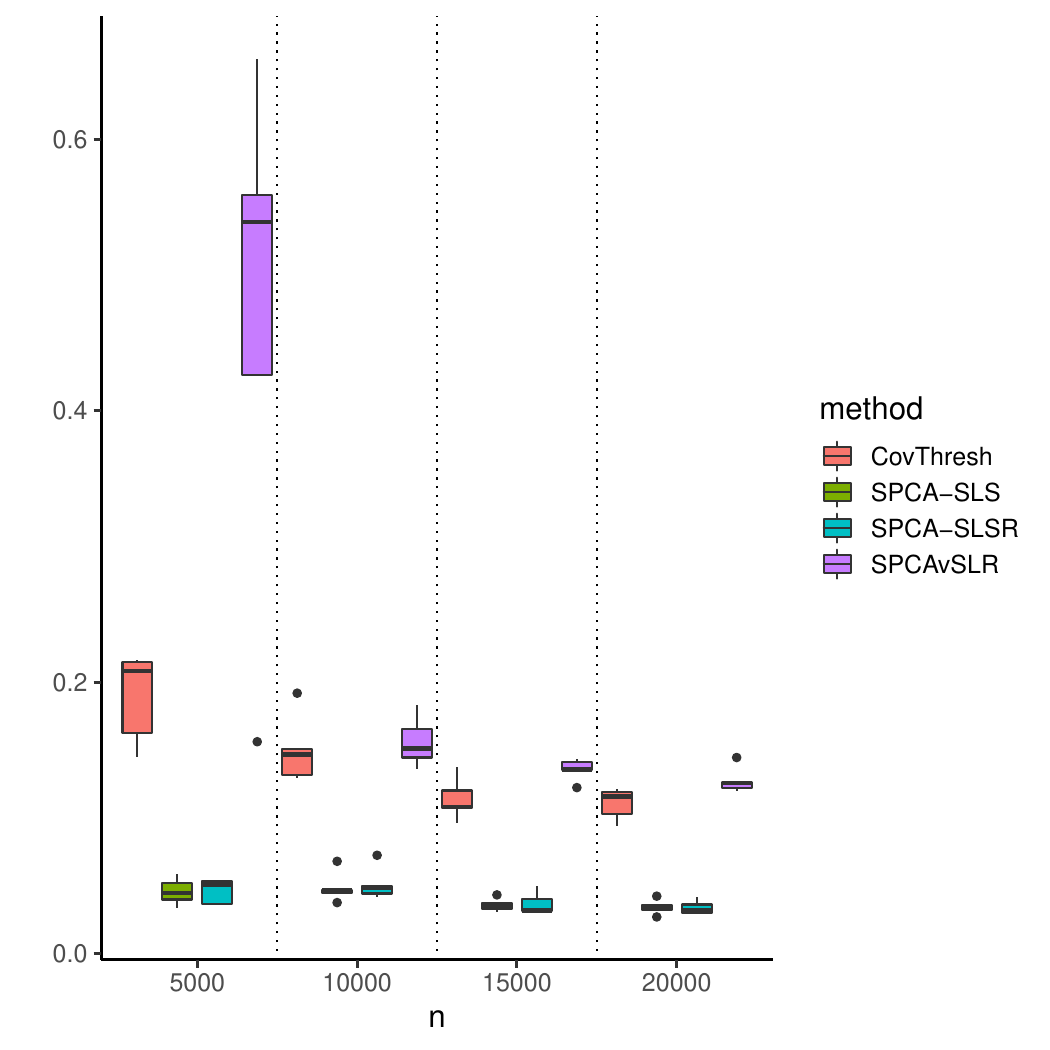}\\ 
\end{tabular}
        \caption{\small Numerical results for the synthetic dataset with $p=20000$ in Section \ref{miqpvsmisocp}. The left panel shows the statistical performance for $s=5$. The right panel shows the statistical performance for $s=10$.}
        \label{fig:synt6}
\end{figure*}

\smallskip

\noindent \textbf{Additional Numerical Results:} We present additional experiments studying how different methods perform under departures from the spiked covariance model --- see Appendix~\ref{app:add-exp}).

\subsection{Real dataset example}\label{real}
To show the scalability of our algorithm and to compare them with other well-known algorithms, we do further numerical experiments on the Gisette dataset \citep{NIPS2004_5e751896}. This dataset is a handwritten digit recognition dataset with $n=6000$ samples and $p=5000$ features. We run {\texttt{SPCA-SLS}}, {\texttt{SPCA-SLSR}} (with $\lambda$ chosen as in Section~\ref{vsGurobi}), {\texttt{TrunPow}}, {\texttt{SPCAvSLR}} and \texttt{CovTresh} for this data matrix and two values of $s\in\{4,5\}$. Parameter selection is done similar to the synthetic data experiments in Section~\ref{synthetic}. We use the same personal computer from Section~\ref{miqpvsmisocp} with Julia process limited to 6GB of RAM usage. To compare the results of different methods, we plot the absolute values of estimated PCs and examine the sparsity pattern of the associated eigenvector. Figure~\ref{fig:real} shows the results for two values of $s$ --- interestingly, we observe that the sparsity patterns available from different algorithms are different, with \texttt{CovTresh} leading to a denser support.

\begin{figure*}[t!]
     \centering
     \begin{tabular}{cc}
  $s=4$ &  $s=5$ \\
      \includegraphics[width=0.42\linewidth]{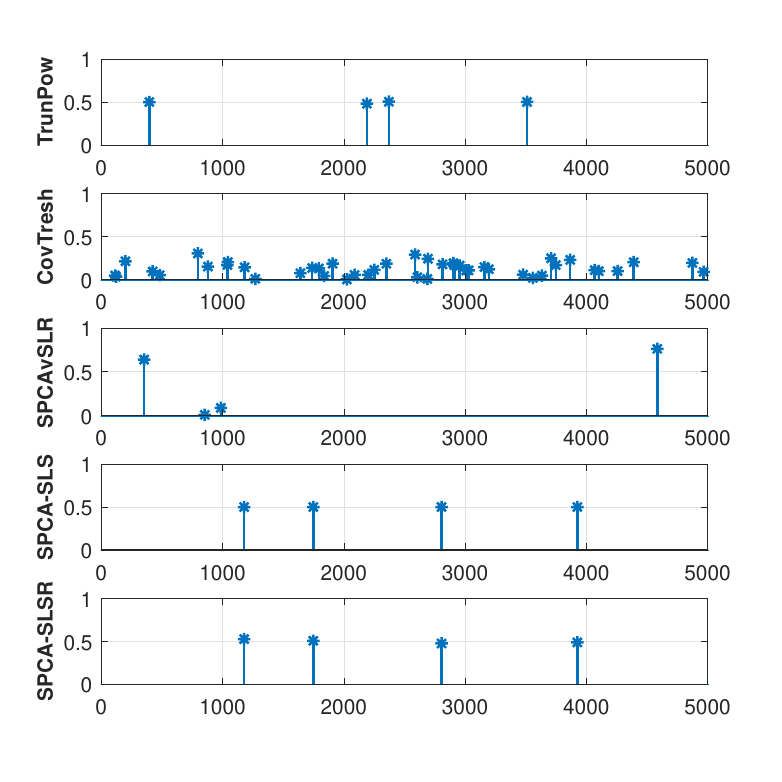}& 
       \includegraphics[width=0.42\linewidth]{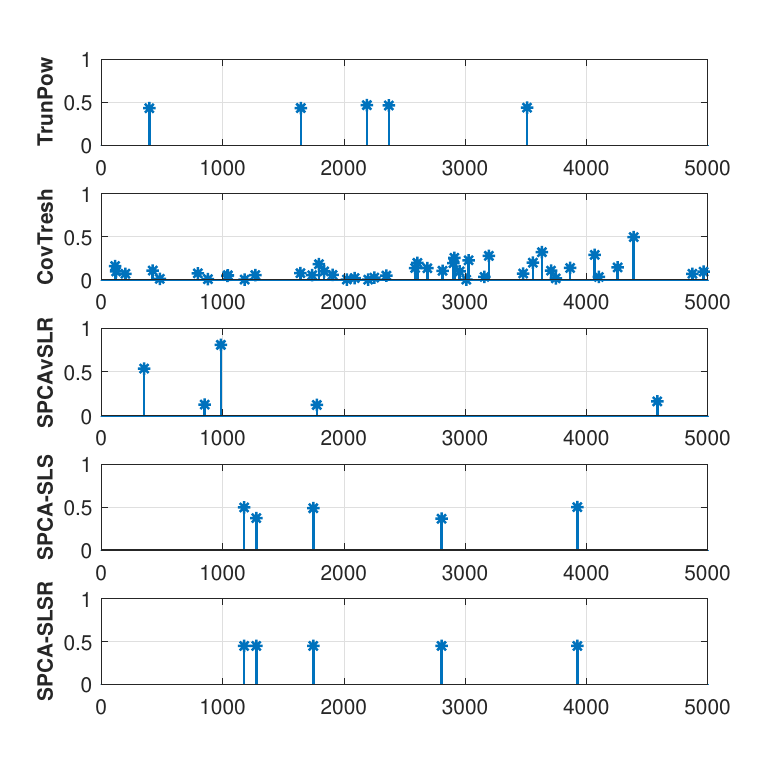}\\ 
\end{tabular}
        \caption{\small Comparison of estimated PCs by different algorithms for the real dataset in Section \ref{real}. The left panel shows the results for $s=4$ and the right panel shows the results for $s=5$. In both cases, {\texttt{SPCA-SLS}} reaches the optimality gap of less than $12\%$ and {\texttt{SPCA-SLSR}} reaches the optimality gap of less than $10\%$ after 10 minutes.}
        \label{fig:real}
\end{figure*}

\section{Conclusion}\label{sec:conclusion}
In this paper, we consider a discrete optimization-based approach for the SPCA problem under the spiked covariance model. We present a novel estimator given by the solution to a mixed integer second order conic optimization formulation. Different from prior work on MIP-based approaches for the SPCA problem, our formulation is based on the properties of the statistical model, which leads to notable improvements in computational performance. 
We analyze the statistical properties of our estimator: both estimation and variable selection properties. 
We present custom cutting plane algorithms for our optimization formulation that can obtain (near) optimal solutions for large problem instances with $p \approx 20000$ features.
Our custom algorithms offer significant improvements in terms of runtime, scalability over off-the-shelf commercial MIP-solvers, and recently proposed MIP-based approaches for the SPCA problem. 
In addition, our numerical experiments appear to suggest that our method outperforms some well-known polynomial-time and heuristic algorithms for SPCA in terms of estimation error and support recovery.

Our statistical model~\eqref{Gcov-general} assumes that the data is drawn from a multivariate Gaussian distribution. We follow prior work studying the SPCA problem under the (sparse) spiked covariance model~\citep{amini2009,bresler-10.5555/3327546.3327752}. 
It is of interest to study the
theoretical guarantees for our proposed estimator or its modifications in the case of non-Gaussian data. In Appendix~\ref{supp:nongaussian}, we explore the numerical performance of our method under two settings when the data is non-Gaussian. Overall, we see that as long as the underlying distribution does not deviate too far from the multivariate Gaussian setting, our method is able to deliver good estimation performance, compared to the alternatives. A deeper investigation for non-Gaussian settings is left as future work.

\subsection*{Acknowledgement} This research was partially supported by grants from the Office of Naval Research (ONR-N000141812298, N000142112841, N000142212665) and
the National Science Foundation (NSF-IIS-1718258).   
Kayhan Behdin contributed to this work while he was a PhD student at MIT.

\bibliographystyle{plainnat} % Style BST file (imsart-number.bst or imsart-nameyear.bst)
\bibliography{ref.bib}

\clearpage

\appendix
\numberwithin{equation}{section}
\numberwithin{lem}{section}
\numberwithin{pro}{section}
\numberwithin{thm}{section}
\numberwithin{figure}{section}

\section*{Additional Numerical Experiments, Proofs and Technical Details}
\section{Studying $B^*$ under deviations from the spiked covariance model}\label{bstar-invest}
\subsection{A general perturbation model}
In this section, we discuss how the leading left singular vector of $B^*$ defined in~\eqref{new-bstar-def} relates to the PC of interest when $G^*$ is a perturbation of a spiked covariance matrix. That is,
\begin{equation}\label{app-gcov}
    G^* = \bar{G} + \Delta
\end{equation}
where
\begin{equation}\label{app-gcov1}
    \bar{G}=I_p+\theta\bar{u}\bar{u}^T\in\mathcal{G}_{p,\theta}
\end{equation}
follows the spiked covariance model for some $\theta\in(0,1]$. We will show that if $\|\Delta\|_F$ is sufficiently small, the leading left singular vector of $B^*$ is close to $\bar{u}$. In this section, we use the notation 
$$\bar{\Theta}=\bar{G}^{-1},~~\text{and}~~~\Theta^*=(G^*)^{-1}.$$ 
Also, we denote the regression coefficients and variances from Lemma~\ref{lem_reg} for $G^*$ and $\bar{G}$ as $\{\beta_{i,j}^*,(\sigma_i^*)^2\}$ and $\{\bar{\beta}_{i,j},\bar{\sigma}_i^2\}$ respectively. We also use these coefficients to form matrices $B^*,\bar{B}$ as in~\eqref{new-bstar-def}:
\begin{equation}\label{new-bstar-def-app}
  B^*_{i,j} = \begin{cases} \beta^*_{i,j} &\text{if}~ i \neq j \\
 (\sigma^*_i)^2-1 & \text{if}~ i = j \end{cases}~~~\text{and}~~~\bar{B}_{i,j} = \begin{cases} \bar{\beta}_{i,j} &\text{if}~ i \neq j \\
 \bar{\sigma}_i^2-1 & \text{if}~ i = j. \end{cases}   
\end{equation}

In this section, we make the following assumption.
\begin{asu*} We have the following:
\begin{enumerate}[leftmargin=*,label=\textbf{(B\arabic*})]
    \item \label{thetastarinverse} The matrix $G^*$ is positive definite and invertible.
    \item \label{b-sigmal} There exist absolute numerical constants $u_{\sigma},l_{\sigma}$ such that $u_{\sigma}\geq \sigma^*_j,\bar{\sigma}_j\geq l_{\sigma}$  for $j\in[p]$. 
\end{enumerate}
\end{asu*}
Assumption~\ref{thetastarinverse} ensures the matrix $\Theta^*$ and the values $\{\bar{\beta}_{i,j},\bar{\sigma}_i\}$ are well-defined. Assumption~\ref{b-sigmal} ensures that the matrices $B^*,\bar{B}$ are well-scaled. Our main result in this section is Theorem~\ref{app-thm}.

\begin{thm}\label{app-thm}
Suppose Assumptions~\ref{thetastarinverse} and~\ref{b-sigmal} hold. Then, we have the following:
\begin{enumerate}
    \item The matrices $B^*,\bar{B}$ satisfy
    $$\|B^*-\bar{B}\|_F^2\lesssim\|\Theta^*\|_F^2\|\bar{\Theta}\|_F^2\|\Delta\|_F^2.$$
    \item For some absolute constant $c_{gap}>0$, the two largest singular values of $B^*$ satisfy
    $$\sigma_{\max}(B^*)-\sigma_{\max,2}(B^*)\geq \theta/2 -c_{gap}\|\Theta^*\|_F\|\bar{\Theta}\|_F\|\Delta\|_F.$$
    \item The leading left singular vector of $B^*$ which is $u^*$, satisfies
    $$|\sin\angle(u^*,\bar{u})|\lesssim \frac{\|\Theta^*\|_F\|\bar{\Theta}\|_F\|\Delta\|_F}{\theta}.$$
\end{enumerate}
\end{thm}
Before presenting the proof of Theorem~\ref{app-thm}, we discuss its implications for our theoretical results in Section~\ref{statreg}.  
The first part of Theorem~\ref{app-thm} states that as long as $\|\Delta\|_F$ is sufficiently small, the matrices $B^*$ and $\bar{B}$ are close. This ensures Assumption~\ref{orsclebetaassumtion} and~\ref{sum_assum} (the assumptions that the oracle error is small) can be satisfied when the amount of perturbation is small. The second part of Theorem~\ref{app-thm} shows a gap in singular values for $B^*$ when $\|\Delta\|_F$ is sufficiently small, hence, Corollary~\ref{bcro} can be readily applied to the model. This implies that a high-quality estimate of $u^*$, the leading left singular vector of $B^*$ can be achieved by our framework. Finally, the last part of Theorem~\ref{app-thm} shows that $u^*$ is close to $\bar{u}$, which is the underlying PC of interest under the model given by~\eqref{app-gcov} and~\eqref{app-gcov1}. Therefore, $\hat{u}$, the leading left singular vector of $\hat{B}$ (as defined in Corollary~\ref{bcro}) will be a good approximation to $\bar{u}$.
We now present the proof of Theorem~\ref{app-thm}.
\begin{proof}[\bf Proof of Theorem~\ref{app-thm}.]
\textbf{Part 1)} Recall that we denote the regression coefficients and variances from Lemma~\ref{lem_reg} for $G^*$ and $\bar{G}$ as $\{\beta_{i,j}^*,(\sigma_i^*)^2\}$ and $\{\bar{\beta}_{i,j},\bar{\sigma}_i^2\}$ respectively. 
By the definitions of $\sigma_j^*,\bar{\sigma}_j$, one has
\begin{align}
    ((\sigma_j^*)^2-\bar{\sigma}_j^2)^2 & = \left(\frac{1}{\Theta^*_{j,j}}-\frac{1}{\bar{\Theta}_{j,j}}\right)^2 \nonumber \\
    & = \frac{(\bar{\Theta}_{j,j}-\Theta_{j,j}^*)^2}{(\Theta^*_{j,j}\bar{\Theta}_{j,j})^2} \nonumber \\
    & \stackrel{(a)}{\leq} {(\bar{\Theta}_{j,j}-\Theta_{j,j}^*)^2}u_{\sigma}^4\nonumber \\
    & \lesssim (\bar{\Theta}_{j,j}-\Theta_{j,j}^*)^2\label{sigma-theta}
\end{align}
where $(a)$ is by Assumption~\ref{b-sigmal}. 

Next, by the definition of $\beta^*_{i,j},\bar{\beta}_{i,j}$ 
for $j\in[p]$ we have
\begin{align}
   \sum_{i:i\neq j} (\beta^*_{i,j}-\bar{\beta}_{i,j})^2 & = \sum_{i:i\neq j} \left(\Theta^*_{i,j}(\sigma^*_j)^2-\bar{\Theta}_{i,j}\bar{\sigma}_j^2\right)^2\nonumber \\
    & = \sum_{i:i\neq j} \left[(\sigma^*_j)^2(\Theta^*_{i,j}-\bar{\Theta}_{i,j})+\bar{\Theta}_{i,j}((\sigma_j^*)^2-\bar{\sigma}_j^2)\right]^2 \nonumber \\
    & \lesssim \sum_{i:i\neq j} \left\{(\sigma_j^*)^4 (\Theta^*_{i,j}-\bar{\Theta}_{i,j})^2 + \bar{\Theta}_{i,j}^2 ((\sigma_j^*)^2-\bar{\sigma}_j^2)^2\right\} \nonumber \\
    & \stackrel{(a)}{\lesssim} \sum_{i:i\neq j} (\Theta^*_{i,j}-\bar{\Theta}_{i,j})^2 + ((\sigma_j^*)^2-\bar{\sigma}_j^2)^2 \sum_{i:i\neq j} \bar{\Theta}_{i,j}^2 \nonumber \\
    & \stackrel{(b)}{\lesssim} \sum_{i:i\neq j} (\Theta^*_{i,j}-\bar{\Theta}_{i,j})^2 + ((\sigma_j^*)^2-\bar{\sigma}_j^2)^2\nonumber \\
    & \stackrel{(c)}{\lesssim} \sum_{i=1}^p (\Theta^*_{i,j}-\bar{\Theta}_{i,j})^2
\end{align}
where $(a)$ is by Assumption~\ref{b-sigmal}, $(b)$ is true as from Lemma~\ref{Ginv} we have
\begin{align*}
    \sum_{i:i\neq j}\bar{\Theta}_{i,j}^2 & = \sum_{i:i\neq j} \frac{\theta^2 \bar{u}_{i}^2\bar{u}_j^2}{(1+\theta)^2}\leq  \sum_{i:i\neq j}\bar{u}_i^2 \leq \sum_{i=1}^p\bar{u}_i^2 = 1
\end{align*}
and $(c)$ is by~\eqref{sigma-theta}.  As a result,
\begin{align}
    \|B^*-\bar{B}\|_F^2 &= \sum_{j=1}^p  \sum_{i:i\neq j} (\beta^*_{i,j}-\bar{\beta}_{i,j})^2 + \sum_{j=1}^p ((\sigma_j^*)^2-\bar{\sigma}_j^2)^2 \nonumber \\
    & \lesssim \sum_{j=1}^p\sum_{i=1}^p (\Theta^*_{i,j}-\bar{\Theta}_{i,j})^2 = \|\Theta^*-\bar{\Theta}\|_F^2.\label{btotheta}
\end{align}
Next, note that
\begin{align}
    I_p & = I_p  + \bar{G}^{-1}\Delta - \bar{G}^{-1}\Delta \nonumber \\
    & = \bar{G}^{-1}(\bar{G}+\Delta)- \bar{G}^{-1}\Delta (\bar{G}+\Delta)^{-1}(\bar{G}+\Delta) \nonumber \\
    & = (\bar{G}^{-1}-\bar{G}^{-1}\Delta (\bar{G}+\Delta)^{-1})(\bar{G}+\Delta)\nonumber
\end{align}
which shows 
\begin{equation}
    \Theta^*= \bar{\Theta}-\bar{\Theta}\Delta \Theta^*
\end{equation}
as $G^*=\bar{G}+\Delta$. Hence, 
\begin{equation}\label{pert-norm}
    \|\Theta^*-\bar{\Theta}\|_F^2 \leq \|\Theta^*\|_F^2\|\bar{\Theta}\|_F^2\|\Delta\|_F^2.
\end{equation}
Equations~\eqref{pert-norm} and~\eqref{btotheta} together complete the proof.\\

\noindent \textbf{Part 2)} From the first part of the theorem and Weyl's inequality,
\begin{equation}
    \begin{aligned}
      \sigma_{\max}(B^*) & \geq \sigma_{\max}(\bar{B}) - \|B^*-\bar{B}\|_F \geq \sigma_{\max}(\bar{B}) - c_{gap}\|\Theta^*\|_F\|\bar{\Theta}\|_F\|\Delta\|_F/2 \\
      \sigma_{\max,2}(B^*) & \leq \sigma_{\max,2}(\bar{B}) + \|B^*-\bar{B}\|_F \leq \sigma_{\max,2}(\bar{B}) + c_{gap}\|\Theta^*\|_F\|\bar{\Theta}\|_F\|\Delta\|_F/2
    \end{aligned}
\end{equation}
where $c_{gap}$ is an absolute constant. Moreover, as $\bar{G}$ follows a spiked covariance model, from~\eqref{Bsigma}, we have $\sigma_{\max}(\bar{B})-\sigma_{\max,2}(\bar{B})\geq \theta/2$ which completes the proof. \\

\noindent \textbf{Part 3)} This follows by using Wedin's theorem~\cite{wedin1972perturbation}, part 1 of the theorem and the fact that $\sigma_{\max}(\bar{B})-\sigma_{\max,2}(\bar{B})\geq \theta/2$.
\end{proof}
\subsection{Tighter results under a structured perturbation}\label{appendix-a2-misspecification}
Here we analyze a special perturbation structure where $\Delta$ in~\eqref{app-gcov} is of a specific form. While the general results from the previous section are applicable to this model, we can obtain tighter results for the specific form of $\Delta$ we discuss here. 
We show that when the amount of perturbation is small, the left singular vector of $B^*$ satisfies $u^*=\bar{u}$ under the model described by~\eqref{app-gcov} and~\eqref{app-gcov1}. In other words, for this structure of $\Delta$, no additional error arises from the perturbation. We build our model based on the notation from the last section. Specifically, for simplicity, we assume $\|\bar{u}\|_0=s$ and $\bar{u}_{s+1}=\cdots=\bar{u}_p=0$, i.e., only the first $s$ coordinates of $\bar{u}$ are nonzero. Moreover, $\Delta$ has a block diagonal structure 
\begin{equation}
    \Delta = \begin{bmatrix} 0 & 0 \\ 0 & \Delta_0
    \end{bmatrix}
\end{equation}
where $\Delta_0\in\R^{(p-s)\times(p-s)}$. This implies $G^*$ is given by
\begin{equation}
    G^* = \begin{bmatrix}
    I_s + \theta \bar{u}_{1:s}\bar{u}_{1:s}^T & 0 \\ 0 & I_{p-s}+\Delta_0
    \end{bmatrix}.
\end{equation}
This model is inspired by the perturbation model considered in~\cite{amini2009}. 
While~\cite{amini2009} assumes that the nonzero entries of $\bar{u}$ are $\pm1/\sqrt{s}$, we do not make such an assumption, and hence ours is a more general form for $G^*$.

Let us define 
$$\bar{G}_0=I_s + \theta \bar{u}_{1:s}\bar{u}_{1:s}^T\in\R^{s\times s}~~~~\text{and}~~~~G_0=I_{p-s}+\Delta_0\in\R^{(p-s)\times(p-s)}.$$ Note that $\bar{G}_0,G_0$ can be considered as covariance matrices but they have different dimensions. 
Similar to the previous section, we define  the regression coefficients and variances from Lemma~\ref{lem_reg} for $G_0$ as $\{(\beta_{0})_{i,j},(\sigma_0)_i^2\}$. For these regression coefficients and variances, we define $B_0\in\R^{(p-s)\times(p-s)}$ similar to~\eqref{new-bstar-def-app},
\begin{equation}\label{new-bstar-def-app2}
  (B_0)_{i,j} = \begin{cases} (\beta_0)_{i,j} &\text{if}~ i \neq j \\
 (\sigma_0)_i^2-1 & \text{if}~ i = j. \end{cases}
\end{equation}
We define $\bar{B}_0\in\R^{s\times s}$ for $\bar G_0$ in a similar fashion. We make the following assumptions:

\begin{asu*} We have the following:
\begin{enumerate}[leftmargin=*,label=\textbf{(C\arabic*})]
    \item \label{g0inverse} The matrix $G_0$ is positive definite and invertible.
    \item \label{b0-sigma} The matrix $B_0$ satisfies $\sigma_{\max}(B_0)\leq \theta/10$.
\end{enumerate}
\end{asu*}
Our main result in this section is as follows.
\begin{thm}\label{app-thm2}
Under the model setup discussed, if Assumptions~\ref{g0inverse} and~\ref{b0-sigma} hold we have the following.
\begin{enumerate}
    \item The matrix $B^*$ satisfies
    $$\|B^*-\bar{B}\|_F = \|B_0\|_F.$$
    \item The leading left singular vector of $B^*$ is $\bar{u}$.
    \item The top two singular values of the matrix $B^*$ satisfy
    $$\sigma_{\max}(B^*)-\sigma_{\max,2}(B^*)\gtrsim \theta.$$
\end{enumerate}
\end{thm}
Theorem~\ref{app-thm2} shows that as long as $\|B_0\|_F$ is sufficiently small (as in Assumption~\ref{b0-sigma}), we have $u^*=\bar{u}$ and also $B^*$ has a nonzero singular value gap. This means the perturbation does not result in any additional error when estimating $\bar{u}$ through the leading left singular vector of $B^*$.

\begin{proof}[\bf Proof of Theorem~\ref{app-thm2}.]\textbf{Part 1)} 
As $G^*$ has a block diagonal structure, we have
\begin{equation}
    \Theta^* = (G^*)^{-1}=\begin{bmatrix} \bar{G}_0^{-1} & 0 \\ 0 & G_0^{-1}
    \end{bmatrix}.
\end{equation}
We compute the entries of $\{\beta_{i,j}^*\}$. 
Note that $B^*$ has a block diagonal structure: this follows by observing that $\beta_{i,j}^*=\Theta_{i,j}^*=0$ if $i>s,j\leq s$ or $i\leq s, j>s$. 
If $i,j\leq s$, then
\begin{equation}
    \beta_{i,j}^* =-\frac{\Theta^*_{i,j}}{\Theta^*_{j,j}}=-\frac{(\bar{G}_0^{-1})_{i,j}}{(\bar{G}_0^{-1})_{j,j}}.
\end{equation}
Moreover, if $i,j>s$, then 
\begin{equation}
    \beta_{i,j}^* =-\frac{\Theta^*_{i,j}}{\Theta^*_{j,j}}=-\frac{({G}_0^{-1})_{i-s,j-s}}{({G}_0^{-1})_{j-s,j-s}}.
\end{equation}
As a result, we have
\begin{equation}\label{app-bstar-diag}
    B^* =\begin{bmatrix} \bar{B}_0 & 0 \\ 0 & B_0
    \end{bmatrix}.
\end{equation}
This completes the proof of the first part.\\

\noindent \textbf{Part 2)} Suppose we have the singular value decompositions $\bar{B}_0=\bar{U}_0\bar{S}_0\bar{V}_0^T$ and ${B}_0={U}_0{S}_0{V}_0^T$. By~\eqref{app-bstar-diag}, we have that $B^*$ has the singular value decomposition: 
\begin{equation}
    B^* =\begin{bmatrix} \bar{U}_0 & 0 \\ 0 & U_0
    \end{bmatrix}\begin{bmatrix} \bar{S}_0 & 0 \\ 0 & S_0
    \end{bmatrix}\begin{bmatrix} \bar{V}^T_0 & 0 \\ 0 & V^T_0
    \end{bmatrix}.
\end{equation}
Note that as $\bar{G}_0$ follows a ($s$-dimensional) spiked covariance model, $\bar{B}_0$ is a rank-1 matrix and its nonzero singular value is at least $\theta/2$ from~\eqref{Bsigma}. As the rest of singular values of $\bar{B}_0$ are zero, the second largest singular value of $B^*$ is the largest singular value of $B_0$ by Assumption~\ref{b0-sigma}. 
As a result, the leading left singular vector of $B^*$ is the leading left singular vector of $\bar{B}_0$ concatenated with zeros. On the other hand, from~\eqref{bstardef} we know that the leading left singular vector of $\bar{B}_0$ is $\bar{u}_{1:s}$. This completes the proof of the second part.\\

\noindent \textbf{Part 3)} From the proof of 
part~2, we have
\begin{align}
    \sigma_{\max}(B^*) - \sigma_{\max,2}(B^*) \geq \theta/2 - \theta/10 \gtrsim \theta
\end{align}
where the first inequality is a result of Assumption~\ref{b0-sigma}.
\end{proof}

\section{Additional Numerical Experiments}\label{app:add-exp}
In this section, we present additional numerical experiments to show the utility of our method compared to other methods under deviations from the spiked covariance model. Our experimental setup here is similar to the one in Section~\ref{synthetic}. The main difference in this section is that we assume the generative model has the covariance matrix 
$$G^*=I_p+u^*(u^*)^T+\underbrace{\frac{u_1u_1^T}{10}+\frac{u_2u_2^T}{20}}_{\Delta}$$
where $\Delta$ is defined in~\eqref{app-gcov}, and $u^*,u_1,u_2\in\R^p$.
The vector $u^*$ is randomly drawn as explained in Section~\ref{synthetic} with $\|u^*\|_0=s$. In our first set of experiments, we set $u_2=0$. We draw each coordinate of $u_1$ from $\text{Unif}[0,1]$. Then, we use the Gram-Schmidt process to make $u_1$ a unit vector orthonormal to $u^*$ (note that $u_1$ and consequently $G^*$ are not sparse anymore). We set $p=5000$ and $s\in\{5,10\}$ and $n\in\{1000,2000,3000,4000,5000\}$. The results for this case are shown in Figure~\ref{fig:synt-th1}. 
We observe that the performance of all methods deteriorate compared to the case when $G^*$ is a spiked covariance model (i.e, $\Delta = 0$).
However, our proposed estimator still appears to work better compared to other methods --- this suggests the usefulness of our estimator even when $G^*$ deviates from a spiked covariance model.

\begin{figure*}[t!]
     \centering
     \begin{tabular}{lcc}
&  $s=5$ & $s=10$ \\
     \rotatebox{90}{~~~~~~~~~~~~~~~~~~~~~~~~~~~~~~~~~$|\sin\angle(\hat{u},u^*)|$}& \includegraphics[width=0.44\textwidth,trim=.5cm 0cm 0cm 0cm, clip]{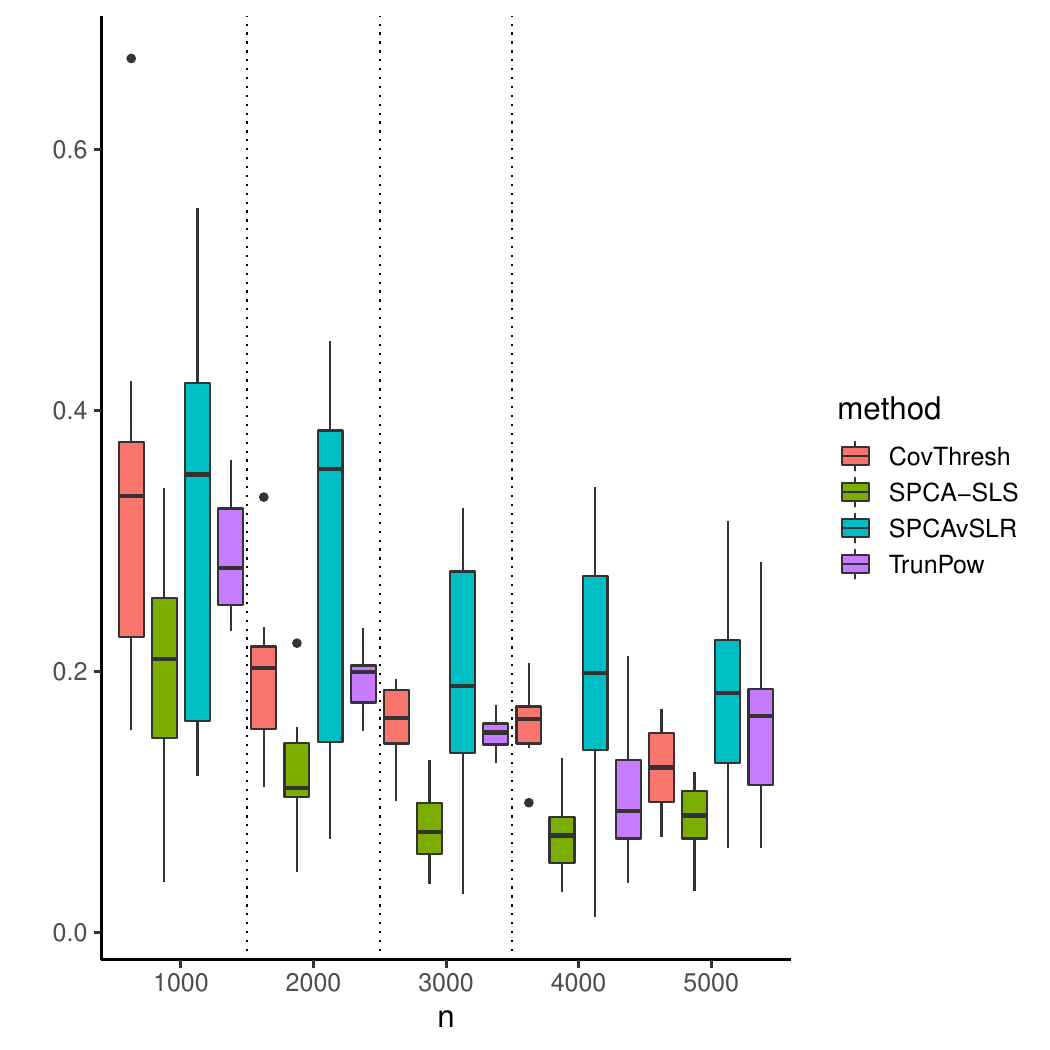}& 
 \includegraphics[width=0.44\textwidth,trim=.5cm 0cm 0cm 0cm, clip]{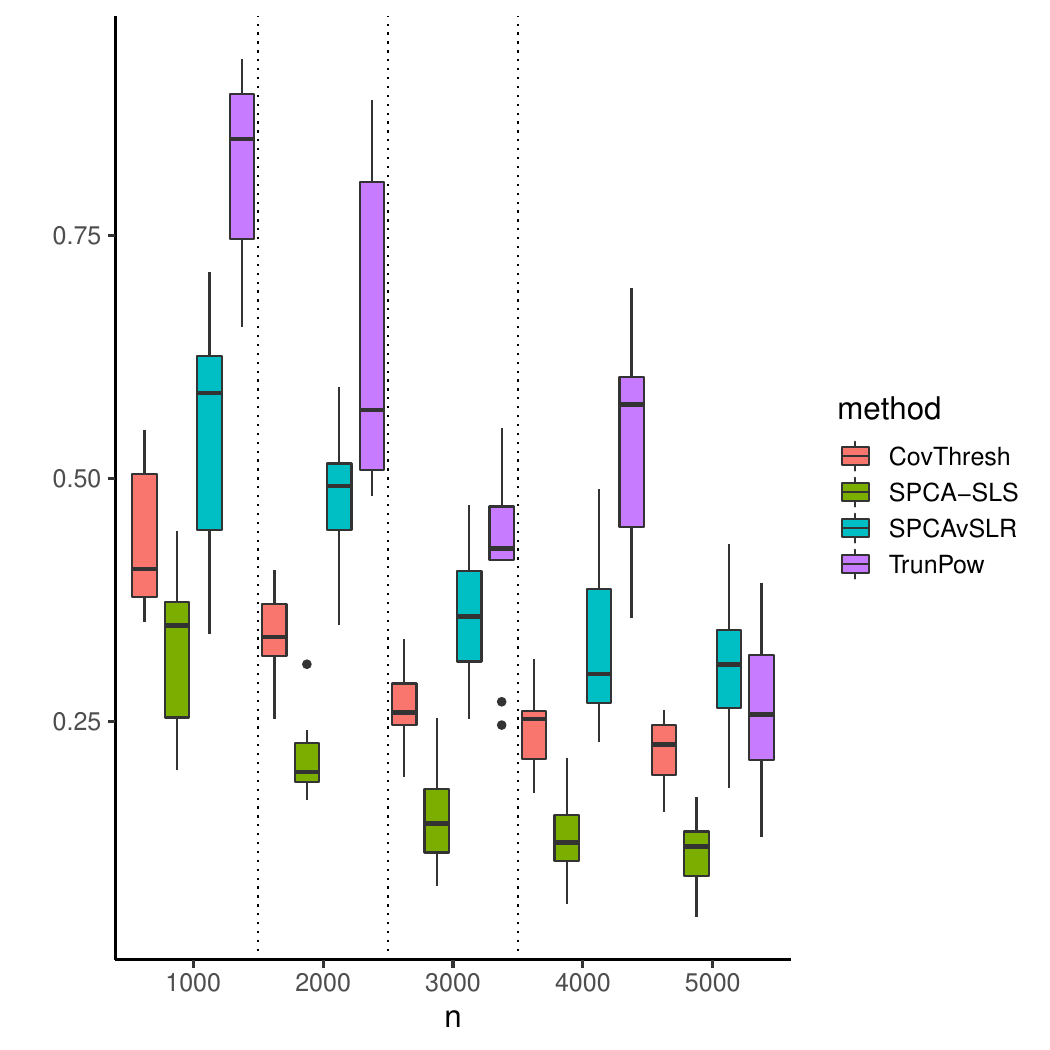}\\ 
\end{tabular}
        \caption{\small Experiment results for the first setup in Appendix~\ref{app:add-exp} }
        \label{fig:synt-th1}
\end{figure*}

In the next set of experiments, we let $u_2$ to be nonzero. Similar to $u_1$, the entries of the vector $u_2$ are drawn from a uniform distribution and then we use the Gram-Schmidt process to form the orthonormal set of vectors $u^*,u_1,u_2$. The results for this case are shown in Figure~\ref{fig:synt-th2}. Overall, we see that our proposed method is still generally outperforming other estimators, although all estimators suffer due to the generative model being misspecified.

\begin{figure*}[t!]
     \centering
     \begin{tabular}{lcc}
&  $s=5$ & $s=10$ \\
     \rotatebox{90}{~~~~~~~~~~~~~~~~~~~~~~~~~~~~~~~~~$|\sin\angle(\hat{u},u^*)|$}& \includegraphics[width=0.44\textwidth,trim=.5cm 0cm 0cm 0cm, clip]{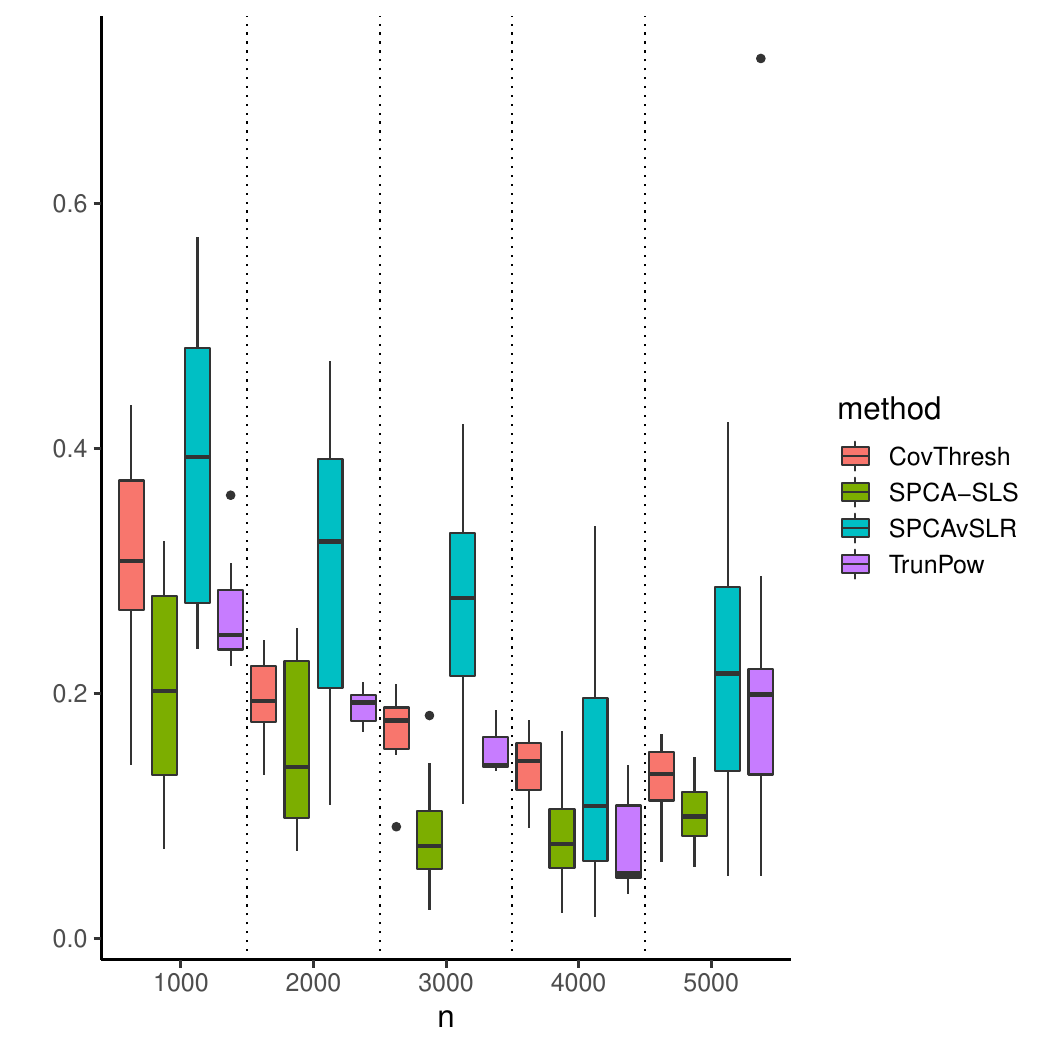}& 
 \includegraphics[width=0.44\textwidth,trim=.5cm 0cm 0cm 0cm, clip]{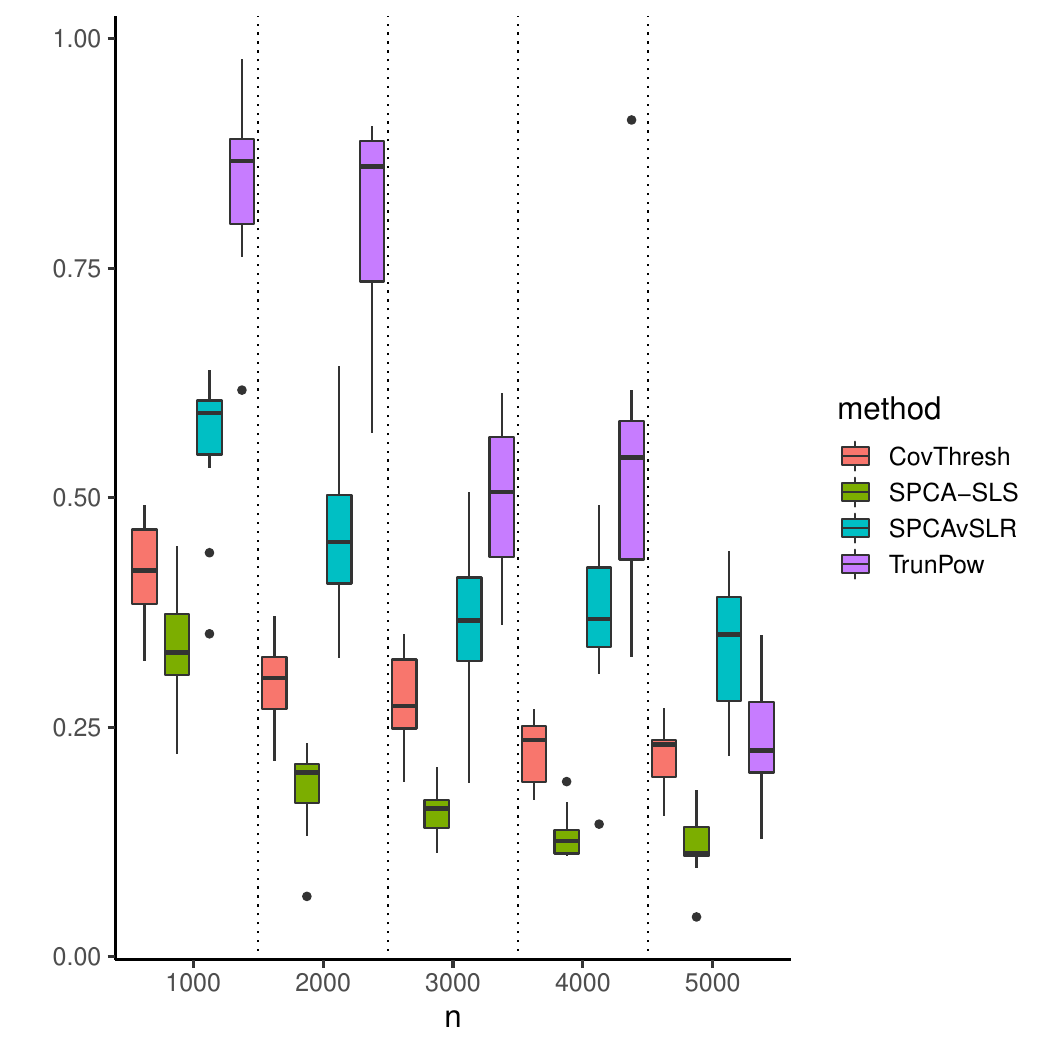}\\ 
\end{tabular}
        \caption{\small Experiment results for the second setup in Appendix~\ref{app:add-exp} }
        \label{fig:synt-th2}
\end{figure*}

\section{Proofs of Theorems~1 and 2}\label{appendix:c}
Before proceeding with the proof of main results, we present an extended version of Lemma~\ref{lem_reg} that we use throughout the proofs. 
\begin{lem}\label{lem_reg_ext}
Let $\varepsilon\in\R^{n\times p}$ be such that for $j\in[p]$,
$$\varepsilon_{:,j}=X_{:,j}-\sum_{i\neq j}\beta^*_{i,j}X_{:,i}.$$
Then, $\varepsilon_{:,j}$ and $\{X_{:,i}\}_{i\neq j}$ are independent for every $j$. Moreover, for every $j$,
$$\varepsilon_{:,j}\sim\mathcal{N}(0,(\sigma_j^*)^2I_n).$$
\end{lem}

\subsection{Proof of Theorem \ref{reg_thm}}
To prove this theorem, we first mention a few basic results that we use throughout the proof.

\begin{lem}[Theorem 1.19, \cite{rigollet2015high}]\label{maximal}
Let $\omega \in\mathbb{R}^p$ be a random vector with $\omega_i\stackrel{\text{iid}}{\sim}\mathcal{N}(0,\sigma^2)$ for $i \in [p]$,  then 
\begin{equation}
    \p(\sup_{{\theta}\in \mathcal{B}(p)}{\theta}^T{\omega}>t)\leq \exp\left(-\frac{t^2}{8\sigma^2}+p\log 5\right),
\end{equation}
where $\mathcal{B}(p)$ denotes the unit Euclidean ball of dimension $p$.
\end{lem}
\begin{lem}\label{minnorm}
Suppose $X\in\mathbb{R}^{n\times p}$ (with $n\geq p$) be a data matrix with independent rows, with the $i$-th row distributed as $\mathcal{N}(0,G)$ (with positive definite $G$) where $0<\underline{\sigma}\leq \sqrt{\lambda_{\min}(G)}\leq \sqrt{\lambda_{\max}(G)}\leq \overline{\sigma}$. The following holds:
$$\underline{\sigma}\left(1-c_0\left(\sqrt{\frac{p}{n}}+\frac{t}{\sqrt{n}}\right)\right)\lesssim \frac{1}{\sqrt{n}}\sigma_{\min}(X)\leq \frac{1}{\sqrt{n}}\sigma_{\max}(X)\lesssim \overline{\sigma}\left(1+c_0\left(\sqrt{\frac{p}{n}}+\frac{t}{\sqrt{n}}\right)\right)$$
with probability at least $\geq 1-\exp(-Ct^2)$ for 
some universal constants $C,c_0>0$.
\end{lem}
\begin{proof}
Let $G = USU^{T}$ be the eigendecomposition of $G$ with diagonal $S$ containing the eigenvalues of $G$. Let $Y=XUS^{-1/2}$. For any $i,j\in[n]$ such that $i\neq j$, we have:
$$\E[Y_{i,:}^TY_{i,:}]=\E[S^{-1/2}U^TX_{i,:}^TX_{i,:}US^{-1/2}]=S^{-1/2}U^TG US^{-1/2}=I_p,$$
$$\text{and}~~~~\E[Y_{i,:}^TY_{j,:}]=\E[S^{-1/2}U^TX_{i,:}^TX_{j,:}US^{-1/2}]=0$$
showing that $Y \sim \mathcal{N}(0,{I}_{p})$. Therefore, by Theorem 4.6.1 of \citet{vershynin2018high}, the following holds
\begin{equation}\label{eqn:ineq1}
    1-c_0\left(\sqrt{\frac{p}{n}}+\frac{t}{\sqrt{n}} \right)\lesssim \frac{1}{\sqrt{n}}\sigma_{\min}(Y)\leq \frac{1}{\sqrt{n}}\sigma_{\max}(Y)\lesssim 1+c_0\left(\sqrt{\frac{p}{n}}+\frac{t}{\sqrt{n}}\right)
    \end{equation}
with probability at least $\geq 1-\exp(-Ct^2)$ for some universal constants $C,c_0>0$.
Moreover,
\begin{align*}
   \sigma_{\min}(X) & =\inf_{\|z\|_2=1}\|Xz\|_2 \\
  & =  \inf_{\|z\|_2=1}\|YS^{1/2}U^Tz\|_2\\
  & =  \inf_{\|z\|_2=1}\|YS^{1/2}z\|_2\\
  & = \inf_{\|z\|_2=1}\|S^{1/2}z\|_2\left\Vert Y\frac{S^{1/2}z}{\|S^{1/2}z\|_2}\right\Vert_2\\
  & \geq \sqrt{\lambda_{\min}(G)}~\sigma_{\min}(Y).
\end{align*}
Combining the above with~\eqref{eqn:ineq1}, we arrive at the first inequality in this lemma. 
The proof for the third inequality follows by an argument similar to the above, but for $\sigma_{\max}$.
\end{proof}
\begin{lem}\label{minnorm2}
Suppose the rows of the matrix $X\in\mathbb{R}^{n\times p}$ are iid draws from a multivariate Gaussian distribution $\mathcal{N}(0,G)$. Moreover, suppose for any $S\subseteq [p]$ such that $|S|\leq s$, 
$$\lambda_{\min}(G_{S,S})\geq\kappa^{2}>0.$$
Then, if $n\gtrsim s\log p$, with probability at least $1-s\exp(-10s\log p)$, we have:
$$\sigma_{\min}(X_S)\gtrsim \kappa\sqrt{n}~~~\text{for all $S$ with $|S|\leq s$}.$$
(We recall that $X_S$ is a sub-matrix of $X$ restricted to the columns indexed by $S$).
\end{lem}
\begin{proof}
Suppose $S$ is fixed. Take $t=\sqrt{cs\log(p)}$ for some constant $c>0$. By Lemma \ref{minnorm}, with probability greater than $1-\exp(-Ccs\log(p))$
$$\frac{1}{\sqrt{n}}\sigma_{\min}(X_S)\gtrsim \kappa\left(1-c_0\left(\sqrt{\frac{|S|}{n}}+\sqrt{\frac{cs\log p}{n}}\right)\right)\gtrsim \kappa \left(1-c_0(1+\sqrt{c})\sqrt{\frac{s\log p}{n}}\right).$$
We consider a union bound over all possible sets $S$ (of size at most $s$):
\begin{align*}
    1-\sum_{k\leq s}\sum_{|S|=k}\exp(-Ccs\log(p)) & \geq 1-\sum_{k\leq s}{p \choose k}\exp(-Ccs\log(p)) \\
    & \stackrel{(a)}{\geq}     1-\sum_{k\leq s}\left(\frac{ep}{k}\right)^{k}\exp(-Ccs\log(p)) \\
    & \geq 1-\sum_{k\leq s}\left(ep\right)^{s}\exp(-Ccs\log(p)) \\
    & = 1-s \exp(-Ccs\log p + s\log(ep))
\end{align*}
where $(a)$ is due to the inequality ${p \choose k}\leq \left(\frac{ep}{k}\right)^{k}$.

Therefore, with probability greater than $1-s \exp(-Ccs\log p + s\log(ep))$,
we have
$$\frac{1}{\sqrt{n}}\sigma_{\min}(X_S)\gtrsim  \kappa \left(1-c_0(1+\sqrt{c})\sqrt{\frac{s\log p}{n}}\right)~~~\text{for all $S \subset [p]$ with $|S|\leq s$}.$$
Finally, the result follows by choosing $c$ large enough to have $-Ccs\log p+s\log(ep)<- 10s\log p$ and $n$ such that 
$n>2c_0^2(1+\sqrt{c})^2s\log p.$
\end{proof}
\begin{lem}\label{bernlem}
For fixed $j_1,j_2\in[p]$, and $n>\frac{2}{C_b}\log(1/\delta)$, we have
\begin{equation}\label{bern2}
    \p\left(\left\vert \frac{1}{n}\sum_{i=1}^n X_{i,j_1}X_{i,j_2} -  G^*_{j_1,j_2}\right\vert>c_{\psi}\sqrt{\frac{\log(1/\delta)}{C_b n}}\right) \leq 2\delta
\end{equation}
for some universal constants $C_b,c_{\psi}>0$ (We recall that $G^*$ is the covariance matrix from~\eqref{Gcov-general}).
\end{lem}
\begin{proof}
Suppose $i\in[n]$ and $j_1,j_2\in[p]$. We have $|\E[X_{i,j_1}X_{i,j_2}]|=|G^*_{j_1,j_2}|\leq 2$ by Assumption~\ref{gdtsrassumption} and $\E[X_{i,j_1}]=0$. As a result, by Bernstein's inequality \cite[Theorem 2.8.1]{vershynin2018high}, we have
    \begin{equation}
        \p\left(\left\vert \frac{1}{n}\sum_{i=1}^n X_{i,j_1}X_{i,j_2} -  G^*_{j_1,j_2}\right\vert>t\right)\leq 2\exp\left(-C_bn((t/c_{\psi})^2 \land (t/c_{\psi}))\right)\label{bernstein}
    \end{equation}
    for some universal constants $c_{\psi},C_b>0$. Therefore, if $n> \frac{2}{C_b}\log(\frac{1}{\delta})$ and $t=c_{\psi}\sqrt{\frac{\log(1/\delta)}{C_b n}}$, we have:
    $$\frac{t}{c_{\psi}}=\sqrt{\frac{\log(1/\delta)}{C_bn}}<1$$
which implies $(t/c_{\psi})^2\land(t/c_{\psi})=(t/c_{\psi})^2$.
    As a result, from \eqref{bernstein}, we arrive at~\eqref{bern2}. \end{proof}

\noindent The proof of Theorem \ref{reg_thm} is based on the following technical lemma. 
\begin{lem}\label{thm3techlem}
Under the assumptions of Theorem \ref{reg_thm}, one has
\begin{align*}
    \frac{1}{n}\sum_{j=1}^p \left[\left\Vert \sum_{i\neq j}(\hat{\beta}_{i,j}-\beta_{i,j}^*)X_{:,i} \right\Vert_2^2\right]&\lesssim  \frac{s^2\log(p/s)}{n} + \frac{1}{n}\sum_{j=1}^p\left\Vert\sum_{i\neq j}(\beta^*_{i,j}-\tilde{\beta}_{i,j})X_{:,i}\right\Vert_2^2 
\end{align*}
with probability greater than $1-2s\exp(-10s\log(p/s))$.
\end{lem}
\begin{proof}
Let $\beta_{i,j}^*=-\frac{({(G^*)}^{-1})_{j,i}}{({(G^*)}^{-1})_{j,j}}$ be the true regression coefficients and $\tilde{\beta}_{i,j}=-\frac{(\tilde{G}^{-1})_{j,i}}{(\tilde{G}^{-1})_{j,j}}$ be the oracle regression coefficients defined in~\eqref{tildebeta}.
By taking $z$ to represent the support of $\tilde{u}$, we can see that the matrix $\{\tilde{\beta}_{i,j}\}$ is feasible for Problem~\eqref{reg_multi}.
    Using the optimality of $\hat{\beta}$ and feasibility of $\tilde{\beta}$ for Problem~\eqref{reg_multi}, we have:
\begin{align}
    & \sum_{j=1}^p \left\Vert X_{:,j} -\sum_{i\neq j}\hat{\beta}_{i,j}X_{:,i} \right\Vert_2^2 \leq   \sum_{j=1}^p \left\Vert X_{:,j} -\sum_{i\neq j}\tilde{\beta}_{i,j}X_{:,i} \right\Vert_2^2 \nonumber \\
    \Rightarrow & \sum_{j=1}^p \left[\left\Vert \sum_{i\neq j}(\hat{\beta}_{i,j}-\beta_{i,j}^*)X_{:,i} \right\Vert_2^2-\left\Vert \sum_{i\neq j}(\tilde{\beta}_{i,j}-\beta_{i,j}^*)X_{:,i} \right\Vert_2^2\right]\leq 2\sum_{j=1}^p {\varepsilon}_{:,j}^T\left(\sum_{i\neq j}(\hat{\beta}_{i,j}-\tilde{\beta}_{i,j})X_{:,i}\right) \label{regthmineq1}
\end{align}
where we used the representation $X_{:,j}=\sum_{i\neq j}\beta^*_{i,j}X_{:,i}+\varepsilon_{:,j}$ by Lemma \ref{lem_reg_ext}. 

By using the inequality $2ab\leq 4a^2 + b^2/4$, we have
\begin{align}
    2 {\varepsilon}_{:,j}^T\left(\sum_{i\neq j}(\hat{\beta}_{i,j}-\tilde{\beta}_{i,j}){X}_{:,i}\right) &=2{\varepsilon}_{:,j}^T\frac{\sum_{i\neq j}(\hat{\beta}_{i,j}-\tilde{\beta}_{i,j})X_{:,i}}{\|\sum_{i\neq j}(\hat{\beta}_{i,j}-\tilde{\beta}_{i,j})X_{:,i}\|_2}\|\sum_{i\neq j}(\hat{\beta}_{i,j}-\tilde{\beta}_{i,j})X_{:,i}\|_2  \nonumber\\ 
    & \leq 4\left[{\varepsilon}_{:,j}^T\frac{\sum_{i\neq j}(\hat{\beta}_{i,j}-\tilde{\beta}_{i,j})X_{:,i}}{\|\sum_{i\neq j}(\hat{\beta}_{i,j}-\tilde{\beta}_{i,j})X_{:,i}\|_2}\right]^2+\frac{1}{4}\|\sum_{i\neq j}(\hat{\beta}_{i,j}-\tilde{\beta}_{i,j})X_{:,i}\|_2^2. \label{regthmineq2}
\end{align}
In addition,
\begin{equation}\label{regthmineq3}
    \frac{1}{4}\|\sum_{i\neq j}(\hat{\beta}_{i,j}-\tilde{\beta}_{i,j})X_{:,i}\|_2^2 \leq \frac{1}{2} \|\sum_{i\neq j}(\hat{\beta}_{i,j}-\beta^*_{i,j})X_{:,i}\|_2^2+ \frac{1}{2}\|\sum_{i\neq j}(\beta^*_{i,j}-\tilde{\beta}_{i,j})X_{:,i}\|_2^2.
\end{equation}
Substituting \eqref{regthmineq2} and \eqref{regthmineq3} into \eqref{regthmineq1}, we have:
\begin{align}\label{eq1}
& \sum_{j=1}^p \left\Vert \sum_{i\neq j}(\hat{\beta}_{i,j}-\beta_{i,j}^*)X_{:,i} \right\Vert_2^2  \\ \leq & 8\sum_{j=1}^p \left[{\varepsilon}_{:,j}^T\frac{\sum_{i\neq j}(\hat{\beta}_{i,j}-\tilde{\beta}_{i,j})X_{:,i}}{\|\sum_{i\neq j}(\hat{\beta}_{i,j}-\tilde{\beta}_{i,j})X_{:,i}\|_2}\right]^2 + 3\sum_{j=1}^p\left\Vert\sum_{i\neq j}(\beta^*_{i,j}-\tilde{\beta}_{i,j})X_{:,i}\right\Vert_2^2.\nonumber
\end{align}
For convenience, let us introduce the following notation:
\begin{align}
    \hat{y}^{(j)}&=\sum_{i\neq j}(\hat{\beta}_{i,j}-\tilde{\beta}_{i,j})X_{:,i},\label{yhatdef}\\
\hat{S}&=S(\hat{z})\cup S(\tilde{u}).
\end{align}
Note that $|\hat{S}|\leq 2s$. In the rest of the proof, $\p(\cdot)$ denotes the probability w.r.t. the data matrix $X$ (which includes $\varepsilon$). Note that $\hat{\beta}_{i,j}-\tilde{\beta}_{i,j}$ can be nonzero only if $i,j\in\hat{S}$ (as $\hat{\beta}_{i,j},\tilde{\beta}_{i,j}$ can be nonzero for $i,j\in S(\hat{z}) $ and $i,j\in S(\tilde{u})$, respectively) so $\hat{y}^{(j)}=0$ for $j\notin \hat{S}$. As a result, one has
\begin{align}
    \p\left(\sum_{j=1}^p \left[\varepsilon_{:,j}^T \frac{\hat{y}^{(j)}}{\|\hat{y}^{(j)}\|_2}\right]^2>t\right)&= \p\left(\sum_{j\in\hat{S}} \left[\varepsilon_{:,j}^T \frac{\hat{y}^{(j)}}{\|\hat{y}^{(j)}\|_2}\right]^2>t\right) \nonumber\\
    & \leq \p\left(\max_{\substack{ S\subseteq[p]\\ |S|=2s}} \sum_{j\in S} \sup_{\substack{v\in\R^p\\ S(v)=S  }}\left[\varepsilon_{:,j}^T \frac{\sum_{i\neq j}v_iX_{:,i}}{\|\sum_{i\neq j}v_iX_{:,i}\|_2}\right]^2>t\right) \label{techlemineq0}
\end{align}
where the inequality~\eqref{techlemineq0} takes into account the definition of
$\hat{y}^{(j)}$ in~\eqref{yhatdef} and the observation that
$$ \sum_{j \in \hat{S}} \left[\varepsilon_{:,j}^T \frac{\hat{y}^{(j)}}{\|\hat{y}^{(j)}\|_2}\right]^2 \leq \sum_{j\in \hat{S}} \sup_{\substack{v\in\R^p\\ S(v)=\hat{S}  }}\left[\varepsilon_{:,j}^T \frac{\sum_{i\neq j}v_iX_{:,i}}{\|\sum_{i\neq j}v_iX_{:,i}\|_2}\right]^2 .$$
Let $\mathbb{X}^{(j)}\in\mathbb{R}^{n\times (p-1)}$ be the matrix $X$ with column $j$ removed. Let $\Phi_S\in\mathbb{R}^{n\times |S|}$ be an orthonormal basis for the linear space spanned by $X_{:,i}$, $i\in S$ for $S\subseteq[p]$.
For a fixed set $S$ and $j\in S$, and $t>0$,
\begin{align}
    \p\left(\sup_{\substack{v\in\R^p\\S(v)=S }}\left[\varepsilon_{:,j}^T \frac{\sum_{i\neq j}v_iX_{:,i}}{\|\sum_{i\neq j}v_iX_{:,i}\|_2}\right]^2>t\right) & \stackrel{(a)}{=} \p\left(\left.\sup_{\substack{v\in\R^p\\ S(v)=S }}\left[\varepsilon_{:,j}^T \frac{\sum_{i\neq j}v_i{X}_{:,i}}{\|\sum_{i\neq j}v_iX_{:,i}\|_2}\right]^2>t\right\vert \mathbb{X}^{(j)}\right)\nonumber\\
    & \stackrel{(b)}{\leq} \p\left(\left.\sup_{\alpha\in \mathcal{B}(|S\setminus\{j\}|)}\left[\varepsilon_{:,j}^T\Phi_{S\setminus\{j\}}\alpha\right]^2>t\right\vert \mathbb{X}^{(j)}\right) \nonumber\\
    & \leq  \p\left(\left.\sup_{{\alpha}\in\mathcal{B}(|S\setminus\{j\}|)}\left[\varepsilon_{:,j}^T{\Phi}_{S\setminus\{j\}}{\alpha}\right]>\sqrt{t}\right\vert \mathbb{X}^{(j)}\right)\nonumber\\
     & \stackrel{(c)}{\leq} \exp\left(-\frac{t}{8({\sigma^*_{j}})^2}+|S\setminus\{j\}|\log 5\right) \nonumber \\
     &  \stackrel{(d)}{\leq} \exp\left(-\frac{t}{16}+2s\log 5\right)\label{techlemineq1}
\end{align}
 where $(a)$ is due to the independence of $\varepsilon_{:,j}$ and $\mathbb{X}^{(j)}$ from Lemma~\ref{lem_reg_ext}, $(b)$ is due to the fact that $\sum_{i\neq j}v_iX_{:,i}$ is in the column span of $\Phi_{S\setminus\{j\}}$, $(c)$ is due to Lemma \ref{maximal} and the conditional distribution $\Phi_{S\setminus\{j\}}^T\varepsilon_{:,j}|\mathbb{X}^{(j)}\stackrel{}{\sim}\mathcal{N}(0,({\sigma_j^*})^2I_{|S|-1})$, and $(d)$ is due to $|S|\leq 2s$ and $({\sigma_j^*})^2\leq 2$. Therefore, by taking $$t_1=c_1 s\log(p/s)+16\log(1/\delta)$$ with $c_1=32\log 5$ in~\eqref{techlemineq1}, we get
\begin{equation}\label{techlemineq2}
 \p\left(\sup_{\substack{v\in\R^p\\S(v)=S }}\left[\varepsilon_{:,j}^T \frac{\sum_{i\neq j}v_iX_{:,i}}{\|\sum_{i\neq j}v_iX_{:,i}\|_2}\right]^2>t_1\right)  \leq \exp\left(-2s\log 5\log(p/s)+\log\delta +2s\log 5\right) \leq \delta,
\end{equation}
where in the second inequality in~\eqref{techlemineq2} we used the assumption $p/s\geq 5>e$. Therefore, by taking $t = 2c_1s^2\log(p/s) +32s\log(1/\delta)$ and 
\begin{equation}\label{deltatechlem}
   \delta=\exp(-(10+e)s\log(2p/s)), 
\end{equation}
we have
 \begin{align}
     \p\left(\sum_{j=1}^p \left[\varepsilon_{:,j}^T \frac{\hat{y}^{(j)}}{\|\hat{y}^{(j)}\|_2}\right]^2> t\right) & \stackrel{(a)}{\leq} \p\left(\max_{\substack{ S\subseteq[p]\\ |S|=2s}} \sum_{j\in S} \sup_{\substack{v\in\R^p\\ S(v)=S }}\left[\varepsilon_{:,j}^T \frac{\sum_{i\neq j}v_iX_{:,i}}{\|\sum_{i\neq j}v_iX_{:,i}\|_2}\right]^2>t\right) \nonumber\\
     & \stackrel{(b)}{\leq} \sum_{\substack{ S\subseteq[p]\\ |S|=2s}} \p\left(\sum_{j\in S} \sup_{\substack{v\in\R^p\\ S(v)=S }}\left[\varepsilon_{:,j}^T \frac{\sum_{i\neq j}v_iX_{:,i}}{\|\sum_{i\neq j}v_iX_{:,i}\|_2}\right]^2>t\right)\nonumber\\
     & \stackrel{(c)}{\leq} \sum_{\substack{ S\subseteq[p]\\ |S|=2s}} \p\left(\exists j\in S:  \sup_{\substack{v\in\R^p\\ S(v)=S }}\left[\varepsilon_{:,j}^T \frac{\sum_{i\neq j}v_iX_{:,i}}{\|\sum_{i\neq j}v_iX_{:,i}\|_2}\right]^2>t_1 \right)\nonumber \\
     & \stackrel{(d)}{\leq} \sum_{\substack{ S\subseteq[p]\\ |S|=2s}} \sum_{j\in S} \p\left( \sup_{\substack{v\in\R^p\\ S(v)=S }}\left[\varepsilon_{:,j}^T \frac{\sum_{i\neq j}v_iX_{:,i}}{\|\sum_{i\neq j}v_iX_{:,i}\|_2}\right]^2>t_1 \right) \nonumber\\
     & \stackrel{(e)}{\leq}  2s {p \choose 2s} \delta\nonumber \\
     &  \stackrel{(f)}{\leq}2s \left(\frac{ep}{2s}\right)^{2s}\delta \nonumber\\
     & \stackrel{(g)}{\leq} 2s\left(\frac{2p}{s}\right)^{es}\delta \nonumber\\
     & = 2s \exp(-10s\log(2p/s)) \nonumber \\
     & \leq 2s \exp(-10s\log(p/s)), \label{techlemineq5}
 \end{align}
where $(a)$ is due to \eqref{techlemineq0}, $(b),(d)$ are true because of union bound, $(c)$ follows by noting that: for $j\in S$, if  
$$\sup_{\substack{v\in\R^p\\ S(v)=S }}\left[\varepsilon_{:,j}^T \frac{\sum_{i\neq j}v_iX_{:,i}}{\|\sum_{i\neq j}v_iX_{:,i}\|_2}\right]^2<t_1,$$
then 
$$\sum_{j\in S} \sup_{\substack{v\in\R^p\\ S(v)=S }}\left[\varepsilon_{:,j}^T \frac{\sum_{i\neq j}v_iX_{:,i}}{\|\sum_{i\neq j}v_iX_{:,i}\|_2}\right]^2<2st_1=t.$$
Inequality $(e)$ above, is due to \eqref{techlemineq2} and the fact that $|S|=2s$, $(f)$ is due to the inequality ${p\choose 2s}\leq (ep/2s)^{2s}$ and $(g)$ is true as $e>2$. By \eqref{eq1} and \eqref{techlemineq5}, and substituting $\delta$ from~\eqref{deltatechlem}, 
\begin{align*}
    \sum_{j=1}^p \left[\left\Vert \sum_{i\neq j}(\hat{\beta}_{i,j}-\beta_{i,j}^*)X_{:,i} \right\Vert_2^2\right]&\lesssim  {s^2\log(p/s)} + \sum_{j=1}^p\left\Vert\sum_{i\neq j}(\beta^*_{i,j}-\tilde{\beta}_{i,j})X_{:,i}\right\Vert_2^2
\end{align*}
    with the probability greater than $1-2s\exp(-10s\log(p/s))$.
\end{proof}

\begin{proof}[\bf Proof of Theorem \ref{reg_thm}]
Let $\mathbb{X}^{(j)}\in\mathbb{R}^{n\times (p-1)}$ be the matrix $X$ with column $j$ removed. Let us define the following events:
\begin{equation}\label{eventsthm1}
    \begin{aligned}
    \mathcal{A}&=\left\{ \max_{j_{1}, j_{2} \in [p]}~\frac{1}{n}|X_{:,j_1}^TX_{:,j_2}|\leq 3\right\},\\
      \mathcal{E}_j& =\left\{\mathbb{X}^{(j)}: \sigma_{\min}(\mathbb{X}^{(j)}_{:,S})\geq c_r \sqrt{n} \,\,\,\forall S\subseteq [p-1], |S|\leq 2s\right\},~j\in[p]
    \end{aligned}
\end{equation}
for some constant $c_r>0$. From \eqref{bern2} for fixed $j_1,j_2\in[p]$, with probability at least $1-2\delta$, we have:
\begin{equation}
    \frac{1}{n}|X_{:,j_1}^TX_{:,j_2}|\leq |G^*_{j_1,j_2}|+c_{\psi}\sqrt{\frac{\log(1/\delta)}{C_bn}}.
\end{equation}
Let $\delta=p^{-10}$ and $n>10c_{\psi}^2\log p/{C_b}$. 
Noting that $|G^*_{j_1,j_2}|\leq 2$ for all $j_{1}, j_{2}$, and by a union bound over all $j_1,j_2\in[p]$, the following holds true
\begin{equation}\label{proof-thm3-in3}
   \p\left(\mathcal{A}\right) >1-2p^{-8}.
\end{equation}
Note that for any $S\subseteq[p]$, with $1\leq|S|\leq 2s$, by Assumption~\ref{gdtsrassumption}, we have:
$$\lambda_{\min}\left(G^*_{{S},{S}}\right)\geq 1.$$
Hence, by Lemma \ref{minnorm2}, as $n\gtrsim s\log p$,
$$\p(\mathcal{E}_j)\geq 1-2s\exp(-20s\log (p-1)).$$
For the rest of the proof, we assume $\mathcal{A},\mathcal{E}_j$ (for all $j\in[p]$) and the event of Lemma~\ref{thm3techlem} hold true. Note that the intersection of these events holds true with probability at least 
 \begin{equation}\label{thm1prob}
     1-2p^{-8}-2sp\exp(-20s\log (p-1))-2s\exp(-10s\log(p/s))
 \end{equation}
 by union bound. \\
 Let $\hat{S}= S(\hat{z})\cup S(\tilde{u})$. Recall that $\hat{\beta}_{i,j} = 0$ for all $i \notin \hat{S}$. We have the following:
\begin{align}
    & \left\Vert \sum_{i:i\neq j}(\hat{\beta}_{i,j}-\beta_{i,j}^*)X_{:,i} \right\Vert_2^2 \nonumber \\ = & \left\Vert \sum_{i\in \hat{S}}(\hat{\beta}_{i,j}-\beta_{i,j}^*)X_{:,i}+\sum_{i\notin \hat{S}}(\hat{\beta}_{i,j}-\beta_{i,j}^*)X_{:,i} \right\Vert_2^2 \nonumber\\
     = & \left\Vert \sum_{i\in \hat{S}}(\hat{\beta}_{i,j}-\beta_{i,j}^*)X_{:,i}\right\Vert_2^2+\left\Vert \sum_{i\notin \hat{S}}\beta_{i,j}^*X_{:,i}\right\Vert_2^2-2\sum_{\substack{i_1\notin \hat{S}\\i_2\in \hat{S}}}\beta^*_{i_1,j}(\hat{\beta}_{i_2,j}-\beta^*_{i_2,j})X_{:,i_1}^TX_{:,i_2}\nonumber\\
     \geq & \left\Vert \sum_{i\in \hat{S}}(\hat{\beta}_{i,j}-\beta_{i,j}^*)X_{:,i}\right\Vert_2^2-2\sum_{\substack{i_1\notin \hat{S}\\i_2\in \hat{S}}}|\beta^*_{i_1,j}||\hat{\beta}_{i_2,j}-\beta^*_{i_2,j}||X_{:,i_1}^TX_{:,i_2}|\label{proof-thm3-aux1}.
\end{align}
Next, note that
\begin{align}
    2\sum_{\substack{i_1\notin \hat{S}\\i_2\in \hat{S}}}|\beta^*_{i_1,j}||\hat{\beta}_{i_2,j}-\beta^*_{i_2,j}||X_{:,i_1}^TX_{:,i_2}|& = 2\sum_{\substack{i_1\notin \hat{S}\\i_2\in \hat{S}}}\sqrt{|\beta^*_{i_1,j}|}\sqrt{|\beta^*_{i_1,j}|}|\hat{\beta}_{i_2,j}-\beta^*_{i_2,j}||X_{:,i_1}^TX_{:,i_2}|\nonumber\\
    & \stackrel{(a)}{\leq}  \sum_{\substack{i_1\notin \hat{S}\\i_2\in \hat{S}}}\left(c_{\beta}|\beta^*_{i_1,j}| +\frac{1}{c_{\beta}} |\beta^*_{i_1,j}||\hat{\beta}_{i_2,j}-\beta^*_{i_2,j}|^2\right)|X_{:,i_1}^TX_{:,i_2}|\nonumber\\
    & \stackrel{(b)}{\leq}  \frac{4n}{c_{\beta}}\sum_{i_1\notin \hat{S}}|\beta^*_{i_1,j}|\sum_{i_2\neq j}(\beta_{i_2,j}^*-\hat{\beta}_{i_2,j})^2  +c_{\beta}\sum_{\substack{i_1\notin \hat{S}\\i_2\in \hat{S}}}|\beta^*_{i_1,j}||X_{:,i_1}^TX_{:,i_2}|\nonumber\\
    & \stackrel{(c)}{\leq}\frac{4n}{c_{\beta}}\sum_{i:i\neq j}(\beta_{i,j}^*-\hat{\beta}_{i,j})^2  +c_{\beta}\sum_{\substack{i_1\notin \hat{S}\\i_2\in \hat{S}}}|\beta^*_{i_1,j}||X_{:,i_1}^TX_{:,i_2}|,\label{proof-thm3-aux2}
\end{align}
for some sufficiently large constant $c_{\beta}>0$ (that we discuss later), where $(a)$ is due to the inequality $2ab\leq c_{\beta}a^2+b^2/c_{\beta}$, $(b)$ is due to event $\mathcal{A}$ defined in~\eqref{eventsthm1} and $(c)$ is true as 
$$\sum_{i\notin \hat{S}}|\beta^*_{i,j}|=\sum_{i\notin \hat{S}}|\tilde{\beta}_{i,j}-\beta^*_{i,j}|\leq \sum_{j=1}^p\sum_{i\neq j}|\tilde{\beta}_{i,j}-\beta^*_{i,j}|\leq 1$$
by Assumption~\ref{orsclebetaassumtion}. As a result, from~\eqref{proof-thm3-aux1} and~\eqref{proof-thm3-aux2},
 \begin{equation} \label{proof-thm3-ie1}
     \left\Vert \sum_{i:i\neq j}(\hat{\beta}_{i,j}-\beta_{i,j}^*)X_{:,i} \right\Vert_2^2 \geq \left\Vert \sum_{i\in \hat{S}}(\hat{\beta}_{i,j}-\beta_{i,j}^*)X_{:,i}\right\Vert_2^2 -\frac{4n}{c_{\beta}}\sum_{i:i\neq j}(\beta_{i,j}^*-\hat{\beta}_{i,j})^2  -c_{\beta}\sum_{\substack{i_1\notin \hat{S}\\i_2\in \hat{S}}}|\beta^*_{i_1,j}||X_{:,i_1}^TX_{:,i_2}|.
 \end{equation}
 On event $\mathcal{E}_j$ for $j\in [p]$, 
\begin{align}
   \frac{1}{n} \left\Vert\sum_{i\in \hat{S}}(\hat{\beta}_{i,j}-\beta^*_{i,j})X_{:,i}\right\Vert_2^2 & =  \frac{1}{n} \left\Vert X_{\hat{S}}(\hat{\beta}_{\hat{S},j}-\beta^*_{\hat{S},j})\right\Vert_2^2\nonumber\\
   & \geq \frac{1}{n}\sigma_{\min}\left(\mathbb{X}^{(j)}_{:,\hat{S}\setminus \{j\}}\right)^2\left\Vert \hat{\beta}_{\hat{S}\setminus \{j\},j}-\beta^*_{\hat{S}\setminus \{j\},j}\right\Vert_2^2\nonumber\\
   & \geq c_r\left\Vert \hat{\beta}_{\hat{S},j}-\beta^*_{\hat{S},j}\right\Vert_2^2 \nonumber\\
   & \geq c_r \sum_{i\in \hat{S}}(\hat{\beta}_{i,j}-\beta^*_{i,j})^2\label{proof-thm3-in2},
\end{align}
where in the first inequality, we used the convention $\beta^*_{j,j}=\hat{\beta}_{j,j}=0$. As a result, for $j\in[p]$
\begin{align}
  & \frac{1}{n} \left\Vert \sum_{i:i\neq j}(\hat{\beta}_{i,j}-\beta_{i,j}^*)X_{:,i} \right\Vert_2^2 \nonumber
    \\ \stackrel{(a)}{\geq} & \frac{1}{n} \left\Vert \sum_{i\in \hat{S}}(\hat{\beta}_{i,j}-\beta_{i,j}^*)X_{:,i}\right\Vert_2^2-\frac{c_{\beta}}{n}\sum_{\substack{i_1\notin \hat{S}\\i_2\in \hat{S}}}|\beta^*_{i_1,j}||X_{:,i_1}^TX_{:,i_2}|-\frac{4}{c_{\beta}}\sum_{i:i\neq j}(\beta_{i,j}^*-\hat{\beta}_{i,j})^2\nonumber \\
    \stackrel{(b)}{\geq} & c_r\sum_{i\in \hat{S}}(\hat{\beta}_{i,j}-\beta^*_{i,j})^2-\frac{c_{\beta}}{n}\sum_{\substack{i_1\notin \hat{S}\\i_2\in \hat{S}}}|\beta^*_{i_1,j}||X_{:,i_1}^TX_{:,i_2}|-\frac{4}{c_{\beta}}\sum_{i:i\neq j}(\beta_{i,j}^*-\hat{\beta}_{i,j})^2\nonumber \\
     \stackrel{(c)}{\gtrsim} & \sum_{i\in \hat{S}}(\hat{\beta}_{i,j}-\beta^*_{i,j})^2 -  \sum_{\substack{i_1\notin \hat{S}\\i_2\in \hat{S}}}|\beta^*_{i_1,j}|\nonumber\\
     \stackrel{(d)}{\gtrsim} & \sum_{i\in \hat{S}}(\hat{\beta}_{i,j}-\beta^*_{i,j})^2 - s \sum_{i\notin \hat{S}}|\beta^*_{i,j}|\nonumber\\
     \stackrel{(e)}{\gtrsim} & \sum_{i\in \hat{S}}(\hat{\beta}_{i,j}-\beta^*_{i,j})^2 - s \sum_{i:i\neq j} |\tilde{\beta}_{i,j}-\beta^*_{i,j}|\label{proof-thm3in4}
\end{align}
where $(a)$ is due to \eqref{proof-thm3-ie1}, $(b)$ is due to \eqref{proof-thm3-in2} and \eqref{proof-thm3-in3}, $(c)$ is by taking $16/c_r>c_{\beta}>8/c_r$, $(d)$ is true because $|\hat{S}|\leq 2s$ and $(e)$ is true as $\tilde{\beta}_{i,j}=0$ for $i\notin \hat{S}$. Therefore, 
\begin{align}
   & \sum_{j=1}^p \sum_{i:i\neq j}(\hat{\beta}_{i,j}-\beta^*_{i,j})^2 \nonumber \\ = & \sum_{j=1}^p \sum_{i\in \hat{S}}(\hat{\beta}_{i,j}-\beta^*_{i,j})^2 + \sum_{j=1}^p \sum_{i\notin\hat{S}}(\beta^*_{i,j})^2 \nonumber\\
     \stackrel{(a)}{=} & \sum_{j=1}^p \sum_{i\in \hat{S}}(\hat{\beta}_{i,j}-\beta^*_{i,j})^2 + \sum_{j=1}^p \sum_{i\notin\hat{S}}(\tilde{\beta}_{i,j}-\beta^*_{i,j})^2 \nonumber\\
     \stackrel{(b)}{\lesssim} & \frac{1}{n} \sum_{j=1}^p\left\Vert \sum_{i:i\neq j}(\hat{\beta}_{i,j}-\beta_{i,j}^*)X_{:,i} \right\Vert_2^2 + \sum_{j=1}^p \sum_{i:i\neq j}\left[(\tilde{\beta}_{i,j}-\beta^*_{i,j})^2+s|\tilde{\beta}_{i,j}-\beta^*_{i,j}|\right]\label{proof-thm3-ineq8}
\end{align}
where $(a)$ is true as $\tilde{\beta}_{i,j}=0$ for $i\notin \hat{S}$ and $(b)$ is true because of \eqref{proof-thm3in4}. In addition,
 \begin{align}
     \sum_{j=1}^p\left\Vert\sum_{i:i\neq j}(\beta^*_{i,j}-\tilde{\beta}_{i,j})X_{:,i}\right\Vert_2^2 & \leq \sum_{j=1}^p\left[\sum_{i:i\neq j}|\beta^*_{i,j}-\tilde{\beta}_{i,j}|\|X_{:,i}\|_2\right]^2 \nonumber\\
     & \lesssim n\sum_{j=1}^p\left[\sum_{i:i\neq j}|\beta^*_{i,j}-\tilde{\beta}_{i,j}|\right]^2 \nonumber \\
     & \lesssim n\left[\sum_{j=1}^p\sum_{i:i\neq j}|\beta^*_{i,j}-\tilde{\beta}_{i,j}|\right]^2\label{proof-thm3-in4}
 \end{align}
 where above, we used event $\mathcal{A}$ to arrive at the second inequality. 
 Starting with~\eqref{proof-thm3-ineq8}, we arrive at the following:
 \begin{align*}
 &\sum_{j=1}^p \sum_{i:i\neq j}(\hat{\beta}_{i,j}-\beta^*_{i,j})^2\nonumber \\  {\lesssim} & \frac{1}{n} \sum_{j=1}^p\left\Vert \sum_{i:i\neq j}(\hat{\beta}_{i,j}-\beta_{i,j}^*)X_{:,i} \right\Vert_2^2 + \sum_{j=1}^p \sum_{i:i\neq j}\left[(\tilde{\beta}_{i,j}-\beta^*_{i,j})^2+s|\tilde{\beta}_{i,j}-\beta^*_{i,j}|\right] \\
  \stackrel{(a)}{\lesssim} & \frac{s^2\log(p/s)}{n} + \frac{1}{n}\sum_{j=1}^p\left\Vert\sum_{i:i\neq j}(\beta^*_{i,j}-\tilde{\beta}_{i,j})X_{:,i}\right\Vert_2^2 + \sum_{j=1}^p \sum_{i:i\neq j}\left[(\tilde{\beta}_{i,j}-\beta^*_{i,j})^2+s|\tilde{\beta}_{i,j}-\beta^*_{i,j}|\right]\\
  \stackrel{(b)}{\lesssim} &  \frac{s^2\log(p/s)}{n} + \sum_{j=1}^p\left[\sum_{i:i\neq j}|\beta^*_{i,j}-\tilde{\beta}_{i,j}|\right]^2+ \sum_{j=1}^p \sum_{i:i\neq j}\left[(\tilde{\beta}_{i,j}-\beta^*_{i,j})^2+s|\tilde{\beta}_{i,j}-\beta^*_{i,j}|\right] \\
  \lesssim & \frac{s^2\log(p/s)}{n} + \left[\sum_{j=1}^p\sum_{i:i\neq j}|\beta^*_{i,j}-\tilde{\beta}_{i,j}|\right]^2+ s\sum_{j=1}^p \sum_{i:i\neq j}|\tilde{\beta}_{i,j}-\beta^*_{i,j}| \\
  \stackrel{(c)}{\lesssim} & \frac{s^2\log(p/s)}{n} + s\sum_{j=1}^p \sum_{i:i\neq j}|\tilde{\beta}_{i,j}-\beta^*_{i,j}|,
 \end{align*}
 where above,  $(a)$ is due to Lemma \ref{thm3techlem}, $(b)$ is true because of \eqref{proof-thm3-in4} and $(c)$ is true by Assumption~\ref{orsclebetaassumtion}. 
\end{proof}

\subsection{Proof of Theorem \ref{Bestconst}}
We first present a few lemmas that are needed in the proof of Theorem~\ref{Bestconst}.
\begin{lem}[Lemma 8, \cite{6034739}]\label{sparseinner}
Let $s<p/4$. Suppose ${\varepsilon}\sim\mathcal{N}(0,\sigma^2{I}_n)$, ${X}\in\mathbb{R}^{n\times p}$ and ${\theta}\in\mathbb{R}^p$ is such that $\|{\theta}\|_0\leq 2s$. We have
$$\frac{1}{n}\left\vert{\varepsilon}^T{X}{\theta}\right\vert\leq 9\sigma \frac{\left\Vert{X}{\theta}\right\Vert_2}{n}\sqrt{s\log(p/s)}$$
with probability greater than $1-\exp(-10s\log(p/2s))$.
\end{lem}
\begin{lem}\label{thm4techlem}
Let $\hat{\beta}$ be an optimal solution to~\eqref{reg_multi} and $\beta^*$ be the true regression coefficients. Under the assumptions of Theorem \ref{Bestconst},
\begin{equation}
    \frac{1}{n}\sum_{j=1}^p\left\vert \varepsilon_{:,j}^T\sum_{i:i\neq j}(\hat{\beta}_{i,j}-\beta^*_{i,j})X_{:,i}\right\vert\lesssim \frac{s^2\log(p/s)}{n}
\end{equation}
with probability greater than 
$$ 1-4p^{-8}-2sp\exp(-20s\log (p-1))-2s\exp(-10s\log(p/s))-p\exp(-10s\log(p/2s)).$$
\end{lem}
\begin{proof}
Let us define the following events:
\begin{equation}
    \begin{aligned}
      \mathcal{C} &= \left\{\max_{\substack{i,j\in[p]\\ i\neq j}} |\varepsilon_{:,j}^TX_{:,i}|\lesssim n\right\},\\
      \mathcal{F}_j &= \left\{ \frac{1}{n}|\varepsilon_{:,j}^T\hat{y}^{(j)}| \lesssim \frac{\sqrt{s\log(p/s)}}{n}\|\hat{y}^{(j)}\|_2 \right\}
    \end{aligned}
\end{equation}
where $\hat{y}^{(j)}$ is defined in \eqref{yhatdef} for $j \in [p]$. As $\varepsilon_{:,j}$ has iid normal coordinates with bounded variance, similar to \eqref{proof-thm3-in3}, by Bernstein inequality 
$$\p(\mathcal{C})\geq 1-2p^{-8}.$$
Moreover, based on Lemma \ref{sparseinner}, we have:
\begin{align}\label{thm4techin2}
    \p(\mathcal{F}_j) \geq 1-\exp(-10s\log(p/2s))
\end{align}
 for all $j\in[p]$. In the rest of the proof, we assume $\mathcal{C}, \mathcal{F}_j$ (for all $j\in[p]$) and the event of Theorem~\ref{reg_thm} hold true. By union bound, this happens with probability greater than 
$$ 1-4p^{-8}-2sp\exp(-20s\log (p-1))-2s\exp(-10s\log(p/s))-p\exp(-10s\log(p/2s)).$$
One has for a fixed $j\in [p]$,
\begin{align}
    \frac{1}{n}\left\vert \varepsilon_{:,j}^T\sum_{i\neq j}(\hat{\beta}_{i,j}-\beta^*_{i,j})X_{:,i} \right\vert & \leq \frac{1}{n}\left\vert \varepsilon_{:,j}^T\sum_{i\neq j}(\tilde{\beta}_{i,j}-\hat{\beta}_{i,j})X_{:,i} \right\vert + \frac{1}{n}\left\vert \varepsilon_{:,j}^T\sum_{i\neq j}(\beta^*_{i,j}-\tilde{\beta}_{i,j})X_{:,i} \right\vert \nonumber \\
    & = \frac{1}{n}|\varepsilon_{:,j}^T\hat{y}^{(j)}| + \frac{1}{n}\left\vert \varepsilon_{:,j}^T\sum_{i\neq j}(\beta^*_{i,j}-\tilde{\beta}_{i,j})X_{:,i} \right\vert \nonumber \\
    & \leq \frac{1}{n}|\varepsilon_{:,j}^T\hat{y}^{(j)}| + \frac{1}{n}\sum_{i\neq j}|\beta^*_{i,j}-\tilde{\beta}_{i,j}|\left\vert \varepsilon_{:,j}^TX_{:,i} \right\vert \nonumber \\
    & \stackrel{(a)}{\lesssim}\frac{1}{n}|\varepsilon_{:,j}^T\hat{y}^{(j)}| + \sum_{i\neq j}|\beta^*_{i,j}-\tilde{\beta}_{i,j}|\label{thm4techlemin3}
\end{align} 
 where $(a)$ follows as we are considering the event $\mathcal{C}$. Additionally, noting that $\hat{y}^{(j)}=0$ for $j\notin \hat{S}$ we have:
\begin{align}
    \sum_{j=1}^p \|\hat{y}^{(j)}\|_2 & = \sum_{j\in\hat{S}} \|\hat{y}^{(j)}\|_2  \nonumber \\
  &  \stackrel{(a)}{\lesssim} \sqrt{s}\sqrt{\sum_{j\in\hat{S}} \|\hat{y}^{(j)}\|_2^2 } \nonumber \\
    & =\sqrt{s}\sqrt{\sum_{j=1}^p \left\Vert\sum_{i\neq j}(\tilde{\beta}_{i,j}-\hat{\beta}_{i,j})X_{:,i} \right\Vert_2^2 } \nonumber \\
    & \lesssim \sqrt{s}\sqrt{\sum_{j=1}^p \left\Vert\sum_{i\neq j}({\beta}^*_{i,j}-\hat{\beta}_{i,j})X_{:,i} \right\Vert_2^2 } + \sqrt{s}\sqrt{\sum_{j=1}^p \left\Vert\sum_{i\neq j}(\tilde{\beta}_{i,j}-{\beta}^*_{i,j})X_{:,i} \right\Vert_2^2 } \nonumber \\
    & \stackrel{(b)}{\lesssim} \sqrt{s}\sqrt{s^2\log(p/s)+\sum_{j=1}^p\left\Vert\sum_{i\neq j}(\beta^*_{i,j}-\tilde{\beta}_{i,j})X_{:,i}\right\Vert_2^2 } + \sqrt{s}\sqrt{\sum_{j=1}^p \left\Vert\sum_{i\neq j}(\tilde{\beta}_{i,j}-{\beta}^*_{i,j})X_{:,i} \right\Vert_2^2 }\nonumber\\ 
    & \lesssim \sqrt{s^3\log(p/s)} + \sqrt{s}\sqrt{\sum_{j=1}^p \left\Vert\sum_{i\neq j}(\tilde{\beta}_{i,j}-{\beta}^*_{i,j})X_{:,i} \right\Vert_2^2 } \nonumber \\
    & \stackrel{(c)}{\lesssim}\sqrt{s^3\log(p/s)} + \sqrt{sn}\sum_{j=1}^p\sum_{i\neq j}|\beta^*_{i,j}-\tilde{\beta}_{i,j}|\label{thm4techlemin1}
\end{align}
where $(a)$ is true as $|\hat{S}|\leq 2s$, $(b)$ is due to Lemma \ref{thm3techlem} and $(c)$ is true because of \eqref{proof-thm3-in4}. 

Starting with~\eqref{thm4techlemin3}, we have the following:
\begin{align}
 &   \frac{1}{n}\sum_{j=1}^p\left\vert \varepsilon_{:,j}^T\sum_{i\neq j}(\hat{\beta}_{i,j}-\beta^*_{i,j})X_{:,i} \right\vert \nonumber \\ \lesssim & \frac{1}{n}\sum_{j=1}^p|\varepsilon_{:,j}^T\hat{y}^{(j)}| + \sum_{j=1}^p\sum_{i\neq j}|\beta^*_{i,j}-\tilde{\beta}_{i,j}| \nonumber\\
     \lesssim &\frac{\sqrt{s\log(p/s)}}{n}\sum_{j=1}^p \|\hat{y}^{(j)}\|_2 + \sum_{j=1}^p\sum_{i\neq j}|\beta^*_{i,j}-\tilde{\beta}_{i,j}|\nonumber\\
     \lesssim & \frac{\sqrt{s\log(p/s)}}{n}\sqrt{s^3\log(p/s)} +\left(1+\sqrt{\frac{s^2\log(p/s)}{n}}\right)\sum_{j=1}^p\sum_{i\neq j}|\beta^*_{i,j}-\tilde{\beta}_{i,j}| \nonumber \\
     \lesssim &\frac{s^2\log(p/s)}{n} 
\end{align}
where above, the second inequality is a result of the event $\mathcal{F}_j$, the third inequality is due to \eqref{thm4techlemin1} and the last inequality is true because of \eqref{sumassum}.\end{proof}

\begin{lem}\label{thm4techlem2}
Let $\{\hat{\sigma}_j\}$ be as defined in~\eqref{sigmahatdef}.
Under the assumptions of Theorem~\ref{Bestconst}, we have:
\begin{equation}\label{thm4ineq201}
     \sum_{j\in S(\hat{z})} \left(  \hat{\sigma}^2_{j}-({\sigma_{j}^*})^2  \right)^2 \lesssim \frac{s^2\log(p/s)}{n}
\end{equation}
with probability greater than 
$$ 1-6p^{-8}-4ps\exp(-20s\log (p-1))-4s\exp(-10s\log(p/s))-p\exp(-10s\log(p/2s))-2pe^{-Cs\log(p/s)}$$
for some universal constant $C>2$.
\end{lem}
\begin{proof}
Define the event
\begin{equation}\label{dj-def}
    \mathcal{D}_j=\left\{\left\vert\frac{1}{n({\sigma^*_j})^2}\|\varepsilon_{:,j}\|_2^2-1\right\vert \lesssim \sqrt{\frac{s\log(p/s)}{n}}\right\}.
\end{equation}
Note that ${\|\varepsilon_{:,j}\|_2^2}/{({\sigma_{j}^*})^2}\sim \chi^2_n$ and by Bernstein inequality, 
\begin{equation}
 \p\left(\mathcal{D}_j\right)\geq 1- 2e^{-Cs\log(p/s)}   \label{thm4ineq104}
\end{equation}
for some numerical constant $C>2$ as $n\gtrsim s\log (p/s)$. Note that from Lemma~\ref{thm4techlem},
\begin{multline}\label{event01}
    \p\left(\frac{1}{n}\sum_{j=1}^p\left\vert \varepsilon_{:,j}^T\sum_{i\neq j}(\hat{\beta}_{i,j}-\beta^*_{i,j})X_{:,i}\right\vert\lesssim \frac{s^2\log(p/s)}{n}\right) \\ \geq 1-4p^{-8}-2sp\exp(-20s\log (p-1))-2s\exp(-10s\log(p/s))-p\exp(-10s\log(p/2s))
\end{multline}
and from Theorem~\ref{reg_thm},
\begin{multline}\label{event02}
\p\left(\sum_{j=1}^p \sum_{i:i\neq j}(\hat{\beta}_{i,j}-\beta^*_{i,j})^2\lesssim    \frac{s^2\log(p/s)}{n} + s\sum_{j=1}^p \sum_{i:i\neq j}|\tilde{\beta}_{i,j}-\beta^*_{i,j}| \right)\\ \geq  1-2p^{-8}-2sp\exp(-20s\log (p-1))-2s\exp(-10s\log(p/s)).
\end{multline}
The rest of the proof is on the intersection of events $\mathcal{D}_j$ for all $j\in[p]$ and events in~\eqref{event01} and \eqref{event02}, with probability of at least 
$$ 1-6p^{-8}-4ps\exp(-20s\log (p-1))-4s\exp(-10s\log(p/s))-p\exp(-10s\log(p/2s))-2pe^{-Cs\log(p/s)}.$$
 To lighten the notation, let $\gamma^{(j)}= \sum_{i:i\neq j}(\hat{\beta}_{i,j}-\beta^*_{i,j})X_{:,i}$. By \eqref{sigmahatdef}, one has
\begin{align}
  &  \sum_{j\in S(\hat{z})} \left(  \hat{\sigma}^2_{j}-({\sigma_{j}^*})^2  \right)^2   \nonumber \\ = & \sum_{j\in S(\hat{z})} \left( \frac{\left\Vert X_{:,j} -\sum_{i\neq j}\hat{\beta}_{i,j}X_{:,i} \right\Vert_2^2}{n} -({\sigma_{j}^*})^2  \right)^2 \nonumber\\
     \stackrel{(a)}{=} & \sum_{j\in S(\hat{z})} \left( \frac{\left\Vert \sum_{i\neq j}(\hat{\beta}_{i,j}-\beta^*_{i,j})X_{:,i} \right\Vert_2^2}{n} -2\varepsilon_{:,j}^T \frac{\sum_{i\neq j}(\hat{\beta}_{i,j}-\beta^*_{i,j})X_{:,i}}{n}+\left(\frac{\|\varepsilon_{:,j}\|_2^2}{n} -({\sigma_{j}^*})^2\right)  \right)^2 \nonumber\\
   \stackrel{(b)}{\leq} & 4\sum_{j\in S(\hat{z})} \left\{\left(\frac{ \|\gamma^{(j)}\|_2^2}{n}\right)^2 + \left(2\varepsilon_{:,j}^T \frac{\gamma^{(j)}}{n}\right)^2 + ({\sigma_{j}^*})^4\left(\frac{\left\Vert\frac{\varepsilon_{:,j}}{{\sigma_{j}^*}}\right\Vert_2^2}{n} -1\right)^2\right\}\nonumber\\
  \leq & 4 \left(\sum_{j=1}^p \frac{1}{n}\|\gamma^{(j)}\|_2^2\right)^2 + 4\left(\sum_{j=1}^p \frac{2}{n}|\varepsilon_{:,j}^T \gamma^{(j)}|\right)^2 +4 \sum_{j\in S(\hat{z})}({\sigma_{j}^*})^4\left(\frac{\left\Vert\frac{\varepsilon_{:,j}}{{\sigma_{j}^*}}\right\Vert_2^2}{n} -1\right)^2 \label{thm4ineq101}
\end{align}
where $(a)$ uses the representation $X_{:,j}=\sum_{i\neq j}\beta^*_{i,j}X_{:,i}+\varepsilon_{:,j}$ and $(b)$ is a result of the inequality $(a+b+c)^2\leq 4(a^2+b^2+c^2)$. Note that 
\begin{align}
    \sum_{j=1}^p \frac{1}{n}\|\gamma^{(j)}\|_2^2 & \stackrel{(a)}{\lesssim} \frac{s^2\log(p/s)}{n} + \frac{1}{n}\sum_{j=1}^p\left\Vert\sum_{i\neq j}(\beta^*_{i,j}-\tilde{\beta}_{i,j})X_{:,i}\right\Vert_2^2 \nonumber \\
   & \stackrel{(b)}{\lesssim} \frac{s^2\log(p/s)}{n} + \left[\sum_{j=1}^p\sum_{i\neq j}|\beta^*_{i,j}-\tilde{\beta}_{i,j}|\right]^2 \nonumber\\
   & \stackrel{(c)}{\lesssim} \frac{s^2\log(p/s)}{n}\label{thm4ineq102}
\end{align}
where $(a)$ is due to Lemma \ref{thm3techlem}, $(b)$ is due to \eqref{proof-thm3-in4} and $(c)$ is due to \eqref{sumassum}. Additionally,
\begin{equation}\label{thm4ineq103}
    \frac{1}{n}\sum_{j=1}^p\left\vert \varepsilon_{:,j}^T\gamma^{(j)} \right\vert\lesssim \frac{s^2\log(p/s)}{n}
\end{equation}
by Lemma \ref{thm4techlem}. 
Moreover, from Assumption~\ref{gdtsrassumption}, $({\sigma_j^*})^2\leq 2$ so $({\sigma_j^*})^4 \leq 4$. By substituting \eqref{thm4ineq102}, \eqref{thm4ineq103} into \eqref{thm4ineq101},  and considering that event $\mathcal{D}_j$ hold true, we have
\begin{align*}
    \sum_{j=1}^p  \sum_{j\in S(\hat{z})} \left(  \hat{\sigma}^2_{j}-({\sigma_{j}^*})^2  \right)^2 & \lesssim \frac{s^4\log^2(p/s)}{n^2} + |S(\hat{z})| \frac{s\log(p/s)}{n} \\
    & \lesssim \frac{s^4\log^2(p/s)}{n^2} +  \frac{s^2\log(p/s)}{n} \\
    & \lesssim  \frac{s^2\log(p/s)}{n} 
\end{align*}
where the second inequality is by the fact $|S(\hat{z})|\leq s$ and the last inequality follows by observing that $s^4\log^2(p/s)/n^2\lesssim s^2\log(p/s)/n$ as $n\gtrsim s^2\log(p/s)$.
\end{proof}

\begin{proof}[\bf Proof of Theorem \ref{Bestconst}]
Part i) We prove this part by contradiction. Suppose there exists $j_0\in S(\tilde{u})$ such that $j_0\notin S(\hat{z})$. In this case, $\hat{\beta}_{j_0,1},\cdots,\hat{\beta}_{j_0,p}=0$ while $|\tilde{\beta}_{j_0,i}|\geq \widetilde{\beta_{\min}}$ for $i\in S(\tilde{u})$, $i\neq j$. Therefore by  
Assumption~\ref{beta_min_assum}:
\begin{align}
   c_{\min}^2 \frac{s^2\log p}{n}& \leq  s\widetilde{\beta_{\min}}^2 \nonumber \\
    & \stackrel{(a)}{\lesssim}\sum_{\substack{i\in S(\tilde{u}) \\ i\neq j_0 }}(\tilde{\beta}_{j_0,i})^2 \nonumber\\
    & = \sum_{\substack{i\in S(\tilde{u}) \\ i\neq j_0 }}(\tilde{\beta}_{j_0,i}-\beta^*_{j_0,i}+\beta^*_{j_0,i})^2 \nonumber\\
    & \lesssim \sum_{\substack{i\in S(\tilde{u}) \\ i\neq j_0 }}(\beta^*_{j_0,i})^2+\sum_{\substack{i\in S(\tilde{u}) \\ i\neq j_0 }}(\tilde{\beta}_{j_0,i}-\beta^*_{j_0,i})^2 \nonumber \\
    &  \stackrel{(b)}{=}\sum_{\substack{i\in S(\tilde{u}) \\ i\neq j_0 }}(\beta^*_{j_0,i}-\hat{\beta}_{j_0,i})^2+\sum_{\substack{i\in S(\tilde{u}) \\ i\neq j_0 }}(\tilde{\beta}_{j_0,i}-\beta^*_{j_0,i})^2 \nonumber\\ 
    & \leq \sum_{j=1}^p\sum_{i:i\neq j}(\beta^*_{j,i}-\hat{\beta}_{j,i})^2+\sum_{j=1}^p\sum_{i:i\neq j}(\tilde{\beta}_{j_0,i}-\beta^*_{j_0,i})^2 \nonumber\\
    & \stackrel{(c)}{\lesssim} \frac{s^2\log(p/s)}{n} + s\sum_{j=1}^p \sum_{i:i\neq j}|\tilde{\beta}_{i,j}-\beta^*_{i,j}|+(\sum_{j=1}^p \sum_{i:i\neq j}|\tilde{\beta}_{i,j}-\beta^*_{i,j}|)^2 \nonumber \\
    & \stackrel{(d)}{\lesssim} \frac{s^2\log(p/s)}{n}+ \frac{s^2\log^2(p/s)}{n^2}\nonumber\\
    & \stackrel{(e)}{\lesssim} \frac{s^2\log(p/s)}{n}\label{thm4contrad}
\end{align}
with probability greater than 
$$1-2p^{-8}-2sp\exp(-20s\log (p-1))-2s\exp(-10s\log(p/s))$$
where $(a)$ is true because of the $\widetilde{\beta_{\min}}$ condition in 
Assumption~\ref{beta_min_assum} 
and the fact $|S(\tilde{u})|=s$, $(b)$ is true because for $i\in S(\tilde{u})$, $\hat{\beta}_{j_0,i}=0$, $(c)$ is due to Theorem \ref{reg_thm}, $(d)$ is because of \eqref{sumassum} 
and $(e)$ is due to the fact $n\gtrsim s^2\log p$ is sufficiently large. Note that if $c_{\min}$ is chosen large enough, \eqref{thm4contrad} leads to a contradiction---this shows $S(\tilde{u})\subseteq S(\hat{z})$ and as $|S(\hat{z})|\leq s$, we arrive at $S(\tilde{u})= S(\hat{z})$. \\~\\
\noindent Part ii) This part of proof is on the events: Part 1 of this theorem, Lemma~\ref{thm4techlem2} and Theorem~\ref{reg_thm}. This happens with probability greater than 
\begin{equation}\label{probexp}
\begin{aligned}
     1-10p^{-8}-8sp\exp(-20s\log (p-1))-8s\exp(-10s\log(p/s))\\
     -p\exp(-10s\log(p/2s))-2pe^{-Cs\log(p/s)}.
     \end{aligned}
\end{equation}
By \eqref{sigmahatdef}, $\hat{\sigma}_j=1$ for $j\notin S(\hat{z})$. As a result,  
\begin{align*}
    \|\hat{B}-B^*\|_F^2 & = \sum_{j=1}^p\sum_{i\neq j}(\hat{\beta}_{i,j}-\beta^*_{i,j})^2 + \sum_{j\in S(\hat{z})}\left(  \hat{\sigma}^2_{j}-({\sigma_{j}^*})^2  \right)^2+\sum_{j\notin S(\hat{z})}\left(  1-({\sigma_{j}^*})^2  \right)^2 \\
    & \stackrel{(a)}{=}\sum_{j=1}^p\sum_{i\neq j}(\hat{\beta}_{i,j}-\beta^*_{i,j})^2 + \sum_{j\in S(\hat{z})}\left(  \hat{\sigma}^2_{j}-({\sigma_{j}^*})^2  \right)^2+\sum_{j\notin S(\tilde{u})}\left(  \tilde{\sigma}_j^2-({\sigma_{j}^*})^2  \right)^2 \\
    & \stackrel{(b)}{\lesssim}\sum_{j=1}^p\sum_{i\neq j}(\hat{\beta}_{i,j}-\beta^*_{i,j})^2 + \sum_{j\in S(\hat{z})}\left(  \hat{\sigma}^2_{j}-({\sigma_{j}^*})^2  \right)^2+\frac{s^2\log(p/s)}{n} \\
    & \stackrel{(c)}{\lesssim}\sum_{j=1}^p\sum_{i\neq j}(\hat{\beta}_{i,j}-\beta^*_{i,j})^2 +\frac{s^2\log(p/s)}{n} \\
    & \stackrel{(d)}{\lesssim}\frac{s^2\log(p/s)}{n} + s\sum_{j=1}^p \sum_{i\neq j}|\tilde{\beta}_{i,j}-\beta^*_{i,j}| \\
    & \stackrel{(e)}{\lesssim}\frac{s^2\log(p/s)}{n}
\end{align*}
where $(a)$ is true because by Part 1 of the theorem, $S(\hat{z})=S(\tilde{u})$ and the observation $\tilde{\sigma}_j=1$ if $\tilde{u}_j=0$, $(b)$ is due to \eqref{sigmaassum}, $(c)$ is due to \eqref{thm4ineq201}, $(d)$ is due to Theorem \ref{reg_thm} and $(e)$ is due to \eqref{sumassum}. 
\end{proof}

\section{Proving Theorem \ref{supportthm}}\label{appendix:d}
We first introduce some notation that we will be using in this proof. 

\noindent \textbf{Notation.} In this proof, we use the following notation. For any $S\subseteq [p]$, we let $X_S$ be the submatrix of $X$ with columns indexed by $S$. For $G\in\R^{p\times p}$ and $S_1,S_2\subseteq[p]$, we let $G_{S_1,S_2}$ be the submatrix of $G$ with rows in $S_1$ and columns in $S_2$. We define the operator and max norm of $G\in\R^{p_1\times p_2}$ as
\begin{equation*}
    \begin{aligned}
     \|G\|_{\text{op}}&=\max_{\substack {x\in\R^{p_1}\\x\neq 0}}\frac{\|Gx\|_2}{\|x\|_2}, \\
\|G\|_{\max} &= \max_{(i,j)\in[p_1]\times [p_2]} |G_{i,j}|
    \end{aligned}
\end{equation*}
respectively. We denote the projection matrix onto  the column span of $X_S$ by ${P}_{{X}_{S}}$. Note that if $X_S$ has linearly independent columns, ${P}_{{X}_{S}}={X}_{S}({X}_{S}^T{X}_{S})^{-1}{X}_{S}^T$. In our case, as the data is drawn from a normal distribution with a full-rank covariance matrix, for any $S\subseteq[p]$ with $|S|<n$, $X_S$ has linearly independent columns with probability one. The solution to the least squares problem with the support restricted to $S$,
\begin{equation}\label{limitedls}
    \min_{\beta_{S^c}=0}\frac{1}{n}\|y-X\beta\|_2^2
\end{equation}
for $y\in\R^n$ and $X\in\R^{n\times p}$ is given by $$\beta_{S}=({X}_{S}^T{X}_{S})^{-1}{X}_{S}^Ty.$$
Consequently, we denote the optimal objective in~\eqref{limitedls} by 
\begin{equation}\label{thm3-proof-cr}
    \CR_{S}({y})= \frac{1}{n}{y}^T({I}_n-{P}_{{X}_{S}}){y}.
\end{equation}
We let $\hat{G}=X^TX/n$ be the sample covariance. Throughout this proof, we use the notation $\CS^*_j= S(u^*)\setminus \{j\}$ for $j\in S(u^*)$ and $\CS_j^*=\emptyset$ otherwise. Finally, for $S_1,S_2\subseteq [p]$, $G\in\R^{p\times p}$ positive definite and $S_0=S_2\setminus S_1$, we let
\begin{equation}
    \schur{G}{S_1}{S_2} = G_{S_0,S_0}- G_{S_0,S_1}G^{-1}_{S_1,S_1}G_{S_1,S_0}.
\end{equation}
Note that $\schur{G}{S_1}{S_2}$ is the Schur complement of the matrix
\begin{equation}\label{gs1s2}
  G(S_1,S_2)=  \begin{bmatrix}
G_{S_1,S_1} & G_{S_1,S_0} \\
G_{S_0,S_1} & G_{S_0,S_0}
\end{bmatrix}.
\end{equation}
Let us define the following events for $j\in[p]$ and $\CS\subseteq[p]$:
\begin{equation}\label{supevents}
    \begin{aligned}
    \mathcal{E}_1(j) &= \left\{  {({\beta}^*_{\CS^*_j,j})}^T \hat{G}_{\CS^*_j,\CS^*_j} {\beta}^*_{\CS^*_j,j} \geq 0.8\eta\frac{s\log p}{n}\right\}\\
       \mathcal{E}_2(j,\CS) &= \left\{ {({\beta}^*_{\CS^0_j,j})}^T (\schur{\hat{G}}{\CS}{\CS_j^*}) {\beta}^*_{\CS^0_j,j} \geq  0.8\eta \frac{|\CS_j^*\setminus\CS|\log p}{n}\right\}\\
       \mathcal{E}_3(j,\CS) &= \left\{  \frac{1}{n}\varepsilon_{:,j}^T ({I}_n-{P}_{{X}_{\CS}}){X}_{\CS^0_j}{\beta}^*_{\CS^0_j,j}\geq -c_{t_1} \sqrt{{({\beta}^*_{\CS^0_j,j})}^T (\schur{\hat{G}}{\CS}{\CS_j^*}) {\beta}^*_{\CS^0_j,j}}\sqrt{\frac{|\CS_j^*\setminus\CS|\log p}{n}}\right\}\\
       \mathcal{E}_4(j,\CS) &= \left\{ -c_{t_2}\frac{s\log p}{n}  \leq \frac{1}{n}\varepsilon_{:,j}^T P_{X_S}\varepsilon_{:,j} \leq c_{t_2}\frac{s\log p}{n}   \right\}\\
       \mathcal{E}_5(j,\CS) &= \left\{ \varepsilon_{:,j}^T ({P}_{{X}_{\CS}}-{P}_{{X}_{\CS_j^*}}) \varepsilon_{:,j}\leq  c_{t_3}|\CS_j^*\setminus\CS|\log p\right\}\\
       \mathcal{E}_6(j) & = \left\{2\left(\frac{\varepsilon_{:,j}^T  X_{\CS_j^*}\beta^*_{\CS_j^*,j}}{\|X_{\CS_j^*}\beta^*_{\CS_j^*,j}\|_2}\right)^2\leq c_{t_4}s\log p\right\}
    \end{aligned}
\end{equation}
where $c_{t_1},c_{t_2},c_{t_3},c_{t_4}$ are some given universal constants and $\CS^0_j = \CS_j^*\setminus\CS$. In the following lemmas, we prove that the above events hold true with high probability.

\subsection{Some technical lemmas} We present a few lemmas that will be useful for the proof of Theorem~\ref{supportthm}.
\begin{lem}\label{suptechlem4}
Let 
\begin{equation*}
    \begin{aligned}
      A & = \bigcap_{j\in S(u^*)}\mathcal{E}_1(j) \\
      B & = \bigcap_{j\in S(u^*)}\bigcap_{\substack{\CS\subseteq[p]\setminus\{j\} \\ |\CS|\leq s }}\mathcal{E}_2(j,\CS).
    \end{aligned}
\end{equation*}
Then, under the assumptions of Theorem \ref{supportthm},
\begin{equation}\label{lem11-part1}
    \p\left(A\right)\geq 1-2s\exp(-10s\log p).
\end{equation}
      Moreover, under condition~\ref{betamincond},
      \begin{equation}\label{lem11-part2}
           \p(B)\geq 1-2s\exp(-10s\log p).
      \end{equation}
\end{lem}
\begin{proof}

Let the event $\mathcal{E}_0(\CS)$ for $\CS\subseteq[p],|\CS|\leq 2s$ be defined as (recall, $\hat{G} = X^TX/n$):
$$\mathcal{E}_0(\CS)=\left\{\left\|\hat{G}_{\CS,\CS}-G^*_{\CS,\CS}\right\|_{\text{op}} \lesssim \|G^*_{\CS,\CS}\|_{\text{op}}\sqrt{\frac{s\log p}{n}}\leq c_b \sqrt{\frac{s\log p}{n}}\right\}$$
for some universal constant $c_b>0$ (note that $\|G^*_{\CS,\CS}\|_{\text{op}}\leq 2$).
As $n\gtrsim s\log p$, one has (for example, by Theorem~5.7 of~\citet{rigollet2015high} with $\delta=\exp(-12s\log p)$)
$$\p(\mathcal{E}_0(\CS))\geq 1-\exp(-12s\log p).$$
Let 
\begin{equation}\label{new-e0}
  \mathcal{E}_0 = \bigcap_{\substack{\CS\subseteq[p]\\ |\CS|\leq 2s}}\mathcal{E}_0(\CS).  
\end{equation}
Then, by union bound,
\begin{align}
    \p\left(\mathcal{E}_0\right)&\geq 1-  \sum_{\substack{\CS\subseteq[p]\\ |\CS|\leq 2s}}\p\left(\mathcal{E}_0(\CS)\right)\nonumber \\
    & \geq 1 - \sum_{\substack{\CS\subseteq[p]\\ |\CS|\leq 2s}}\exp(-12s\log p)  \nonumber \\ & \geq 1- 2sp^{2s}\exp(-12s\log p) = 1-2s\exp(-10s\log p).
\end{align}
Rest of the proof is on the event $\mathcal{E}_0$ defined in~\eqref{new-e0}. 
 As a result, for any $\CS\subseteq[p]$ with $|\CS|\leq s$,
\begin{align}
     \|\hat{G}_{\CS,\CS}-G^*_{\CS,\CS}\|_{\text{op}} &  \leq c_b \sqrt{\frac{s\log p}{n}} =\pi\label{stewineq1}
\end{align}
where $\pi=c_b\sqrt{s\log p/n}$. For $n\gtrsim s\log p$ that is sufficiently large, we have $\pi\leq 0.2$. Hence, by Weyl's inequality, 
\begin{align*}
    \lambda_{\min}(\hat{G}_{\CS,\CS}) 
    &\geq  \lambda_{\min}({G}^*_{\CS,\CS}) - \|\hat{G}_{\CS,\CS}-{G}^*_{\CS,\CS}\|_{\text{op}}  \geq 1 - \pi > 0.8.
\end{align*}
As a result, 
\begin{align*}
    {({\beta}^*_{\CS^*_j,j})}^T \hat{G}_{\CS^*_j,\CS^*_j} {\beta}^*_{\CS^*_j,j} \geq \lambda_{\min}(\hat{G}_{\CS^*_j,\CS^*_j})\|{\beta}^*_{\CS^*_j,j}\|_2^2 \stackrel{(a)}{\geq}  0.8\eta\frac{s\log p}{n}
\end{align*}
where $(a)$ follows from either condition~\ref{betamincond} or condition~\ref{betamincond-extend}. This proves the inequality in display~\eqref{lem11-part1}. \\
Similar to \eqref{stewineq1}, for any $j\in S(u^*)$ and $\CS\subseteq[p]\setminus \{j\}$ with $|\CS|\leq s$,
\begin{align}
     \|\hat{G}(\CS,\CS_j^*)-G^*(\CS,\CS_j^*)\|_{\text{op}} \leq\pi\label{stewineq2}
\end{align}
as $|\CS\cup \CS_j^*|\leq 2s$.
One has 
\begin{align}
    \lambda_{\min}(\schur{\hat{G}}{\CS}{\CS_j^*}) & \stackrel{(a)}{\geq} \lambda_{\min}({\hat{G}}({\CS},{\CS_j^*})) 
    &\stackrel{(b)}{\geq}  \lambda_{\min}({{G}^*}({\CS},{\CS_j^*})) - \|{\hat{G}}({\CS},{\CS_j^*}) - {{G}^*}({\CS},{\CS_j^*})\|_{\text{op}} \nonumber \\  & \geq 1 - \pi > 0.8 \label{newhelperD11}
\end{align}
where $(a)$ is by Corollary 2.3 of \citep{zhang2006schur} and $(b)$ is due to Weyl's inequality. Finally,
\begin{align*}
    {({\beta}^*_{\CS^0_j,j})}^T (\schur{\hat{G}}{\CS}{\CS_j^*}) {\beta}^*_{\CS^0_j,j} \geq \lambda_{\min}(\schur{\hat{G}}{\CS}{\CS_j^*})\|{\beta}^*_{\CS^0_j,j}\|_2^2 \stackrel{(a)}{\geq} 0.8\beta_{\min}^2 |\CS^0_j|\stackrel{(b)}{\geq} 0.8\eta \frac{|\CS^0_j|\log p}{n}.
\end{align*}
where $(a)$ is achieved by~\eqref{newhelperD11} and Condition~\ref{betamincond} in Theorem~\ref{supportthm} and $(b)$ is achieved by substituting $\beta_{\min}$ from Condition~\ref{betamincond}.
This completes the proof of~\eqref{lem11-part2}.
\end{proof}

\begin{lem}\label{suptechlem1}
Let us define the event 
$$C = \bigcap_{j\in S(u^*)}\bigcap_{\substack{ \CS\subseteq[p]\setminus\{j\} \\ |\CS|=s-1 }}\mathcal{E}_3(j,\CS). $$
Then, $\p(C) \geq 1-sp^{-7}$.
\end{lem}
\begin{proof}
First, fix  $j\in S(u^*)$ and $\CS\subseteq [p]\setminus \{j\}$ such that $|\CS|=|\CS_j^*|=s-1$. Let $\CS_j^0=\CS_j^*\setminus \CS$ and $t=|\CS_j^0|$. Moreover, let 
\begin{equation*}
    {\gamma}^{(j,\CS)}= ({I}_n-{P}_{{X}_{\CS}}){X}_{\CS^0_j}{\beta}^*_{\CS^0_j,j}.
\end{equation*}
Note that ${\gamma}^{(j,\CS)}$ is in the column span of $ ({I}_n-{P}_{{X}_{\CS}}){X}_{\CS^0_j}$ and the matrix $ ({I}_n-{P}_{{X}_{\CS}}){X}_{\CS^0_j}$ has rank at most $t$ as $X_{\CS_j^0}$ has $t$ columns. In addition, by Lemma \ref{lem_reg_ext}, $\varepsilon_{:,j}$ and ${X}_{\CS\cup \CS_j^*}$ are independent. By following an argument similar to~\eqref{techlemineq1} (considering an orthonormal basis for the column span of $({I}_n-{P}_{{X}_{\CS}}){X}_{\CS^0_j}$), for $x>0$, we have:
\begin{align}
    \p\left(\frac{\varepsilon_{:,j}^T {\gamma}^{(j,\CS)}}{\|{\gamma}^{(j,\CS)}\|_2}< -x\right) \leq  \exp\left(-\frac{x^2}{8({\sigma_j^*})^2}+t\log 5\right).\label{suplem1ineq1}
\end{align}
Let 
$$x^2 = {8\xi^2 ({\sigma_j^*})^2 t\log p }$$
for $\xi^2>10+{\log 5}$. As a result, by \eqref{suplem1ineq1}, 
$$\p\left(\frac{\varepsilon_{:,j}^T {\gamma}^{(j,\CS)}}{n}<- \sqrt{8}\xi\sigma_j^*\|n^{-1/2}{\gamma}^{(j,\CS)}\|_2\sqrt{\frac{t\log p}{n}}\right) \leq \exp(-10t\log p).$$
Next, note that by the definition of $\gamma^{(j,\CS)}$,
\begin{align*}
    \|n^{-1/2}{\gamma}^{(j,\CS)}\|_2^2 & = \frac{1}{n}{({\beta}^*_{\CS^0_j,j})}^T X^T_{\CS_j^0}({I}_n-{P}_{{X}_{\CS}}){X}_{\CS^0_j} {\beta}^*_{\CS^0_j,j} \\
    & =  \frac{1}{n}{({\beta}^*_{\CS^0_j,j})}^T\left(X^T_{\CS_j^0}X_{\CS_j^0}-X^T_{\CS_j^0}X_{\CS}(X_{\CS}^TX_{\CS})^{-1}X_{\CS}^T X_{\CS_j^0}\right) {\beta}^*_{\CS^0_j,j} \\
    & = \frac{1}{n}{({\beta}^*_{\CS^0_j,j})}^T (\schur{n\hat{G}}{\CS}{\CS_j^*}) {\beta}^*_{\CS^0_j,j} \\
    & ={({\beta}^*_{\CS^0_j,j})}^T (\schur{\hat{G}}{\CS}{\CS_j^*}) {\beta}^*_{\CS^0_j,j}
\end{align*}
where the second equality follows from the definition of $P_{X_{\CS}}$ and the third equality follows from the definition of $\schur{\hat{G}}{\CS}{\CS_j^*}$. This completes the proof for a fixed $\CS$. Finally, we use union bound over all possible choices of $\CS$, $t$ and $j$. As a result, $1-\p(C)$ is bounded from above as
\begin{align*}
    p\sum_{t=1}^s {p-s \choose t}{s-1 \choose t}\exp\left(-10t\log p\right) & \leq  p\sum_{t=1}^s p^{2t}\exp\left(-10t\log p\right)  \leq p\sum_{t=1}^s \exp\left(-8\log p\right)= sp^{-7}.
\end{align*}
\end{proof}

\begin{lem}\label{suptechlem2}
Let 
$$D = \bigcap_{j\in [p]}\bigcap_{\substack{\CS\subseteq[p]\setminus \{j\} \\|\CS|=s-1  }}\mathcal{E}_4(j,\CS).$$
Then, $\p\left(D\right) \geq 1-2sp^{-7}.$
\end{lem}
\begin{proof}
First, fix  $S\subseteq[p]\setminus\{j\}$ such that $|S|=s-1$. Note that $\|P_{X_S}\|_{\text{op}}=1$, $\|P_{X_S}\|_F\leq s$ and $\tr(P_{X_S})\leq s$ as $P_{X_S}$ is a projection matrix. In addition, $P_{X_S}$ and $\varepsilon_{:,j}$ are independent by Lemma~\ref{lem_reg_ext} as $j\notin S$. Therefore, 
\begin{equation}\label{expectPX}
    \E[\varepsilon_{:,j}^T{P}_{X_S}\varepsilon_{:,j}]=\E[\varepsilon_{:,j}^T{P}_{X_S}\varepsilon_{:,j}|P_{X_S}]=\E[\tr({P}_{X_S}\varepsilon_{:,j}\varepsilon_{:,j}^T)|P_{X_S}]\leq s({\sigma_j^*})^2.
\end{equation}
By Hanson-Wright inequality (see Theorem 1.1 of \cite{10.1214/ECP.v18-2865} and the calculations leading to (6.7) of \cite{fan2020best}), for $x>0$, the following
\begin{align}\label{expectPX111}
   \E[\varepsilon_{:,j}^T{P}_{X_S}\varepsilon_{:,j}] - ({\sigma_j^*})^2x \leq \varepsilon_{:,j}^T P_{X_S}\varepsilon_{:,j} \leq   \E[\varepsilon_{:,j}^T{P}_{X_S}\varepsilon_{:,j}] + ({\sigma_j^*})^2x 
\end{align}
holds with probability greater than $1-2\exp(-c\min(x,x^2/s))$ for some $c>0$. Note that by~\eqref{expectPX}, the event in~\eqref{expectPX111} implies
$$    -({\sigma_j^*})^2x  \leq \varepsilon_{:,j}^T P_{X_S}\varepsilon_{:,j} \leq s({\sigma_j^*})^2 + ({\sigma_j^*})^2x.$$
Take $x=\xi s\log p$ for $\xi$ that is sufficiently large. As $\sigma^*_j\leq \sqrt{2}$ by Lemma~\ref{Ginv}, this leads to
$$-\frac{s\log p}{n} \lesssim\frac{1}{n} \varepsilon_{:,j}^T P_{X_S}\varepsilon_{:,j} \lesssim \frac{s\log p}{n} $$
with probability greater than $1-2\exp(-10s\log p)$, which completes the proof for a fixed $\CS$. Similar to Lemma \ref{suptechlem1}, we use union bound to achieve the desired result.
\end{proof}

\begin{lem}\label{suptechlem3}
Let 
$$F=\bigcap_{j\in S(u^*)}\bigcap_{\substack{ \CS\subseteq[p]\setminus\{j\} \\ |\CS|=s-1 }}\mathcal{E}_5(j,\CS).$$
Then, $\p\left( F\right) \geq 1-4sp^{-7}.$
\end{lem}
\begin{proof}
The proof is similar to the proof of (6.7) in \citep{fan2020best}. However, we recall that the work of \citep{fan2020best} considers a fixed design $X$, while here, we deal with random design. 

First, fix  $j\in S(u^*)$ and $\CS\subseteq [p]\setminus \{j\}$ such that $|\CS|=|\CS_j^*|=s-1$ and $t=|\CS_j^*\setminus \CS|$. Let $\mathcal{W}$ be the column span of ${X}_{\CS\cap\CS_j^*}$. Moreover, let $\mathcal{U},\mathcal{V}$ be orthogonal complements of $\mathcal{W}$ as subspaces of column spans of ${X}_{\CS}$ and ${X}_{\CS_j^*}$, respectively. Let ${P}_{\mathcal{U}},{P}_{\mathcal{V}},P_{\mathcal{W}}$ be projection matrices onto $\mathcal{U},\mathcal{V},\mathcal{W}$, respectively. With this notation in place, one has
\begin{align*}
    \varepsilon_{:,j}^T ({P}_{{X}_{\CS}}-{P}_{{X}_{\CS_j^*}}) \varepsilon_{:,j} & = \varepsilon_{:,j}^T \left(( {P}_{\mathcal{U}}+ {P}_{\mathcal{W}})-( {P}_{\mathcal{V}}+ {P}_{\mathcal{W}})\right) \varepsilon_{:,j} \\
    & = \varepsilon_{:,j}^T ( {P}_{\mathcal{U}}-{P}_{\mathcal{V}} ) \varepsilon_{:,j}.
\end{align*}
Note that $\dim(\mathcal{U}),\dim(\mathcal{V})\leq t$. In addition, by Lemma~\ref{lem_reg_ext}, $\varepsilon_{:,j}$ and $X_{\CS\cup\CS_j^*}$ are independent. As a result, 
$$\E[\varepsilon_{:,j}^T{P}_{\mathcal{U}}\varepsilon_{:,j}]=\E[\varepsilon_{:,j}^T{P}_{\mathcal{U}}\varepsilon_{:,j}|X_{\CS\cup\CS_j^*}]\leq  t({\sigma_j^*})^2 ~~\text{and}~~\E[\varepsilon_{:,j}^T{P}_{\mathcal{V}}\varepsilon_{:,j}]=\E[\varepsilon_{:,j}^T{P}_{\mathcal{V}}\varepsilon_{:,j}|X_{\CS\cup\CS_j^*}]\leq  t({\sigma_j^*})^2.$$
Therefore, by Hanson-Wright inequality, for $x>0$, we have:
\begin{align*}
\p \left(    \varepsilon_{:,j}^T{P}_{\mathcal{U}}\varepsilon_{:,j} \leq t({\sigma_j^*})^2 + ({\sigma_j^*})^2x,~~~\varepsilon_{:,j}^T{P}_{\mathcal{V}}\varepsilon_{:,j} \geq - ({\sigma_j^*})^2x \right) \geq 1-4\exp(-c \min(x,x^2/t)).
\end{align*} 
Take $x=\xi t\log p$ for $\xi$ that is sufficiently large so we have 
$$\p \left(  \varepsilon_{:,j}^T ({P}_{{X}_{\CS}}-{P}_{{X}_{\CS_j^*}}) \varepsilon_{:,j}\leq  c_{t_3}t\log p \right) \geq 1-4\exp(-10t \log p).$$
Then, similar to Lemma \ref{suptechlem1}, we use union bound to complete the proof.
\end{proof}

\begin{lem}\label{suptechlem5}
Let us define the event
$$H=\bigcap_{j\in S(u^*)} \mathcal{E}_6(j).$$
One has
\begin{equation}
    \p\left(H\right)\geq 1-2sp^{-7}.
\end{equation}
\end{lem}
\begin{proof}
Note that for a given $S\subseteq[p]$ such that $|S|=2s$ and $j\in S$, 
\begin{equation}\label{suptechlem5ineq}
    \p \left(\sup_{\substack{v\in\R^p\\S(v)=S }}\left[\varepsilon_{:,j}^T \frac{\sum_{i\neq j}v_iX_{:,i}}{\|\sum_{i\neq j}v_iX_{:,i}\|_2}\right]^2 \lesssim s\log p \right) \geq 1 - \exp(-10s\log p)
\end{equation}
which is a consequence of~\eqref{techlemineq2}. As a result, 
\begin{align*}
    \p(H^c) & \leq \p\left(\max_{j\in S(u^*)}\max_{\substack{S\subseteq[p] \\ |S|=2s}}\sup_{\substack{v\in\R^p\\S(v)=S }}\left[\varepsilon_{:,j}^T \frac{\sum_{i\neq j}v_iX_{:,i}}{\|\sum_{i\neq j}v_iX_{:,i}\|_2}\right]^2\gtrsim s\log p\right) \\
    & \stackrel{(a)}{\leq}\sum_{j\in S(u^*)}\sum_{\substack{S\subseteq[p] \\ |S|=2s}} \p\left(\sup_{\substack{v\in\R^p\\S(v)=S }}\left[\varepsilon_{:,j}^T \frac{\sum_{i\neq j}v_iX_{:,i}}{\|\sum_{i\neq j}v_iX_{:,i}\|_2}\right]^2\gtrsim s\log p\right) \\
    & \stackrel{(b)}{\leq } \sum_{j\in S(u^*)}\sum_{\substack{S\subseteq[p] \\ |S|=2s}}\exp(-10s\log p) \\
    & \leq 2s p^{2s}p^{-10s}\leq 2sp^{-8s}\leq 2sp^{-7}
\end{align*}
where $(a)$ is due to union bound and $(b)$ is due to~\eqref{suptechlem5ineq}. This completes the proof.
\end{proof}

\subsection{Proof of  Theorem \ref{supportthm}}

\begin{proof}
In this proof, we assume events given in~\eqref{supevents} (for all $j,\CS$, as shown in Lemmas~\ref{suptechlem4}--\ref{suptechlem5}, above) hold true. Note that by Lemmas~\ref{suptechlem4}, \ref{suptechlem1},~\ref{suptechlem2},~\ref{suptechlem3} and~\ref{suptechlem5}, this happens with probability greater than
\begin{equation}\label{supthmprob}
    1-9sp^{-7}-4s\log(-10s\log p).
\end{equation}
Suppose $z\in\{0,1\}^p$ with $S=S(z)$. Let 
$$\text{$\CS_j=S\setminus\{j\}$,~~$\CS_j^*=S(u^*)\setminus\{j\}$~~and~~$\CS_j^0=\CS_j^*\setminus \CS_j$}.$$ 
Recalling the definition of $\CR_{\CS}(\cdot)$ from~\eqref{thm3-proof-cr}, we have the following (see calculations in (6.1) of \citep{fan2020best}):
\begin{align}
    n\CR_{\CS_j}(X_{:,j})  =& ({X}_{\CS_j}{\beta}^*_{\CS_j,j} + {X}_{\CS^0_j}{\beta}^*_{\CS^0_j,j}+ \varepsilon_{:,j})^T({I}_n-{P}_{{X}_{\CS_j}})({X}_{\CS_j}{\beta}^*_{\CS_j,j} + {X}_{\CS^0_j}{\beta}^*_{\CS^0_j,j}+ \varepsilon_{:,j}) \nonumber \\
    =&  ( {X}_{\CS^0_j}{\beta}^*_{\CS^0_j,j} +  \varepsilon_{:,j})^T({I}_n-{P}_{{X}_{\CS_j}})( {X}_{\CS^0_j}{\beta}^*_{\CS^0_j,j} + \varepsilon_{:,j}) \nonumber \\
    =& n{({\beta}^*_{\CS^0_j,j})}^T (\schur{\hat{{G}}}{\CS_j}{\CS_j^*}) {\beta}^*_{\CS^0_j,j} + 2\varepsilon_{:,j}^T ({I}_n-{P}_{{X}_{\CS_j}}){X}_{\CS^0_j}{\beta}^*_{\CS^0_j,j} + \varepsilon_{:,j}^T({I}_n-{P}_{{X}_{\CS_j}})\varepsilon_{:,j}\label{thmsup1}
\end{align}
and 
\begin{align}
    n\CR_{\CS^*_j}(X_{:,j})  = \varepsilon_{:,j}^T({I}_n-{P}_{{X}_{\CS^*_j}})\varepsilon_{:,j}.\label{thmsup2}
\end{align}
For $j\in[p]$ and $z\in\{0,1\}^p$, let us define:
\begin{equation}\label{defn-g_j=z-app1}
    g_j(z) = \min_{\beta_{:,j}} ~\left\Vert X_{:,j}-\sum_{i: i\neq j}\beta_{i,j}X_{:,i}\right\Vert_2^2~~\text{s.t.} ~~\beta_{i,j}(1-z_i)=\beta_{i,j}(1-z_j)=0~~\forall i\in[p],\beta_{j,j}=0.
\end{equation}
Note that in~\eqref{defn-g_j=z-app1}, if $z_j=0$, then $\beta_{:,j}=0$ and 
$g_j(z)=\|X_{:,j}\|_2^2$.
On the other hand, if $z_j=1$, we have $g_j(z)=n\CR_{\CS_j}(X_{:,j})$ as $\beta_{i,j}=0$ for $i\notin \CS_j$. With this definition in place, Problem~\eqref{reg_multi} can be equivalently written as
\begin{equation}\label{reg_multi_thm3}
    \min_{z} ~\sum_{j=1}^p g_j(z)~~~\text{s.t.}~~~z\in\{0,1\}^p, \sum_{i=1}^p z_i \leq s.
\end{equation}

To analyze Problem \eqref{reg_multi_thm3}, for every $j$, we compare the value of $g_j(z)$ for a feasible $z$ to a ``oracle'' candidate $g_j(z^*)$ where $z^*$ has the same support as $u^*$. Then, we show that 
$$\sum_{j=1}^p g_j(z^*)<\sum_{j=1}^p g_j(z) $$
unless $z=z^*$, which completes the proof. We let $S^*=S(u^*)$ and $S$ be a feasible support for Problem~\eqref{reg_multi_thm3}. We also recall that $\CS_j=S\setminus\{j\}$, $\CS_j^*=S^*\setminus\{j\}$ and $t=|S\setminus S(u^*)|$. 

\smallskip

We first present the proof under condition~\ref{betamincond} and then present the more general proof for~\ref{betamincond-extend}.

\subsubsection{Proof of Theorem~\ref{supportthm} under Condition~\ref{betamincond}}\label{proof-subsub-conds1}
We consider three cases for $j$.\\
{\bf Case 1} ($j\notin S^*,j\in S$): In this case, we have $g_j(z)=n\CR_{\CS_j}(X_{:,j})$ and $g_j(z^*)=\|X_{:,j}\|_2^2=\varepsilon_{:,j}^T\varepsilon_{:,j}$. As a result, we have
\begin{align}
  g_j(z)-g_j(z^*)=  n\CR_{\CS_j}(X_{:,j}) -\varepsilon_{:,j}^T\varepsilon_{:,j} \stackrel{(a)}{=} -\varepsilon_{:,j}^T P_{X_{\CS_j}}\varepsilon_{:,j} \stackrel{(b)}{\geq} -c_{t_2}s\log p \label{case1}
\end{align}
where $(a)$ is true because $\beta_{\CS_j^0,j}=0$ and by \eqref{thmsup1} and $(b)$ is due to the event $\mathcal{E}_4(\cdot,\cdot)$ in~\eqref{supevents}.\\
{\bf Case 2} ($j\in S^*,j\notin S$): In this case, we have $g_j(z)=\|X_{:,j}\|_2^2$ and $g_j(z^*)=n\CR_{\CS_j^*}(X_{:,j})$. As a result, we have the following:
\begin{align}
    g_j(z)-g_j(z^*)&=\|X_{:,j}\|_2^2 -  n\CR_{\CS^*_j}(X_{:,j})\nonumber \\
    &\stackrel{(a)}{= } \|X_{:,j}\|_2^2 - \varepsilon_{:,j}^T (I_n - P_{X_{\CS_j^*}})\varepsilon_{:,j} \nonumber\\
   & =  \|X_{\CS_j^*}\beta^*_{\CS_j^*,j}+\varepsilon_{:,j}\|_2^2 - \varepsilon_{:,j}^T (I_n - P_{X_{\CS_j^*}})\varepsilon_{:,j} \nonumber\\
   & = n{\beta^*_{\CS_j^*,j}}^T \hat{G}_{\CS_j^*,\CS_j^*}\beta^*_{\CS_j^*,j} + 2\varepsilon_{:,j}^T  X_{\CS_j^*}\beta^*_{\CS_j^*,j} + \varepsilon_{:,j}^T P_{X_{\CS_j^*}}\varepsilon_{:,j} \nonumber \\
   & = n{\beta^*_{\CS_j^*,j}}^T \hat{G}_{\CS_j^*,\CS_j^*}\beta^*_{\CS_j^*,j} + 2\frac{\varepsilon_{:,j}^T  X_{\CS_j^*}\beta^*_{\CS_j^*,j}}{\|X_{\CS_j^*}\beta^*_{\CS_j^*,j}\|_2}\|X_{\CS_j^*}\beta^*_{\CS_j^*,j}\|_2 + \varepsilon_{:,j}^T P_{X_{\CS_j^*}}\varepsilon_{:,j} \nonumber \\
   & \stackrel{(b)}{\geq} \frac{n}{2}{\beta^*_{\CS_j^*,j}}^T \hat{G}_{\CS_j^*,\CS_j^*}\beta^*_{\CS_j^*,j} - 2\left(\frac{\varepsilon_{:,j}^T  X_{\CS_j^*}\beta^*_{\CS_j^*,j}}{\|X_{\CS_j^*}\beta^*_{\CS_j^*,j}\|_2}\right)^2 + \varepsilon_{:,j}^T P_{X_{\CS_j^*}}\varepsilon_{:,j} \nonumber \\
   & \stackrel{(c)}{\geq } \frac{n}{2}{\beta^*_{\CS_j^*,j}}^T \hat{G}_{\CS_j^*,\CS_j^*}\beta^*_{\CS_j^*,j}  + \varepsilon_{:,j}^T P_{X_{\CS_j^*}}\varepsilon_{:,j} - c_{t_4}s\log p \nonumber \\
   & \stackrel{(d)}{\geq } \frac{n}{2}{\beta^*_{\CS_j^*,j}}^T \hat{G}_{\CS_j^*,\CS_j^*}\beta^*_{\CS_j^*,j}  - (c_{t_2} + c_{t_4})s\log p \nonumber \\
    & \stackrel{(e)}{\geq } 0.8\eta s \log p- (c_{t_2} + c_{t_4})s\log p\label{case2}
\end{align}
where $(a)$ is due to \eqref{thmsup2}, $(b)$ is due to the inequality $2ab\geq -2a^2-b^2/2$, $(c)$ is due to event $\mathcal{E}_6(\cdot,\cdot)$,
$(d)$ is because of event $\mathcal{E}_4(\cdot,\cdot)$ and $(e)$ is due to event $\mathcal{E}_1(\cdot,\cdot)$.\\
{\bf Case 3} ($j\in S^*,j\in S$): In this case, $ g_j(z)=n\CR_{\CS_j}(X_{:,j})$ and $g_j(z^*)=n\CR_{\CS^*_j}(X_{:,j})$. As a result, we have
\begin{align}
    & g_j(z)-g_j(z^*)\nonumber \\ = & n\CR_{\CS_j}(X_{:,j}) - n\CR_{\CS^*_j}(X_{:,j}) \nonumber \\
     \stackrel{(a)}{\geq} & n{({\beta}^*_{\CS^0_j,j})}^T (\schur{\hat{{G}}}{\CS_j}{\CS_j^*}) {\beta}^*_{\CS^0_j,j} + 2\varepsilon_{:,j}^T ({I}_n-{P}_{{X}_{\CS_j}}){X}_{\CS^0_j}{\beta}^*_{\CS^0_j,j} + \varepsilon_{:,j}^T({P}_{{X}_{\CS^*_j}}-{P}_{{X}_{\CS_j}})\varepsilon_{:,j} \nonumber\\
    \stackrel{(b)}{\geq} & n{({\beta}^*_{\CS^0_j,j})}^T (\schur{\hat{{G}}}{\CS_j}{\CS_j^*}) {\beta}^*_{\CS^0_j,j}  -c_{t_1} \sqrt{n{({\beta}^*_{\CS^0_j,j})}^T (\schur{\hat{G}}{\CS}{\CS_j^*}) {\beta}^*_{\CS^0_j,j}}\sqrt{{t\log p}} -c_{t_3}t\log p\nonumber \\
     \stackrel{(c)}{\geq} & n{({\beta}^*_{\CS^0_j,j})}^T (\schur{\hat{{G}}}{\CS_j}{\CS_j^*}) {\beta}^*_{\CS^0_j,j}(1-\frac{c_{t_1}}{\sqrt{0.8\eta}})-c_{t_3}t\log p \nonumber\\
     \stackrel{(d)}{\geq} & \left(0.8\eta(1-\frac{c_{t_1}}{\sqrt{0.8\eta}})-c_{t_3} \right) t\log p\stackrel{(e)}{>}0\label{case3}
\end{align}
where $(a)$ is due to \eqref{thmsup1}, \eqref{thmsup2}, $(b)$ is true because of events $\mathcal{E}_3$ and $\mathcal{E}_5$, $(c)$ is due to event $\mathcal{E}_2$, $(d)$ is due to event $\mathcal{E}_2$ and by taking $\eta>c_{t_1}^2/0.8$, and $(e)$ is true by taking $\eta$ large enough so 
$$\left(0.8\eta(1-\frac{c_{t_1}}{\sqrt{0.8\eta}})-c_{t_3} \right)>1.$$
Considering the sum of all terms appearing in \eqref{case1}, \eqref{case2} and \eqref{case3}, the difference between optimal cost of $S,S^*$ is at least
\begin{align}
  \sum_{j=1}^p\left\{g_j(z)-g_j(z^*)\right\} \geq & t \left( 0.8\eta s \log p- (c_{t_2} + c_{t_4})s\log p\right)-tc_{t_2}s\log p \nonumber \\ \geq & ts\log p\left(0.8\eta -(2c_{t_2}+c_{t_4})\right)>0 \label{sup-thm-final}
\end{align}
when $\eta>(2c_{t_2}+c_{t_4})/0.8$, because there are $t$ instances in cases 1 and 2 above. This shows $S$ cannot be optimal unless $S=S^*$.

\medskip 
\subsubsection{Proof of Theorem~\ref{supportthm} under Condition~\ref{betamincond-extend}} Similar to the case of Section~\ref{proof-subsub-conds1}, we consider three cases here. As we do not use the event $\mathcal{E}_2$ in either of the first two cases above,~\eqref{case1} and~\eqref{case2} hold. Let us study the third case with $j\in S^*,j\in S$. We have
\begin{align}
   &  g_j(z)-g_j(z^*) \nonumber \\ = &n\CR_{\CS_j}(X_{:,j}) - n\CR_{\CS^*_j}(X_{:,j}) \nonumber \\
     \stackrel{(a)}{\geq} & n{({\beta}^*_{\CS^0_j,j})}^T (\schur{\hat{{G}}}{\CS_j}{\CS_j^*}) {\beta}^*_{\CS^0_j,j} + 2\varepsilon_{:,j}^T ({I}_n-{P}_{{X}_{\CS_j}}){X}_{\CS^0_j}{\beta}^*_{\CS^0_j,j} + \varepsilon_{:,j}^T({P}_{{X}_{\CS^*_j}}-{P}_{{X}_{\CS_j}})\varepsilon_{:,j} \nonumber\\
     \stackrel{(b)}{\geq} & n{({\beta}^*_{\CS^0_j,j})}^T (\schur{\hat{{G}}}{\CS_j}{\CS_j^*}) {\beta}^*_{\CS^0_j,j}  -c_{t_1} \sqrt{n{({\beta}^*_{\CS^0_j,j})}^T (\schur{\hat{G}}{\CS}{\CS_j^*}) {\beta}^*_{\CS^0_j,j}}\sqrt{{t\log p}} -c_{t_3}t\log p\nonumber \\
     \stackrel{(c)}{\geq} & \frac{3n}{4}{({\beta}^*_{\CS^0_j,j})}^T (\schur{\hat{{G}}}{\CS_j}{\CS_j^*}) {\beta}^*_{\CS^0_j,j}-(c_{t_1}^2+c_{t_3})t\log p \nonumber\\
    \geq & -(c_{t_1}^2+c_{t_3})t\log p\label{case3-new}
\end{align}
where $(a)$ is due to \eqref{thmsup1}, \eqref{thmsup2}, $(b)$ is true because of events $\mathcal{E}_3$ and $\mathcal{E}_5$ and $(c)$ is by $ab\geq -a^2-b^2/4$. Summing~\eqref{case1},~\eqref{case2} and~\eqref{case3-new} and considering there are $s-t$ values of $j$ in case 3, we achieve
\begin{align}\label{sup-thm-final-new}
  \sum_{j=1}^p\left\{g_j(z)-g_j(z^*)\right\}&  \geq t \left( 0.8\eta s \log p- (c_{t_2} + c_{t_4})s\log p\right)-tc_{t_2}s\log p -(c_{t_1}^2+c_{t_3})t(s-t)\log p\nonumber \\
  & =(0.8\eta- (c_{t_1}^2 +2c_{t_2} +c_{t_3}+ c_{t_4}) )st\log p + (c_{t_1}^2+c_{t_3})t^2>0
\end{align}
where $\eta>(c_{t_1}^2 +2c_{t_2} +c_{t_3}+ c_{t_4})/0.8$ and $t\geq 1$, which completes the proof.
\end{proof}

\section{Proofs of other main results}
\subsection{Proof of Proposition \ref{suboptpro}}
\begin{proof}
In this proof, we use the following results. From~\eqref{proof-thm3-in2}, 
\begin{equation}\label{prop1ineq}
    \sum_{j=1}^p \left\Vert \sum_{i\neq j}(\hat{\beta}_{i,j}-\beta_{i,j}^*)X_{:,i} \right\Vert_2^2\gtrsim n\sum_{j=1}^p\sum_{i:i\neq j}(\hat{\beta}_{i,j}-\beta^*_{i,j})
\end{equation}
with probability greater than $1-2sp\exp(-20s\log (p-1))$ as $n\gtrsim s\log p$. From~\eqref{techlemineq5}, 
\begin{equation}\label{prop2ineq}
   \sum_{j=1}^p \left(\varepsilon_{:,j}^T\frac{\sum_{i\neq j}(\hat{\beta}_{i,j}-{\beta}^*_{i,j})X_{:,i}}{\|\sum_{i\neq j}(\hat{\beta}_{i,j}-{\beta}^*_{i,j})X_{:,i}\|_2}\right)^2 \lesssim s^2\log(p/s)
\end{equation}
with probability greater than 
$1-2s\exp(-10s\log(p/s)).$
Moreover, from~\eqref{thm4ineq104}, 
\begin{equation}\label{prop3ineq}
    \|\varepsilon_{:,j}\|_2^2\lesssim n
\end{equation}
with probability greater than $1-2pe^{-Cs\log(p/s)}$ as $n\gtrsim s\log p$. Therefore, the intersection of events appearing in~\eqref{prop1ineq}, \eqref{prop2ineq} and~\eqref{prop3ineq} holds true with probability at least
\begin{equation}\label{prop1prob}
    1-2pe^{-Cs\log(p/s)}-2sp\exp(-20s\log (p-1))-2s\exp(-10s\log(p/s)).
\end{equation}
Let $\tau = (\text{UB}-\text{LB})/\text{LB}$. By the notation introduced in the proposition, 
\begin{align}
    & \sum_{j=1}^p \left\Vert X_{:,j} -\sum_{i\neq j}{\beta}^*_{i,j}X_{:,i} \right\Vert_2^2 \geq \text{LB} \nonumber \\
    \Rightarrow & \sum_{j=1}^p \left\Vert X_{:,j} -\sum_{i\neq j}\hat{\beta}_{i,j}X_{:,i} \right\Vert_2^2 - \sum_{j=1}^p \left\Vert X_{:,j} -\sum_{i\neq j}{\beta}^*_{i,j}X_{:,i} \right\Vert_2^2 \leq \text{UB}-\text{LB}\nonumber \\
    \Rightarrow & \sum_{j=1}^p \left\Vert X_{:,j} -\sum_{i\neq j}\hat{\beta}_{i,j}X_{:,i} \right\Vert_2^2 \leq \left(1+ \frac{\text{UB}-\text{LB}}{\sum_{j=1}^p \left\Vert X_{:,j} -\sum_{i\neq j}{\beta}^*_{i,j}X_{:,i} \right\Vert_2^2}\right)\sum_{j=1}^p \left\Vert X_{:,j} -\sum_{i\neq j}{\beta}^*_{i,j}X_{:,i} \right\Vert_2^2\nonumber.
\end{align}
As a result,
\begin{align}
    & \sum_{j=1}^p \left\Vert X_{:,j} -\sum_{i\neq j}\hat{\beta}_{i,j}X_{:,i} \right\Vert_2^2 \leq  (1+\tau) \sum_{j=1}^p \left\Vert X_{:,j} -\sum_{i\neq j}{\beta}^*_{i,j}X_{:,i} \right\Vert_2^2 \nonumber \\
    \stackrel{(a)}{\Rightarrow} & \sum_{j=1}^p \left\Vert \sum_{i\neq j}(\hat{\beta}_{i,j}-\beta_{i,j}^*)X_{:,i} \right\Vert_2^2\leq 2\sum_{j=1}^p \varepsilon_{:,j}^T\left(\sum_{i\neq j}(\hat{\beta}_{i,j}-{\beta}^*_{i,j})X_{:,i}\right)+\tau \sum_{j=1}^p \left\Vert X_{:,j} -\sum_{i\neq j}{\beta}^*_{i,j}X_{:,i} \right\Vert_2^2 \nonumber \\
    \stackrel{(b)}{\Rightarrow}& \sum_{j=1}^p \left\Vert \sum_{i\neq j}(\hat{\beta}_{i,j}-\beta_{i,j}^*)X_{:,i} \right\Vert_2^2\leq 2\sum_{j=1}^p \varepsilon_{:,j}^T\left(\sum_{i\neq j}(\hat{\beta}_{i,j}-{\beta}^*_{i,j})X_{:,i}\right)+\tau \sum_{j=1}^p \left\Vert \varepsilon_{:,j} \right\Vert_2^2 \label{apppropineq1}
\end{align}
where $(a)$ is by the representation $X_{:,j}=\sum_{i\neq j}\beta_{i,j}^*X_{:,i}+\varepsilon_{:,j}$ by Lemma \ref{lem_reg_ext}, $(b)$ is by the definition of $\varepsilon_{:,j}$ in Lemma~\ref{lem_reg_ext}. Next, note that for $j\in[p]$,
\begin{align}
    2\varepsilon_{:,j}^T\left(\sum_{i\neq j}(\hat{\beta}_{i,j}-{\beta}^*_{i,j})X_{:,i}\right) &= 2\frac{ \varepsilon_{:,j}^T\left(\sum_{i\neq j}(\hat{\beta}_{i,j}-{\beta}^*_{i,j})X_{:,i}\right)}{\|\sum_{i\neq j}(\hat{\beta}_{i,j}-{\beta}^*_{i,j})X_{:,i}\|_2}\|\sum_{i\neq j}(\hat{\beta}_{i,j}-{\beta}^*_{i,j})X_{:,i}\|_2 \nonumber \\
    & \leq  2\left(\varepsilon_{:,j}^T\frac{\sum_{i\neq j}(\hat{\beta}_{i,j}-{\beta}^*_{i,j})X_{:,i}}{\|\sum_{i\neq j}(\hat{\beta}_{i,j}-{\beta}^*_{i,j})X_{:,i}\|_2}\right)^2 + \frac{1}{2}\|\sum_{i\neq j}(\hat{\beta}_{i,j}-{\beta}^*_{i,j})X_{:,i}\|_2^2\label{apppropineq2}
\end{align}
where the above inequality follows from the observation $2ab \leq 2a^2+b^2/2$. As a result, from~\eqref{apppropineq1} and~\eqref{apppropineq2}, we arrive at
\begin{align}\label{propproofineq}
\sum_{j=1}^p \left\Vert \sum_{i\neq j}(\hat{\beta}_{i,j}-\beta_{i,j}^*)X_{:,i} \right\Vert_2^2 \lesssim \sum_{j=1}^p\left(\varepsilon_{:,j}^T\frac{\sum_{i\neq j}(\hat{\beta}_{i,j}-{\beta}^*_{i,j})X_{:,i}}{\|\sum_{i\neq j}(\hat{\beta}_{i,j}-{\beta}^*_{i,j})X_{:,i}\|_2}\right)^2+\tau \sum_{j=1}^p \left\Vert \varepsilon_{:,j} \right\Vert_2^2.
\end{align}
From \eqref{propproofineq},~\eqref{prop1ineq},~\eqref{prop2ineq} and~\eqref{prop3ineq}, we have:
\begin{align*}
    \sum_{j=1}^p\sum_{i:i\neq j}(\hat{\beta}_{i,j}-\beta^*_{i,j}) & \lesssim \frac{1}{n} \sum_{j=1}^p \left\Vert \sum_{i\neq j}(\hat{\beta}_{i,j}-\beta_{i,j}^*)X_{:,i} \right\Vert_2^2 \\
    & \lesssim \frac{s^2\log(p/s)}{n} + \frac{\tau}{n}  \sum_{j=1}^p \left\Vert \varepsilon_{:,j}\right\Vert_2^2\\
    & \lesssim \frac{s^2\log(p/s)}{n} + p\tau
\end{align*}
with probability at least as large as the quantity appearing in~\eqref{prop1prob}. This completes the proof.
\end{proof}

\subsection{Proof of Proposition \ref{approx-supp}}
\begin{proof}
Building on the notation we used in the proof of Theorem~\ref{supportthm}, from~\eqref{sup-thm-final} we have
\begin{align}\label{approx-supp-helper}
      \sum_{j=1}^p\left\{g_j(\hat{z})-g_j(z^*)\right\} \geq  c_{0}ts\log p
\end{align}
with high probability for some absolute constant $c_0>0$, where
$$t = |\{j\in[p]:z_j^*=1,\hat{z}_j=0\}|$$
is the total number of mistakes in the support. Similar to the proof of Proposition~\ref{suboptpro}, let $\tau=(\text{UB}-\text{LB})/\text{LB}$. Then,
\begin{align*}
    &\sum_{j=1}^p g_j(z^*) \geq \text{LB} \\
    \Rightarrow & \sum_{j=1}^p g_j(\hat{z}) -\sum_{j=1}^p g_j(z^*) \leq \text{UB}-\text{LB}.
\end{align*}
As a result, from~\eqref{approx-supp-helper}, one has
$$t\lesssim \frac{\text{UB}-\text{LB}}{s\log p}.$$
\end{proof}

\subsection{Proof of Proposition \ref{subgradpers}}
\begin{proof}
{\bf Part 1)} The map appearing in the cost function of~\eqref{Freg2}, that is:
\begin{equation}\label{Gconvex}
    (z,\xi,\beta,W,q) \mapsto H(z,\xi,\beta,W,q):=\frac{1}{2}\sum_{j=1}^p \left\Vert \xi_{:,j} \right\Vert_2^2+\lambda\sum_{j=1}^p\sum_{i\neq j}q_{i,j}
\end{equation}
is jointly convex in $z,\xi,\beta,W,q$. Let the (3-dimensional) rotated second order cone be denoted by
$$\mathcal{Q}=\{x\in\R^3: x_1^2\leq x_2x_3, x_2,x_3\geq 0\}.$$
Note that ${\mathcal Q}$ is convex. 
Using the above notation, Problem~\eqref{Freg2} can be written as
\begin{align}\label{Freg2-append}
    F_1(z)=\min_{\beta,\xi,W,q} \quad &  \frac{1}{2}\sum_{j=1}^p \left\Vert \xi_{:,j} \right\Vert_2^2+\lambda\sum_{j=1}^p\sum_{i\neq j}q_{i,j} \\
     \text{s.t.} \quad & (\beta_{i,j},q_{i,j},z_j)\in\mathcal{Q},~~|\beta_{i,j}|\leq MW_{i,j},~ \forall i\neq j;~~\beta_{i,i}=0,~~W_{i,j}\leq z_i,~~W_{i,j}\leq z_j~~ \forall i,j\in[p]\nonumber\\
     \quad & \xi_{:,j} = X_{:,j}-X_{-j}\beta_{:,j}~~~\forall j\in[p].\nonumber
\end{align}
Let 
\begin{multline*}
    \mathbb{S} = \bigg\{ (z,\xi,\beta,W,q): (\beta_{i,j},q_{i,j},z_j)\in\mathcal{Q}~\forall i\neq j;~~|\beta_{i,j}|\leq MW_{i,j},\beta_{i,i}=0,~~W_{i,j}\leq z_i,~~W_{i,j}\leq z_j~~ \forall i,j\in[p]\\
      \xi_{:,j} = X_{:,j}-X_{-j}\beta_{:,j}~~~\forall j\in[p],~~z\in[0,1]^p\bigg\}\subseteq\R^{p}\times \R^{n\times p}\times \R^{p\times p}\times \R^{p\times p}\times \R^{p\times p}.
\end{multline*}
Note that as $\mathcal{Q}$ is a convex set, $\mathbb{S}$ defined above is also convex. Let $\1_C(\cdot)$ denote the characteristic function of a set $C$, 
$$\1_C(x) = \begin{cases} 0 &\text{if}~ x\in C \\
 \infty & \text{if}~ x\notin C. \end{cases} $$
 Note that if $C$ is a convex set, $\1_C(\cdot)$ is a convex function. With this notation in place, Problem~\eqref{Freg2-append} can be written as 
 \begin{equation}\label{Freg_chara}
     F_1(z)=\min_{\beta,\xi,W,q} ~~ H(z,\xi,\beta,W,q) +\1_{\mathbb{S}}(z,\xi,\beta,W,q)
 \end{equation}
 where $H$ is defined in~\eqref{Gconvex}. Based on our discussion above, the  function $(z,\xi,\beta,W,q) \mapsto H(z,\xi,\beta,W,q) +\1_{\mathbb{S}}(z,\xi,\beta,W,q)$ 
 is convex. As $F_1(z)$ is obtained after a marginal minimization of a jointly convex function over a convex set, the map $z \mapsto F_1(z)$ is convex on $z \in [0,1]^p$ \citep[Chapter 3]{boyd2004convex}.

\smallskip

\noindent {\bf Part 2)} We prove this part for $\lambda>0$. The proof for $\lambda=0$ is similar.

Let us define 
    \begin{align}\label{fregdual}
     \tilde{F}(z)=\max_{\Lambda,\Gamma,\delta\geq 0,\alpha,\upsilon} ~~ & \sum_{j=1}^p \bigg\{-\frac{1}{2}  \|\alpha_{:,j}\|_2^2+\alpha_{:,j}^TX_{:,j}-z_j\sum_{i:i\neq j}\frac{( \Lambda^{+}_{i,j}-\Lambda^{-}_{i,j}-(X_{-j}^T\alpha_{:,j})_i)^2}{4(\lambda-\delta_{i,j})}\bigg\}- \langle \Gamma^1,z1_p^T\rangle - \langle \Gamma^2, 1_p{z}^T\rangle\\
     \text{s.t.} ~~~~~ &    \Gamma^1+\Gamma^2-M\Lambda^+-M\Lambda^-=0,~~\delta_{i,j}\leq \lambda~\forall i,j\in[p] \nonumber \\
     &  (\Lambda^{+})_{i,i}-(\Lambda^{-})_{i,i}=-\upsilon_i~~\forall i\in[p]\nonumber.
\end{align}

~~~~{\bf Claim 1:} For $z\in(0,1]^p$, $F_1(z)=\tilde{F}(z)$.

~~~~{\bf Proof of Claim 1:}
 Note that as the objective function of \eqref{Freg2-append} is convex [see Part 1] and its non-affine constraints are strictly feasible, strong duality holds for \eqref{Freg2-append} by enhanced Slater's condition (see \citet{boyd2004convex}, Section 5.2.3). 
  Considering Lagrange multipliers $\delta,\zeta$, $\Lambda^{+},\Lambda^{-}$, $\Gamma^1,\Gamma^2\in\R^{p\times p}$, $\upsilon\in\R^p$  and $\alpha\in\R^{n\times p}$ for problem constraints, the Lagrangian for this problem
${\mathcal L}(\beta,\xi,q,\zeta,\delta,\Lambda^{+},\Lambda^{-},\Gamma^1,\Gamma^2,\upsilon,\alpha)$ or ${\mathcal L}$ (for short),
 can be written as
\begin{equation}\label{persconds1}
    \begin{aligned}
     \mathcal{L} = &  \sum_{j=1}^p\left\{ \frac{1}{2}\left\Vert\xi_{:,j} \right\Vert_2^2  + \lambda \sum_{i\neq j}q_{i,j}\right\}+   \langle\Lambda^{+},\beta - MW\rangle -\langle\Lambda^{-},\beta + MW\rangle\\
     & +\sum_{\substack{i,j \\ i\neq j}}\left\{\zeta_{i,j}(\beta_{i,j}^2-q_{i,j}z_j)-\delta_{i,j}q_{i,j}\right\}
     +   \langle\Gamma^1,W - z1_p^T\rangle +\langle\Gamma^2,W - 1_pz^T\rangle \\ & +\sum_{j=1}^p\alpha_{:,j}^T(X_{:,j} - X_{-j}\beta_{:,j}-\xi_{:,j}) +\sum_{i=1}^p \upsilon_i \beta_{i,i}  \\
     = & \sum_{j=1}^p \left\{\frac{1}{2}\|\xi_{:,j}\|_2^2 + \beta_{:,j}^T(-X_{-j}^T\alpha_{:,j}+(\Lambda^{+}-\Lambda^{-})_{:,j})+\sum_{i:i\neq j}\zeta_{i,j}\beta_{i,j}^2\right\}\\ &  -\langle W,\Gamma^1+\Gamma^2-M\Lambda^+-M\Lambda^-\rangle
      -  \langle \Gamma^1,z1_p^T\rangle -  \langle \Gamma^2, 1_p{z}^T\rangle \\&+\sum_{j=1}^p\alpha_{:,j}^T(X_{:,j} -\xi_{:,j})+\sum_{ \substack{i,j\\ i\neq j}}(\lambda-\zeta_{i,j}z_j-\delta_{i,j})q_{i,j}+\sum_{i=1}^p \upsilon_i \beta_{i,i}
\end{aligned}
\end{equation}
where $1_p\in\R^p$ is a vector of all ones. By setting the gradient of the Lagrangian w.r.t. $\beta$, $\xi,W$ and $q$ equal to zero, we get
    \begin{equation}\label{kkt}
        \begin{aligned}
         \xi_{:,j} & = \alpha_{:,j}~~\forall j\in[p], \\
         (\Lambda^{+})_{i,j}-(\Lambda^{-})_{i,j}&=(X_{-j}^T\alpha_{:,j})_i-2\zeta_{i,j}\beta_{i,j}~~\forall i,j\in[p], i\neq j \\
(\Lambda^{+})_{i,i}-(\Lambda^{-})_{i,i}&=-\upsilon_i~~\forall i\in[p]\\
          \Gamma^1+\Gamma^2&=M\Lambda^++M\Lambda^-,\\
          \delta_{i,j}&=\lambda-\zeta_{i,j}z_j~~\forall i,j\in[p],i\neq j.
        \end{aligned}
    \end{equation}
    By rearranging some terms above, we get 
    \begin{equation}\label{persconds2}
        \begin{aligned}
         \beta_{i,j} &=\frac{(X_{-j}^T\alpha_{:,j})_i- (\Lambda^{+})_{i,j}+(\Lambda^{-})_{i,j}}{2\zeta_{i,j}}, i \neq j\\
          \zeta_{i,j}& = \frac{\lambda-\delta_{i,j}}{z_j}, i\neq j.
        \end{aligned}
    \end{equation}
    Additionally, note that as $\zeta_{i,j},z_j\geq 0$, we have $\delta_{i,j}\leq \lambda$. By substituting~\eqref{persconds2} into~\eqref{persconds1}, we can obtain a dual of~\eqref{Freg2-append} which is given by~\eqref{fregdual}.
    This completes the proof of claim 1.

    ~~~~{\bf Claim 2:} For $z\in[0,1]^p$, $F_1(z)=\tilde{F}(z)$.

~~~~{\bf Proof of Claim 2:}
    Without loss of generality, assume all entries of $z_{1:k}$ are nonzero for some $k\geq 1$, while $z_{k+1:p}=0$. Define $\check{z}=z_{1:k}$, $\check{X}=X_{:,1:k}$ and $\check{\beta}=\beta_{1:k,1:k}$. We define a version of $F_1(z)$ restricted to $z_{1:k}$, as

    \begin{align}\label{Freg2-append-restrict}
   \check{F}_1(\check{z})=\min_{\check{\beta},\check\xi,\check{W},\check{q}} \quad &  \frac{1}{2}\sum_{j=1}^k \left\Vert \check{\xi}_{:,j} \right\Vert_2^2+\lambda\sum_{j=1}^k\sum_{i\neq j}\check{q}_{i,j} \\
     \text{s.t.} \quad & (\check{\beta}_{i,j},\check{q}_{i,j},\check{z}_j)\in\mathcal{Q},~~|\check{\beta}_{i,j}|\leq M\check{W}_{i,j},~ \forall i\neq j;~~\check{\beta}_{i,i}=0,~~\check{W}_{i,j}\leq \check{z}_i,~~\check{W}_{i,j}\leq \check{z}_j~~ \forall i,j\in[k]\nonumber\\
     \quad & \check{\xi}_{:,j} = \check{X}_{:,j}-\check{X}_{-j}\check{\beta}_{:,j}~~~\forall j\in[k].\nonumber
\end{align}
From claim 1, we have that 
    \begin{align}\label{fregdual-restrict}
     \check{F}_1(\check{z})=\max_{\check{\Lambda},\check{\Gamma},\check{\delta}\geq 0,\check{\alpha},\check{\upsilon}} ~~ & \sum_{j=1}^k \bigg\{-\frac{1}{2}  \|\check{\alpha}_{:,j}\|_2^2+\check{\alpha}_{:,j}^T\check{X}_{:,j}-\check{z}_j\sum_{i:i\neq j}\frac{( \check{\Lambda}^{+}_{i,j}-\check{\Lambda}^{-}_{i,j}-(\check{X}_{-j}^T\check{\alpha}_{:,j})_i)^2}{4(\lambda-\check{\delta}_{i,j})}\bigg\}- \langle \check{\Gamma}^1,\check{z}1_k^T\rangle - \langle \check{\Gamma}^2, 1_k\check{z}^T\rangle\\
     \text{s.t.} ~~ &    \check{\Gamma}^1+\check{\Gamma}^2-M\check{\Lambda}^+-M\check{\Lambda}^-=0,~~\check{\delta}_{i,j}\leq \lambda~\forall i,j\in[k] \nonumber\\
     & (\check{\Lambda}^{+})_{i,i}-(\check{\Lambda}^{-})_{i,i}=-\check{\upsilon}_i~~\forall i\in[k]\nonumber.
\end{align}    
Let us define the variables: $\bar{\check{\Lambda}},\bar{\check{\Gamma}},\bar{\check{\delta}},\bar{\check{\alpha}},\bar{\check{\upsilon}}$ be the optimal solution to~\eqref{fregdual-restrict}. Let
$$\bar{\alpha}_{:,j} = \begin{cases}
   \bar{\check{\alpha}}_{:,j} & j\leq k \\
  X_{:,j} & j>k
\end{cases},~~~~~~~\bar{\delta}_{i,j} = \begin{cases}
   \bar{\check{\delta}}_{i,j} & i,j\leq k \\
  0 & \text{otherwise},
\end{cases}$$
$$\bar{\Lambda}^+_{i,j} = \begin{cases}
    \bar{\check{\Lambda}}^+_{i,j} & i,j\leq k \\
  (X_{-j}^T\bar{\alpha}_{:,j})_i & i>k\lor j>k, (X_{-j}^T\bar{\alpha}_{:,j})_i\geq 0 \\
  0 & \text{otherwise}
\end{cases},$$
$$\bar{\Lambda}^-_{i,j} = \begin{cases}
    \bar{\check{\Lambda}}^-_{i,j} & i,j\leq k \\
  (X_{-j}^T\bar{\alpha}_{:,j})_i & i>k\lor j>k, (X_{-j}^T\bar{\alpha}_{:,j})_i< 0 \\
  0 & \text{otherwise},
\end{cases}$$

$$\bar{\Gamma}^2_{i,j} = \begin{cases}
    \bar{\check{\Gamma}}^2_{i,j} & i,j\leq k \\
   0& i>k, j\leq k \\
  M(\bar{\Lambda}^+_{i,j} + \bar{\Lambda}^-_{i,j}) & \text{otherwise}
\end{cases},~~~~~~\bar{\Gamma}^1_{i,j} = \begin{cases}
    \bar{\check{\Gamma}}^1_{i,j} & i,j\leq k \\
  M(\bar{\Lambda}^+_{i,j} + \bar{\Lambda}^-_{i,j}) & i>k, j\leq k\\
  0 & j>k,
\end{cases}$$
$$\bar{\upsilon}_{j} = \begin{cases}
   \bar{\check{\upsilon}}_{j} & j\leq k \\
  \bar{\Lambda}_{i,i}^- - \bar{\Lambda}_{i,i}^+ & j>k
\end{cases}.$$
Similar to the case of Claim~1, one can see that  $\bar{{\Lambda}},\bar{{\Gamma}},\bar{{\delta}},\bar{{\alpha}},\bar{{\upsilon}}$
are optimal for~\eqref{fregdual} upon inspection, with the objective of~\eqref{fregdual} given as
\begin{align}
   & \sum_{j=1}^p \bigg\{-\frac{1}{2}  \|\bar{\alpha}_{:,j}\|_2^2+\bar{\alpha}_{:,j}^TX_{:,j}-z_j\sum_{i\in[p]:i\neq j}\frac{( \bar{\Lambda}^{+}_{i,j}-\bar{\Lambda}^{-}_{i,j}-(X_{-j}^T\bar{\alpha}_{:,j})_i)^2}{4(\lambda-\bar{\delta}_{i,j})}\bigg\}- \langle \bar{\Gamma}^1,z1_p^T\rangle - \langle \bar{\Gamma}^2, 1_p{z}^T\rangle 
     = F_1(z)
\end{align}
which completes the proof of the claim.

For the rest of proof, we assume $z\in\{0,1\}^p$ as we focus on the subgradient for feasible binary solutions. To this end, we make use of the following result.

\smallskip

~~~~{\bf Claim 3:} We claim at optimality of the primal/dual pair~\eqref{Freg2-append}/\eqref{fregdual}, for $j\in[p]$
\begin{equation}\label{claim}
    \begin{aligned}
      \Lambda^{+}_{i,j}-\Lambda^{-}_{i,j}&=(X_{-j}^T\alpha_{:,j})_i-2\lambda\bar{\beta}_{i,j}, i\neq j\\
       \Lambda^{+}_{j,j}-\Lambda^{-}_{j,j} +\upsilon_j &= \bar{\beta}_{j,j} = 0,\\
      z_j\sum_{i:i\neq j}\frac{\left( \Lambda^{+}_{i,j}-\Lambda^{-}_{i,j}-(X_{-j}^T\alpha_{:,j})_i\right)^2}{4(\lambda-\delta_{i,j})}&=\frac{\lambda}{z_j}\sum_{i=1}^p \bar{\beta}_{i,j}^2
    \end{aligned}
\end{equation}
where $\bar{\beta}_{:,j}$ is given by \eqref{linregperstosolve}. \\

~~~~{\bf Proof of Claim 3:} First, note that the second part of the claim is a result of feasibility of~\eqref{fregdual}. To prove the rest of the claim, we consider the following cases:
\begin{enumerate}
    \item $z_j=0$: In this case, from~\eqref{persconds2} we have $(\Lambda^{+})_{i,j}-(\Lambda^{-})_{i,j}-(X_{-j}^T\alpha_{:,j})_i=-2\zeta_{i,j}\beta_{i,j}=0$ so 
    $$z_j\sum_{i:i\neq j}\frac{\left( \Lambda^{+}_{i,j}-\Lambda^{-}_{i,j}-(X_{-j}^T\alpha_{:,j})_i\right)^2}{4(\lambda-\delta_{i,j})}=0=\frac{\lambda}{z_j} \sum_{i}\bar{\beta}_{i,j}^2.$$
    \item $z_j>0,\beta_{i,j}=0$: In this case, similar to the above case (1), $(\Lambda^{+})_{i,j}-(\Lambda^{-})_{i,j}-(X_{-j}^T\alpha_{:,j})_i=-2\zeta_{i,j}\beta_{i,j}=0$ so
    $$z_j\frac{\left( \Lambda^{+}_{i,j}-\Lambda^{-}_{i,j}-(X_{-j}^T\alpha_{:,j})_i\right)^2}{4(\lambda-\delta_{i,j})}=0=\frac{\lambda}{z_j} \bar{\beta}_{i,j}^2.$$
    \item $z_j>0,\beta_{i,j}\neq 0$: In this case, from complementary slackness we have $\delta_{i,j}q_{i,j}=0$ and as $q_{i,j}=\beta_{i,j}^2>0$, we have $\delta_{i,j}=0$. This implies $\zeta_{i,j}=\lambda/z_j$. Therefore, from~\eqref{kkt}, 
    $$z_j\frac{\left( \Lambda^{+}_{i,j}-\Lambda^{-}_{i,j}-(X_{-j}^T\alpha_{:,j})_i\right)^2}{4(\lambda-\delta_{i,j})}= z_j\frac{4\zeta_{i,j}^2\bar{\beta}_{i,j}^2}{4z_j^2\lambda}=\frac{\lambda}{z_j}\bar{\beta}_{i,j}^2.$$
\end{enumerate}
This completes the proof of claim~\eqref{claim}.\\

In~\eqref{fregdual}, the cost function is larger if the value of $\Gamma^1,\Gamma^2$ is smaller. Therefore, at optimality, we have that
\begin{equation}
    \Lambda^+_{i,j}=\begin{cases}
     (X_{-j}^T\alpha_{:,j})_{i}-2\lambda\bar{\beta}_{i,j} &\mbox{ if }  (X_{-j}^T\alpha_{:,j})_{i}-2\lambda\bar{\beta}_{i,j}\geq 0\\
      0 &\mbox{ if } (X_{-j}^T\alpha_{:,j})_{i}-2\lambda\bar{\beta}_{i,j}< 0,
    \end{cases}
    \end{equation}
    \begin{equation}
    \Lambda^-_{i,j}=\begin{cases}
     0 &\mbox{ if }  (X_{-j}^T\alpha_{:,j})_{i}-2\lambda\bar{\beta}_{i,j}\geq 0\\
      -(X_{-j}^T\alpha_{:,j})_{i}+2\lambda\bar{\beta}_{i,j} &\mbox{ if } (X_{-j}^T\alpha_{:,j})_{i}-2\lambda\bar{\beta}_{i,j}< 0,
    \end{cases}
\end{equation}
(which implies $\Lambda^+_{i,i}=\Lambda^-_{i,i}=0$). This is true as by~\eqref{claim}, $\Lambda^{+}_{i,j}-\Lambda^{-}_{i,j}=(X_{-j}^T\alpha_{:,j})_i-2\lambda\bar{\beta}_{i,j}$ for $i\neq j$.
By this choice of $\Lambda^+,\Lambda^-$, the optimal value of $\Gamma^1,\Gamma^2$ are given by \eqref{regpersm1m2}. This true as for $i\neq j$, if $z_i=z_j=1$, the contribution of $(\Gamma^1)_{i,j},(\Gamma^2)_{i,j}$ in the dual cost is $(\Gamma^1)_{i,j}+(\Gamma^2)_{i,j}=M|(X_{-j}^T\alpha_{:,j})_{i}-2\lambda\bar{\beta}_{i,j}|$. If $z_i=z_j=0$, the contribution of this term is zero. In other cases, by taking $\Gamma^1,\Gamma^2$ such as in \eqref{regpersm1m2}, the contribution of $(\Gamma^1)_{i,j},(\Gamma^2)_{i,j}$ is zero which is dual optimal. For $i=j$, the choice of $\Lambda^+_{i,i}=\Lambda^-_{i,i}=0$ leads to $(\Gamma^1_{i,i}+\Gamma^2_{i,i})z_i=0$ which is dual optimal. In addition, by complementary slackness, at optimality, $\alpha_{:,j}=\xi_{:,j}=X_{:,j}-X_{-j}\bar{\beta}_{:,j}$. In \eqref{fregdual}, $F_1(z)$ is written as the maximum of linear functions of $z$, therefore, the gradient of the dual cost function w.r.t $z$ at an optimal dual solution, is a subgradient of $F_1(z)$ by Danskin's Theorem \citep{bertsekas1997nonlinear}. Finally, the gradient of cost in~\eqref{fregdual} w.r.t. $z_i$ is $-\sum_j \Gamma^1_{i,j}$ and the gradient w.r.t. $z_j$ is 
$$-\sum_{i=1}^p \Gamma^2_{i,j}-\sum_{i:i\neq j}\frac{\left( \Lambda^{+}_{i,j}-\Lambda^{-}_{i,j}-(X_{-j}^T\alpha_{:,j})_i\right)^2}{4(\lambda-\delta_{i,j})}-\frac{(-\upsilon_j- (\Lambda^{+})_{j,j}+(\Lambda^{-})_{j,j})^2}{4(\lambda-\delta_{j,j})} = -\sum_{i=1}^p\Gamma^2_{i,j}-\lambda\sum_{i=1}^p\bar{\beta}_{i,j}^2$$
by claim~\eqref{claim}.

\end{proof}

\section{Numerical Experiments with Non-Gaussian Data}\label{supp:nongaussian}
Our proposed estimator was designed based on an underlying multivariate Gaussian distribution assumption. To understand how our estimator works when the data distribution deviates from a multivariate Gaussian distribution, we conduct some numerical experiments.

Let $G^*=I_p+(u^*)(u^*)^T+0.05u_1u_1^T$ be a spiked covariance matrix with multiple spikes (below we discuss our choices of $u^*$ and $u_1$). We denote $H^*=(G^*)^{1/2}$ to be the matrix square root of $G^*$. 
We draw $\tilde{X}\in\R^{n\times p}$ with independent coordinates (aka entries) from a non-Gaussian distribution $\mathcal{D}$ as we discuss below. This distribution has zero mean and unit variance. We then obtain the data matrix as $X=\tilde{X}H^*\in\R^{n\times p}$. Note that the rows of $X$ are independent random variables, as $\tilde{X}$ has independent coordinates. Moreover, we have that
$$\E[X]=0~~~~~\text{and}~~~~~\frac{1}{n}\E[X^TX]=\frac{1}{n}H^*\E[\tilde{X}^T\tilde X]H^*=(H^*)^2=G^*$$
as the distribution $\mathcal{D}$ has zero mean and unit variance. Therefore, the data $X$ has zero mean, and its covariance follows a spiked covariance model. However, the rows of $X$ do not follow a multivariate Gaussian distribution. We study how our estimator and other estimators perform under this non-Gaussian spiked covariance model.

For the rest of this section, we will assume $u^*_1,\cdots,u^*_s=1/\sqrt{s}$, while $u^*_{s+1},\cdots,u^*_p=0$. We set the first $s$ coordinates of $u_1$ to be zero, and draw the rest of the coordinates independently from $\text{Unif}[0,1]$ and normalize $u_1$ to have unit norm. In what follows, we present our results for two data generation scenarios. Throughout this section, we set $s=5,p=500$ and vary $n\in\{300,500\}$.

\subsection{Truncated Gaussian Distribution}\label{supp:truncated}

First, we study a case where the non-Gaussian distribution $\mathcal{D}$ follows a (normalized) truncated Gaussian distribution. Let us denote:
$$\varphi(x)=\frac{1}{\sqrt{2\pi}}\exp\left(-\frac{x^2}{2}\right),~~~~~~~~~\Phi(x)=\int_{-\infty}^x \varphi(x)dx$$
to be the probability density function (pdf), and cumulative distribution function (cdf) of a standard one dimensional Gaussian distribution, respectively. For any given $b>0$, we say a distribution follows a truncated Gaussian distribution with threshold $b$, if its pdf is given by
\begin{equation}\label{truncated}
    f(x;b)=  \begin{cases}
       \frac{\varphi(x)}{\Phi(b)-\Phi(-b)} & \text{~if~} x\in[-b,b]\\
       0 & \text{~otherwise}. \\
     \end{cases}
\end{equation}
We note that in general, a truncated Gaussian distribution does not have unit variance. To this end, in our experiments we independently draw coordinates of $\tilde{X}$ from a truncated Gaussian, and then normalize $\tilde{X}$ by the standard deviation of the truncated Gaussian. Note that under this transformation, coordinates of $\tilde{X}$ need not be limited to $[-b,b]$. We observe that by increasing $b$, the function in~\eqref{truncated} more closely resembles a standard Gaussian pdf. Therefore, by varying $b$, we quantify how each estimator performs as we deviate further from a Gaussian spiked covariance model. To demonstrate this, we plot the pdf of one coordinate of $\tilde{X}$ for varying values of $b$ in Figure~\ref{fig:truncated_b}.

\begin{figure*}[t!]
     \centering
     \begin{tabular}{cc}
&\includegraphics[width=0.44\textwidth,trim=4.2cm 8.5cm 4.2cm 8.5cm, clip]{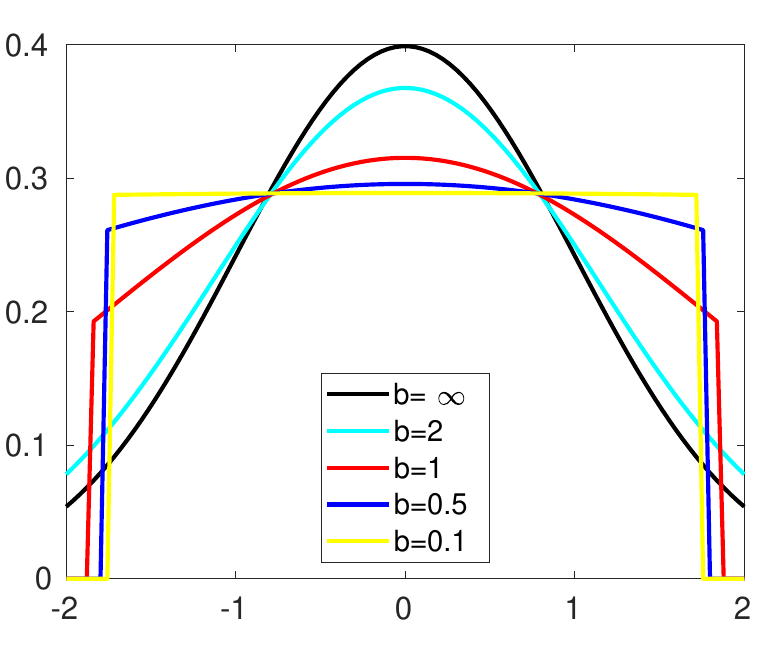}\\
&$x$
\end{tabular}
        \caption{\small The probability density function of coordinates of $\tilde{X}$ for varying values of $b$ in Appendix~\ref{supp:truncated}. }
        \label{fig:truncated_b}
\end{figure*}

We present the estimation error results for this setup in Figure~\ref{fig:synt-non-1} for two values of $n$ and various values of $b$. We see that for larger values of $b$ our method seems to have the lowest estimation error. However, by reducing $b$ (which results in a distribution further from Gaussian), most methods' accuracy reduces, including ours. This suggests that as long as the distribution does not deviate much from the Gaussian, our proposal can have good accuracy.

\begin{figure*}[t!]
     \centering
     \begin{tabular}{lcc}
&  $n=300$ & $n=600$ \\
     \rotatebox{90}{~~~~~~~~~~~~~~~~~~~~~~~~~~~~~~~~~$|\sin\angle(\hat{u},u^*)|$}& \includegraphics[width=0.44\textwidth,trim=.5cm 0cm 0cm 0cm, clip]{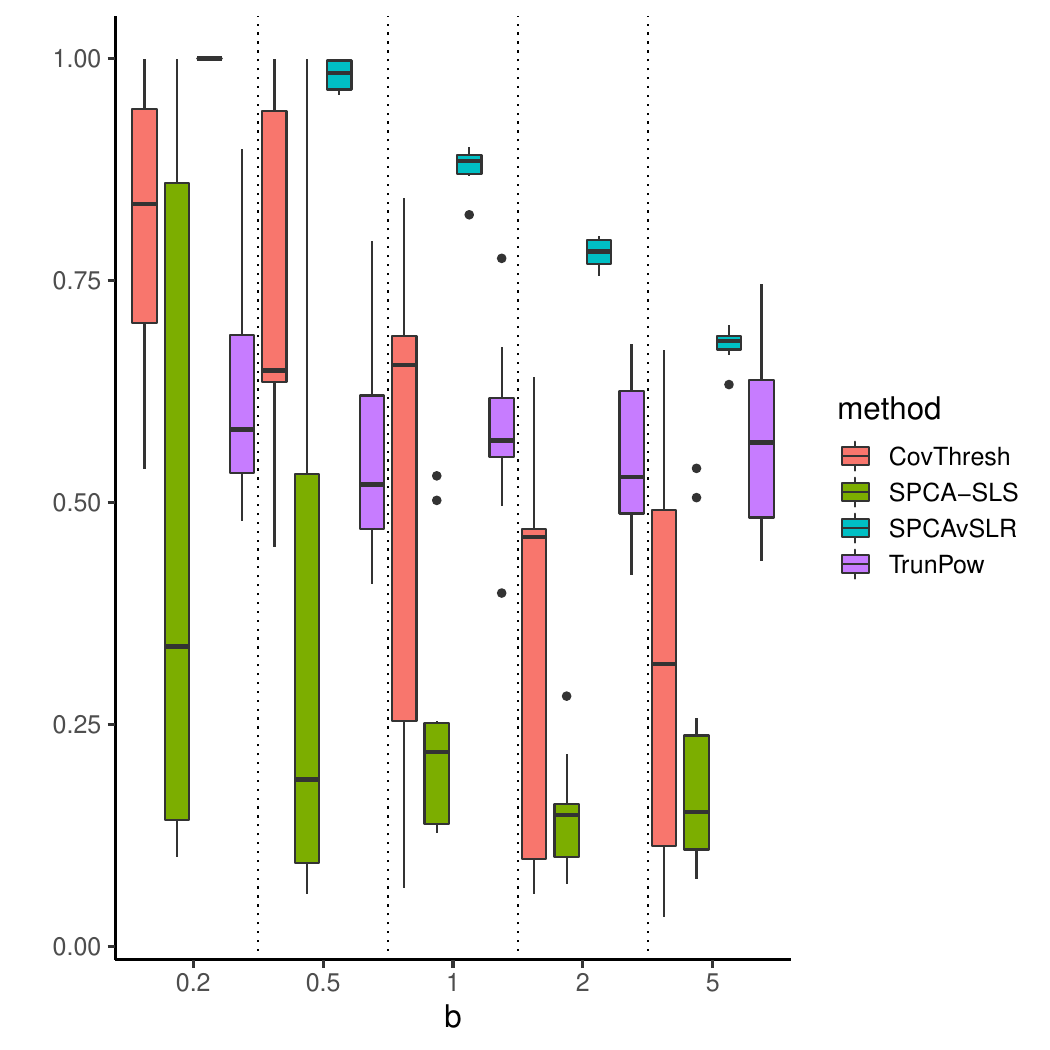}& 
 \includegraphics[width=0.44\textwidth,trim=.5cm 0cm 0cm 0cm, clip]{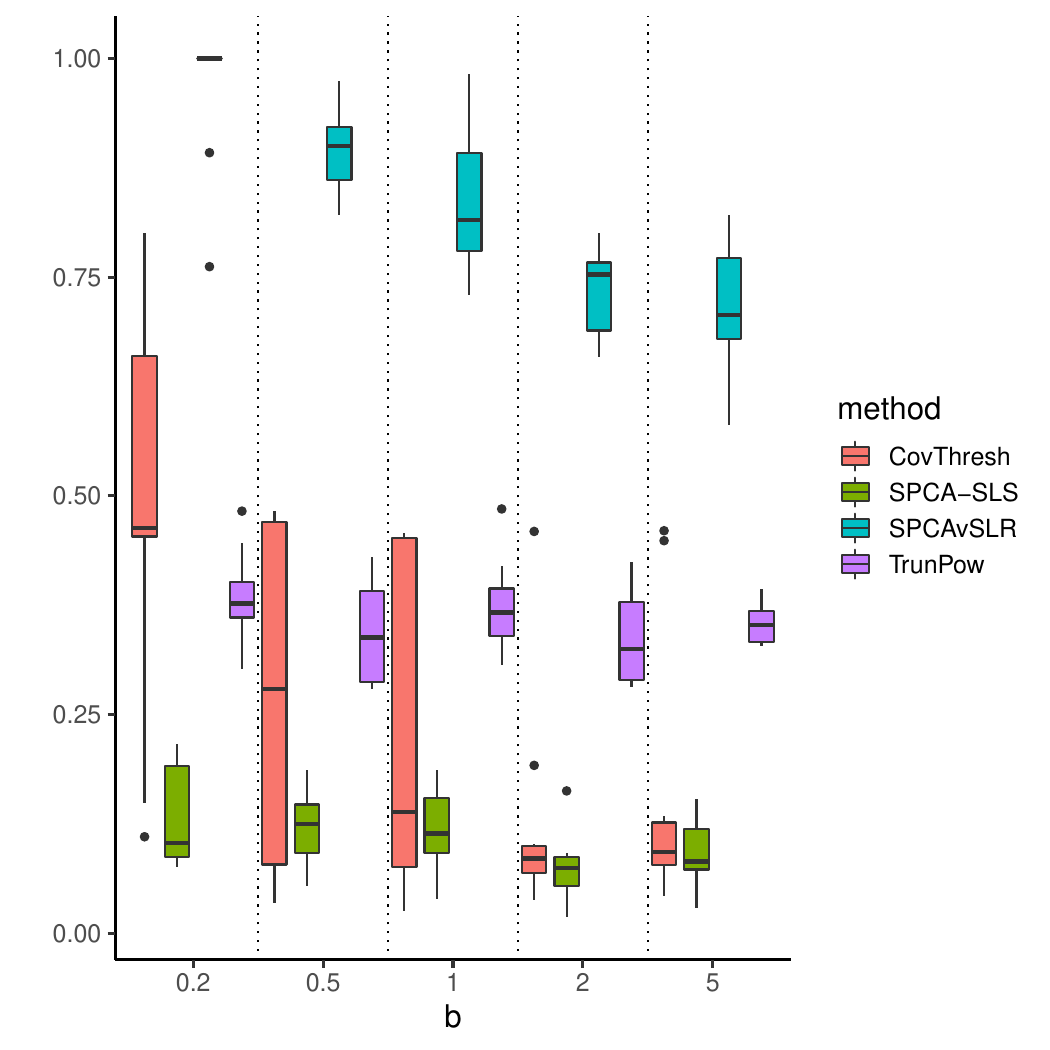}\\ 
\end{tabular}
        \caption{\small Experimental results with the truncated Gaussian distribution in Appendix~\ref{supp:truncated}.}
        \label{fig:synt-non-1}
\end{figure*}

\subsection{Binomial Distribution}\label{supp:binom}
Here we study a discrete distribution. In particular, for any given $N\geq 1$, we draw each coordinate of $\tilde X$ independently from $\text{binom}(N,0.5)$ where $\text{binom}(N,p)$ denotes a binomial distribution with $N$ trials and probability of success $p$. We then standardize $\tilde{X}$ using the mean and the standard deviation of $\text{binom}(N,0.5)$, such that each coordinate of $\tilde X$ has zero mean and unit variance. We note that as we increase $N$, by the central limit theorem, $\text{binom}(N,0.5)$ can be approximated by a Gaussian distribution. Therefore, similar to the previous case, we quantify how each estimator performs as we deviate further from a Gaussian spiked covariance model. The estimation error results for this setup are reported in Figure~\ref{fig:synt-non-2}. Similar to the case of the truncated Gaussian distribution, we see that as the underlying generative distribution deviates from Gaussian (i.e., smaller values of $N$), most methods' accuracy reduces, though for large $N$ we observe our method has the smallest error.

\begin{figure*}[t!]
     \centering
     \begin{tabular}{lcc}
&  $n=300$ & $n=600$ \\
     \rotatebox{90}{~~~~~~~~~~~~~~~~~~~~~~~~~~~~~~~~~$|\sin\angle(\hat{u},u^*)|$}& \includegraphics[width=0.44\textwidth,trim=.5cm 0cm 0cm 0cm, clip]{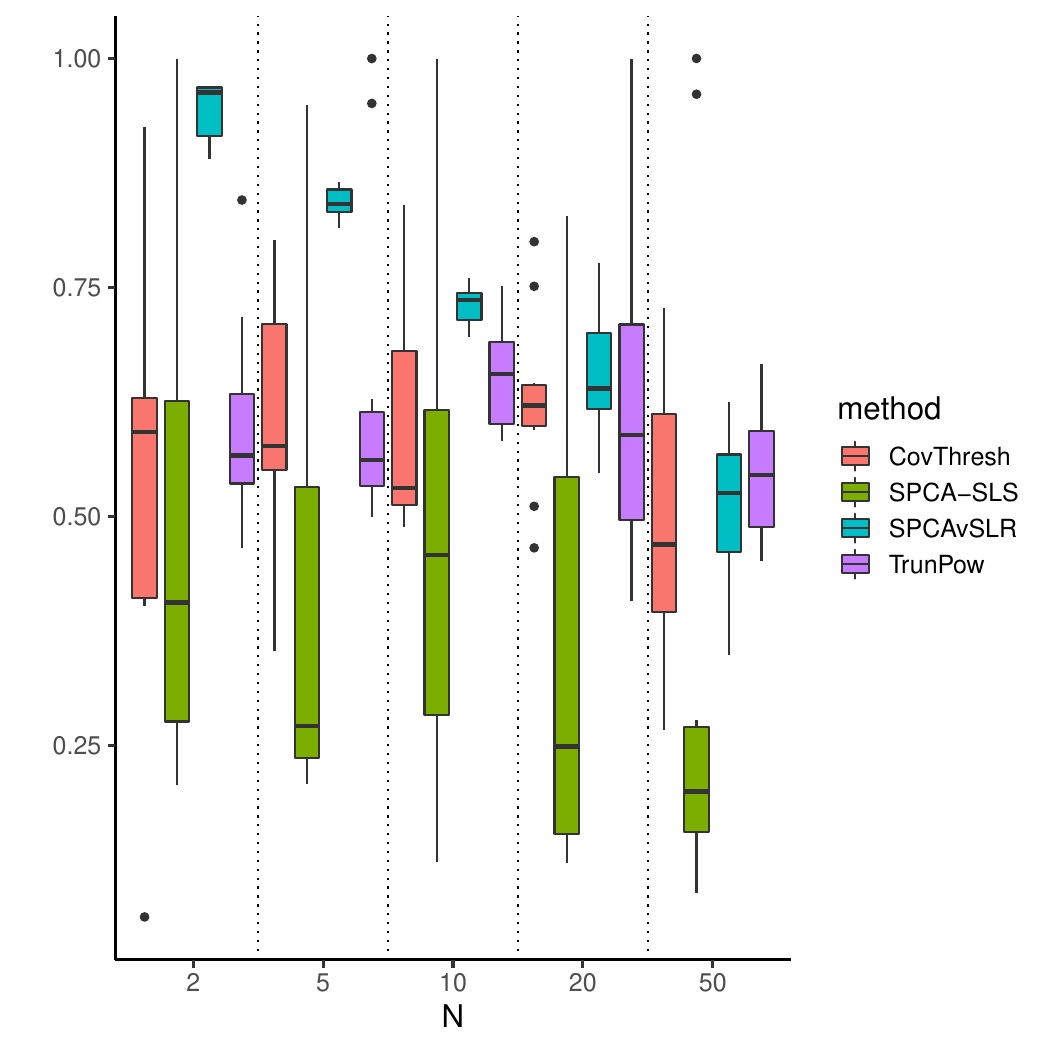}& 
 \includegraphics[width=0.44\textwidth,trim=.5cm 0cm 0cm 0cm, clip]{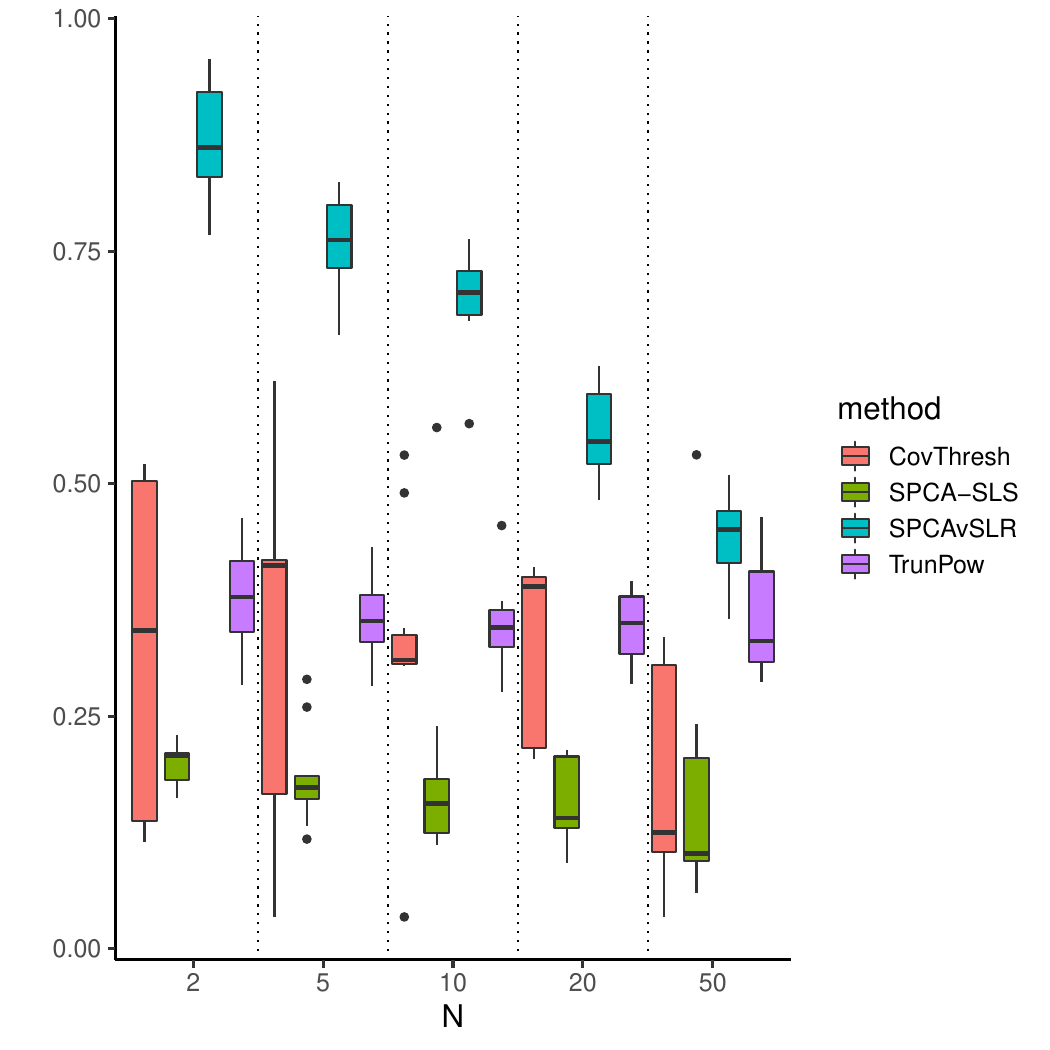}\\ 
\end{tabular}
        \caption{\small  Experimental results with the binomial distribution in Appendix~\ref{supp:binom}.}
        \label{fig:synt-non-2}
\end{figure*}

\end{document}